\pdfoutput=1
\documentclass[12pt, twoside=false]{article}
\usepackage[utf8]{inputenc}
\usepackage{amssymb,amsmath,amsthm}
\PassOptionsToPackage{hyphens}{url}\usepackage[pagebackref=FALSE,pdfstartview=FitH,colorlinks,citecolor=black,filecolor=black,linkcolor=black,urlcolor=black,breaklinks]{hyperref}
\usepackage{longtable}
\usepackage[left=1in, right=1in, top=1in, bottom=1in]{geometry}
\usepackage{setspace}
\usepackage[titletoc,title]{appendix}
\usepackage{graphicx}
\usepackage{xcolor}
\usepackage{cleveref}

\usepackage{subcaption}
\usepackage{float}
\usepackage{natbib}
\usepackage{enumitem}
\usepackage{gensymb}
\usepackage{thmtools,thm-restate}
\usepackage{multirow}
\usepackage{booktabs}
\usepackage{pgfplots}
\pgfplotsset{compat=1.12}
\usepackage{tikz}
\newcommand{\expect}{\mathbb{E}}

\newcommand{\one}{\boldsymbol{1}}

\usepackage{newunicodechar}
\newunicodechar{ }{ }

\usetikzlibrary{calc}
\usepackage{tikz-3dplot}
\usetikzlibrary{quotes,angles}
\usetikzlibrary{shapes.symbols}
\makeatletter
\pgfdeclareshape{document}{
	\inheritsavedanchors[from=rectangle]
	\inheritanchorborder[from=rectangle]
	\inheritanchor[from=rectangle]{center}
	\inheritanchor[from=rectangle]{north}
	\inheritanchor[from=rectangle]{south}
	\inheritanchor[from=rectangle]{west}
	\inheritanchor[from=rectangle]{east}
	\backgroundpath{
		\southwest \pgf@xa=\pgf@x \pgf@ya=\pgf@y
		\northeast \pgf@xb=\pgf@x \pgf@yb=\pgf@y
		\pgf@xc=\pgf@xb \advance\pgf@xc by-10pt
		\pgf@yc=\pgf@yb \advance\pgf@yc by-10pt
		\pgfpathmoveto{\pgfpoint{\pgf@xa}{\pgf@ya}}
		\pgfpathlineto{\pgfpoint{\pgf@xa}{\pgf@yb}}
		\pgfpathlineto{\pgfpoint{\pgf@xc}{\pgf@yb}}
		\pgfpathlineto{\pgfpoint{\pgf@xb}{\pgf@yc}}
		\pgfpathlineto{\pgfpoint{\pgf@xb}{\pgf@ya}}
		\pgfpathclose
		\pgfpathmoveto{\pgfpoint{\pgf@xc}{\pgf@yb}}
		\pgfpathlineto{\pgfpoint{\pgf@xc}{\pgf@yc}}
		\pgfpathlineto{\pgfpoint{\pgf@xb}{\pgf@yc}}
		\pgfpathlineto{\pgfpoint{\pgf@xc}{\pgf@yc}}
	}
}
\makeatother

\declaretheorem[name=Proposition]{proposition}

\declaretheorem[name=Corollary]{corollary}

\crefname{condition}{condition}{conditions}
\Crefname{condition}{Condition}{Conditions}
\crefname{assumption}{assumption}{assumptions}
\Crefname{assumption}{Assumption}{Assumptions}
\crefname{figure}{figure}{figures}
\crefname{equation}{equation}{equations}
\newcommand{\uR}{\underline{R}}
\newcommand{\oR}{\overline{R}}

\newcommand{\selftitle}[1]{}

\usepackage{tocloft}
\newcommand{\listappendixsectionsname}{\bfseries\large Contents of Online Appendix}
\newlistof{appendixsection}{appsec}{\listappendixsectionsname}

\newcommand{\listappendixtablesname}{\bfseries\large List of Online Appendix Tables}
\newlistof[section]{appendixtable}{apptab}{\listappendixtablesname}
\newcommand{\apptable}[1]{
  \addcontentsline{apptab}{appendixtable}{\protect\numberline{\thetable}#1}
}

\newcommand{\listappendixfiguresname}{\bfseries\large List of Online Appendix Figures}
\newlistof[section]{appendixfigure}{appfig}{\listappendixfiguresname}
\newcommand{\appfigure}[1]{
  \addcontentsline{appfig}{appendixfigure}{\protect\numberline{\thefigure}#1}
}

\title{Influencer Cartels\thanks{
We are grateful to the editor Pierre Dubois and anonymous referees for guidance in revising the paper. We thank the data editor Maia Güell and the anonymous replicator for verifying the replication package.
We also thank
Daniel Ershov, Carl-Christian Groh, John Horton, Dina Mayzlin,  Matt Mitchell, Devesh Raval, Johanna Rickne,  Stephen P. Ryan, Stephan Seiler, Egor Starkov,
and seminar participants at 
University of East Anglia, 
University of Leicester,
University of Nottingham, 
Collegio Carlo Alberto,
Mannheim Virtual IO Seminar,
UK Competition and Markets Authority (CMA),
12th Paris Conference on Digital Economics,
IIOC, 
Munich Summer Institute,
Baltic Economics Conference,
Australasian Meeting of the Econometric Society,
NBER Summer Institute on IT and Digitization,
Econometric Society European Meeting, 
EARIE,
NIE Annual Conference,
13th Workshop on the Economics of Advertising and Marketing,
Bristol-Warwick Empirical IO Workshop, 
Berlin IO Day, 
6th Monash-Warwick-Zurich Text-as-Data Workshop,
CESifo Economics of Digitization Conference,
Text-as-Data in Economics Workshop in Liverpool,
City/CMA Workshop on Economics of Competition and Regulation,
Theory and Policy in the Digital Economy Workshop in Durham,
Workshop on Collusion in Markets in Tinbergen Institute Amsterdam,
SAET,
CEPR Virtual IO Seminar,
and
CEPR Paris Symposium
for helpful comments and suggestions.
Refine.ink was used to proofread the paper for consistency and clarity.
First draft: February 12, 2021.
}
}
\author{
Marit Hinnosaar\thanks{
University of Nottingham and CEPR, \url{marit.hinnosaar@gmail.com}
}
\and
Toomas Hinnosaar\thanks{
University of Nottingham and CEPR, \url{toomas@hinnosaar.net}
}
}
\date{May 27, 2026}

\begin{document}
\maketitle
\begin{abstract}
	Social media influencers account for a growing share of marketing worldwide. We demonstrate the existence of a novel form of market failure in the advertising market: influencer cartels, where groups of influencers collude to increase their advertising revenue by inflating their engagement. Our theoretical model shows that influencer cartels can improve consumer welfare if they expand social media engagement to the target audience, or reduce welfare if they divert engagement to less relevant audiences. Drawing on the model's insights, we empirically examine influencer cartels using novel datasets and machine learning tools, and derive policy implications.
\end{abstract}



\section{Introduction} \label{S:intro}

Collusion between a group of market participants to improve their market outcomes is typically considered anticompetitive behavior. While some forms of collusion, such as price-fixing, are illegal in most countries, new industries provide new collusion opportunities for which regulation is not yet well developed. In this paper, we study one such industry---influencer marketing. Influencer marketing combines paid endorsements and product placements by influencers. It allows advertisers to fine-target based on consumer interests by choosing a good product-influencer-consumer match. Influencer marketing is a large and growing industry; with 31 billion U.S. dollars in ad spending in 2023, it is almost as large as the print newspaper advertising.\footnote{Source: \url{https://www.statista.com/outlook/amo/advertising/worldwide}, accessed March 17, 2024.}

Many non-celebrity influencers are not paid based on the success of their marketing campaigns; instead, their compensation depends on past engagement.\footnote{\label{fn:tracking}While influencers with large followings are typically paid based on campaign performance (tracked through sales from personalized links or discount codes), only 19\% of firms employing influencer marketing reported tracking sales during the period covered by our sample \citep{ana_state_2020}. More recent industry data suggest that this share has increased to about 29\% \citep{creatoriq_2024_2024}.} This gives incentives for fraudulent behavior---for inflating their influence. Inflating one's influence is a form of advertising fraud that leads to market inefficiencies by directing ads to the wrong eyeballs. An estimated 15\% of influencer marketing spending was misused due to exaggerated influence.\footnote{Source: \href{https://en.wikipedia.org/wiki/Influencer_marketing}{\url{https://en.wikipedia.org/wiki/Influencer_marketing}}, accessed April 6, 2024.}
To address this problem, the U.S. Federal Trade Commission in 2024 introduced a new rule that prohibits selling and buying fake indicators of social media influence, such as fake followers or views.\footnote{Source: Federal Trade Commission, August 14, 2024, ``Federal Trade Commission Announces Final Rule Banning Fake Reviews and Testimonials'',  \url{https://www.ftc.gov/news-events/news/press-releases/2024/08/federal-trade-commission-announces-final-rule-banning-fake-reviews-testimonials}, accessed November 21, 2025.}
In this paper, we study a way of obtaining fake engagement that does not directly fall under the proposed rule, but is in the same spirit. We study influencer cartels where groups of influencers collude to increase each other's indicators of social media influence.\footnote{\label{fn:cartelsfootnote} On various platforms, these groups are called different names, such as ``pods'' on Instagram and ``exchanges'' on SoundCloud. We use the term ``cartels'' instead of a neutral term like ``club'' to highlight that their primary objective is to manipulate the marketplace to obtain higher prices through prohibited practices, which mirror the traditional concept of cartels.}
While there is substantial literature in economics on fake consumer reviews \citep{mayzlin_promotional_2014, luca_fake_2016, he_market_2022, glazer_fake_2021, smirnov_bad_2022} and other forms of advertising fraud \citep{zinman_wintertime_2016, rhodes_false_2018}, the economics of this fraudulent behavior has not been studied.

We study how influencers collude to inflate engagement, and the conditions under which influencer cartels can be welfare-improving. Our research makes three key contributions. First, we develop a new theoretical framework for influencer cartels, a setting that has not been studied before. Second, we use a novel dataset of influencer cartels and machine learning tools to generate engagement quality measures from text and photos. Third, for each type of cartel, we estimate whether it is likely to be welfare-improving or not. We document that general cartels generate lower-quality engagement than topic cartels, which are closer to natural engagement. This suggests that narrow, topic cartels could improve welfare, while general cartels are detrimental to all involved.

In an influencer cartel, a group of influencers colludes to inflate their engagement, in order to increase the prices they can get from advertisers. As in traditional industries, influencer cartels involve a formal agreement to manipulate the market for members' benefit. In traditional industries, the agreement typically involves price fixing or allocating markets.\footnote{Collusion in cartels does not always occur via fixing prices or output \citep{genesove_rules_2001}.} Influencer cartels involve a formal agreement to inflate the engagement measures to increase their prices. Instead of smoky backroom deals, influencer cartels operate in online chat rooms or discussion boards, where members submit links to their content for additional engagement. In return, they must engage with other members' content by providing likes and meaningful comments. An algorithm enforces the cartel rules.

Our theoretical model focuses on the key market failure in this setting---the free-rider problem. Engaging with other influencers' content brings attention to someone else's content, creating a positive externality. Without cartels, influencers do not engage with each other's content enough, because they do not internalize the externality. A cartel could lessen the free-rider problem by internalizing the externality. By joining the cartel, influencers agree to engage more than the equilibrium engagement. They are compensated for this additional engagement by receiving similar engagement from other cartel members. If the cartel only brings new engagement from influencers with closely related interests, this could benefit cartel members but also consumers and advertisers.
However, the influencer cartel can also create new distortions. The cartel may overshoot and create too much low-quality engagement. The low-quality engagement may hurt all involved parties, consumers, advertisers, and indirectly, even the influencers themselves.

The dimension that separates socially beneficial cooperation from welfare-reducing cartels is the quality of engagement. By high quality we mean engagement coming from influencers with similar interests. The idea is that influencers provide value to advertisers by promoting the product among people with similar interests, e.g., vegan burgers to vegans. If a cartel generates engagement from influencers with other interests (e.g., meat-lovers), this hurts consumers and advertisers. Consumers are hurt because the platform shows them irrelevant content, and advertisers are hurt because their ads are shown to the wrong consumers. Whether or not a particular cartel is welfare-reducing or welfare-improving is an empirical question.

In our empirical analysis, we combine data from two sources: cartel interactions from Telegram and data from Instagram. Our cartel data allows us to directly observe cartel activity (not predict or estimate it). We observe which Instagram posts are included in the cartel and which engagement originates from the cartel.\footnote{The ability to directly observe cartel activity is rather unique. Most studies of cartels in traditional industries have to rely on either historical data of known cartels from the time cartels were legal \citep{porter_study_1983,genesove_rules_2001, roller_workings_2006,hyytinen_cartels_2018, hyytinen_anatomy_2019} or data from the court cases \citep{igami_measuring_2022}, including bidding rings in auctions (for example, \cite{porter_detection_1993, pesendorfer_study_2000, asker_study_2010, kawai_value_2026}), for an overview see \cite{marshall_economics_2012}.}
Our dataset includes two types of cartels differentiated by cartel entry rules: three topic cartels that only accept influencers posting on specific topics and six general cartels with unrestricted topics.\footnote{To make the results comparable, we only study cartels by the same cartel organizer; they all run on the same platform and have the same engagement requirements. The cartels differ only in entry requirements: some cartels have topic restrictions, while others have minimal follower requirements.}

We use machine learning to analyze text and photos from Instagram to measure engagement (match) quality. Our goal is to compare the quality of natural engagement to that originating from the cartel. We measure quality by the topic match between cartel members and users who engage with their content. To quantify the similarity of Instagram users we use Large Language Models to convert text and photos to numeric vectors (embeddings) and calculate cosine similarity between users based on these vectors. We then estimate a panel data fixed effects regression with cosine similarity as the outcome. We complement the analysis using the Latent Dirichlet Allocation to map each user's content into a distribution over topics.

We find that engagement from general cartels is significantly lower in quality compared to natural engagement. Specifically, the quality of engagement from these cartels is nearly as low as that from a counterfactual engagement from a random Instagram user. In contrast, engagement from topic cartels is much closer to the quality of natural engagement. Our back-of-the-envelope calculations show that if an advertiser pays for cartel engagement as if it is natural engagement, they only get 3--18\% of the value in the case of general cartels and 60--85\% in the case of topic cartels. Our results are robust to alternative ways to construct outcome variables and alternative samples. Our estimates also validate our outcome measures: advertisers are known to pay higher fees for more engagement, and we show that users with higher similarity are more likely to engage. This implies our similarity measures capture what advertisers value.

The trade-offs studied in this paper can arise in other settings.
\cite{lampe_strategic_2012} shows that patent applications strategically withhold citations. To overcome the free-rider problem, firms have formed patent pools---agreements to cross-license their patents \citep{lerner_efficient_2004, moser_patents_2013}.\footnote{Patent pools are also formed for reasons beyond internalizing externalities, including reducing transaction costs and addressing blocking patents.}
There is also evidence of citation agreements in academic publishing
\citep{franck_scientific_1999, van_noorden_brazilian_2013, wilhite_coercive_2012}.
Due to anomalous citation patterns, Clarivate (formerly Thomson Reuters) regularly excludes journals from Impact Factor listings.\footnote{Source: \url{https://journalcitationreports.zendesk.com/hc/en-gb/articles/28351398819089-Title-Suppressions}, accessed December 16, 2025.}
There are differences between these settings and influencer cartels. First, in citation cartels, explicit agreements are unobservable, whereas collusion and outcomes are directly observable in influencer cartels. Second, patent and academic citations have rather objectively defined relevance.

Our paper adds to the literature on social media and attention (for an overview, see \cite{aridor_economics_2024}). While the literature on social media has extensively studied the consumption and production of social media content, there is less work on strategic engagement, which is the strategic choice of which content to engage with.
Our paper is most closely related to \citet{filippas_production_2025}, who use Twitter data to study attention bartering. Similarly to our paper, they model social media users' decision to engage (in their setting, whether to follow other users) as a partially reciprocal process.
Unlike us, they study pairwise agreements and they focus on vertically differentiated social media users (i.e., whether social media stars engage with users who have a smaller number of followers), whereas we focus on horizontal differentiation (topics and topic similarities).
Another key difference is that in \citet{filippas_production_2025}, users know each other's characteristics before deciding to barter, whereas in our setting, cartel members commit to engagement without knowing whose content they will engage with. Finally, \citet{filippas_production_2025} don't model advertising, which plays a key role in our analysis.
Our theoretical model of influencer engagement and advertisement builds on classic models of product differentiation \citep{salop_monopolistic_1979} and models of attention and advertising \citep{anderson_market_2005,anderson_competition_2012,anderson_ad_2023}. Unlike all these papers, we study the groups of users who agree to collude in order to increase engagement.

Our paper adds to a small but growing literature in economics on influencer marketing. The empirical literature has analyzed advertising disclosure \citep{ershov_effects_2025, ershov_how_2025}, while the theoretical literature has studied the benefits of mandatory disclosure \citep{pei_influencing_2022, mitchell_free_2021, fainmesser_market_2021}, the prioritization of content \citep{szydlowski_deprioritizing_2023}, and social learning with influence-motivated agents \citep{song_social_2025}. In contrast to these papers, we study collusion between influencers. Influencer cartels have been studied qualitatively in marketing and media studies: \cite{omeara_weapons_2019} examined them through the lens of organized labor, \cite{cotter_playing_2019} analyzed discussions among influencers in closed Facebook groups, including those on Instagram cartels, and \cite{miguel_little_2022} conducted in-depth interviews with 20 food influencers on Instagram cartels. Influencer cartels have been studied quantitatively in computer science. \cite{weerasinghe_pod_2020} analyzed approximately two million Instagram posts included in cartels and built a classifier to predict whether a post is part of a cartel. None of these studies analyze the welfare effects or the type of engagement that is generated by influencer cartels.

Our paper also contributes to the literature on welfare-improving cartels. The literature has shown that with negative externalities, such as environmental damage, collusion can improve welfare \citep{buchanan_external_1969,schinkel_can_2017,schinkel_production_2022,asker_two_2024}. \cite{fershtman_dynamic_2000} showed that the positive effect of collusion on product variety may make collusion welfare-improving. \cite{deltas_consumer-surplus-enhancing_2012} showed that collusion may be welfare-improving due to reduced trade costs. In contrast, in our paper the beneficial aspect of collusion arises from a positive externality to other influencers, which can be internalized through reciprocal engagement.

In our empirical analysis, we build on the recent literature in economics that uses text and photos as data.\footnote{For surveys of the uses of text as data in economics, see \cite{gentzkow_text_2019, ash_text_2023}.} In particular, we use Large Language Models and large neural networks to generate embeddings from text and photos. Large Language Models with social media data have been used in economics before, for example, by \cite{ershov_how_2025}. We also use the Latent Dirichlet Allocation model \citep{blei_latent_2003}, which has been recently used in economics, for example, to extract information from Federal Open Market Committee meeting minutes \citep{hansen_transparency_2018}. We combine these tools with the use of the cosine similarity index. This and other similarity indexes have been used as quality measures in economics, for example, by \cite{chen_motivating_2024} and \cite{hinnosaar_externalities_2022}. While many studies have made use of text as data, using photos is still rare in economics (examples include \cite{adukia_what_2023, ash_visual_2021, caprini_visual_2023}).\footnote{One exception is the use of satellite images mostly in development economics, and typically, to measure electricity use, air pollution, land use, or natural resources \citep{donaldson_view_2016}.} As Instagram as a platform is primarily used to share photos, extracting information from photos is particularly important in our setting.

The rest of the paper is organized as follows. In the next section, we provide some institutional details of influencer cartels. \Cref{S:theory} introduces the theoretical model and discusses the welfare implications of influencer cartels. \Cref{S:data} describes the dataset used in our analysis. \Cref{S:empirical} presents the empirical results. \Cref{S:discussion} concludes.

\section{Background: Influencer Cartels}
\label{S:InstitutionalBackground}

In influencer marketing, firms pay influencers for product placements and endorsements. Many non-celebrity influencers are compensated based on engagement metrics, which gives incentives to inflate the engagement. In an influencer cartel, a group of influencers collude to inflate each other's engagement in order to increase the prices they can get from advertisers. Instagram considers these groups to be violating Instagram's policies\footnote{Source: Devin Coldewey, Apr 29, 2020, ``Instagram `pods' game the algorithm by coordinating likes and comments on millions of posts'', TechCrunch.  \scriptsize{\url{https://techcrunch.com/2020/04/29/instagram-pods-game-the-algorithm-by-coordinating-likes-and-comments-on-millions-of-posts/}}.} and has closed down influencer cartels with hundreds of thousands of members.\footnote{Source: \url{https://www.buzzfeednews.com/article/alexkantrowitz/facebook-removes-ten-instagram-algorithm-gaming-groups-with}, accessed August 19, 2024.}

Influencer cartels operate in an online chat room or a discussion board, typically on Telegram or Reddit.\footnote{For more details, see a computer science overview of Instagram cartels operating on Telegram \citep{weerasinghe_pod_2020} or for example: Apr 9, 2019 ``Instagram Pods: What Joining One Could Do For Your Brand'', Influencer Marketing Hub. \url{https://influencermarketinghub.com/instagram-pods/}.} The cartels in our sample operate in Telegram chat rooms and advertise themselves as a way to ``attract lucrative brand partnerships'' (see screenshots in \Cref{F:wolf_main_earnings} in the Online Appendix). In the chat room, members of the cartels submit links to their Instagram content for which they would like to receive additional engagement. In order to receive that engagement, they themselves must engage with a fixed set of links submitted by other users. Specifically, the cartels have the requirement that before submitting a post, the member must like and write meaningful comments to previous $N$ (in our sample, at least the last five) posts submitted by other members.
\Cref{F:wolf_onyx_comments_mapped_anon,F:wolf_onyx_comments_practice_anon} in the Online Appendix show an example of a post submitted to the cartel receiving the required comments. The rules of the cartel are enforced automatically by an algorithm. The cartels in our sample have entry requirements: either thresholds for the minimum number of followers (ranging from 1,000 to 100,000 followers) or restrictions on the topics of the posts.

The cartel increases engagement via both direct effect and indirect effects. The direct effect is the cartel members engaging with each other's posts. This additional engagement generates two types of indirect effects. First, the Instagram algorithm gives higher exposure to posts with higher engagement, leading to even more engagement. Second, an influencer engaging with another user's post increases the likelihood of the post being shown to the influencer’s followers. This happens as the Instagram algorithm is more likely to show posts that the user's social network has engaged with, that is, posts that the users who the user follows have commented on or liked.\footnote{\Cref{F:suggested_for_you_Instagram} in Online Appendix \ref{A:Screenshots} presents a screenshot of an Instagram post that was recommended to the user because their friend liked it.}

How widespread are influencer cartels? This question is inherently difficult to answer, as the goal of influencers joining such cartels is to generate engagement that is indistinguishable from natural engagement. The challenge is similar to other forms of hidden misconduct, such as accounting fraud or corruption, where most behavior goes unreported or undetected. For example, \citet{dyck_who_2010} document that detected cases of corporate fraud represent only a small fraction of the true incidence. Likewise, \citet{leuz_earnings_2003} show how measurement problems arise when misreporting is endogenous to detection and enforcement. To provide some background context, Online Appendix \ref{A:GoogleTrends} presents suggestive evidence from Google Trends on the widespread interest in influencer cartels \citep{google2024trendsinstagram}.

Our data does not allow us to estimate a causal impact of cartel participation on engagement beyond the direct effect. However, we provide correlational evidence (discussed in \Cref{S:SummaryStatistics}) showing that, after joining the cartels, influencers' posts receive more engagement and influencers are more likely to have disclosed sponsored content. Furthermore, the screenshot in \Cref{F:wolf_amplify_earnings} (Online Appendix \ref{A:Screenshots}) shows that the main arguments the cartel organizer uses to convince influencers to join are related to growth in profile prominence and earnings. The fact that these cartels attract many long-term members indicates that at least some influencers expect and perceive cartels to have a positive impact on both engagement and earnings.

\section{Theoretical Model} \label{S:theory}

We start with a simple model that captures the main trade-offs behind influencer engagement and later extend it to more general settings. The analysis focuses on \emph{engagement}, the activity that influencer cartels directly coordinate. All other aspects of influencer marketing, such as content creation or audience growth, are assumed to be independent of engagement choices.

\subsection{Model Setup} \label{SS:setup}

The model consists of three types of agents: influencers, their followers, and advertisers. There is a continuum of influencers, each having some content. Influencers differ by topic: each influencer is characterized by a type $\alpha \in [0,2\pi]$, denoting a location on the \cite{salop_monopolistic_1979} circle.\footnote{Topic $\alpha$ is measured in radians, corresponding to $\tfrac{360\degree}{2\pi}\alpha \in [0\degree,360\degree]$.} The distribution of topics is uniform around the circle.

Each influencer $\alpha$ is randomly matched with another influencer $\alpha'$, whose content influencer $\alpha$ may engage with.\footnote{We thank an anonymous referee for the suggestion to model engagement through pairwise matches. The results would not change if each influencer $\alpha$ were matched with two random influencers, $\alpha'$ and $\alpha''$, so that $\alpha$ can engage with $\alpha''$, and $\alpha'$ can engage with $\alpha$.} If $\alpha$ engages with $\alpha'$, we denote this by $e(\alpha'|\alpha)=1$; otherwise, $e(\alpha'|\alpha)=0$.

Influencer $\alpha$ has a continuum of followers with total measure $R$.  If influencer $\alpha$ engages with the content of influencer $\alpha'$, then all followers of $\alpha$ are exposed to the content created by $\alpha'$. The utility to each follower of $\alpha$ from such engagement is\footnote{Our modeling assumptions regarding the costs and benefits of attention are similar to models of attention and advertising \citep{anderson_market_2005,anderson_competition_2012,anderson_ad_2023}, who also model the cost of attention as a difference between consumers' preferences and the content they consume. This literature does not consider the externality that is our main focus.}
\begin{equation}
	U^F = e(\alpha'|\alpha)\big[ \cos(\Delta) - C(\Delta) \big],
\end{equation}
where $\Delta = d(\alpha,\alpha')$ denotes the distance between topics $\alpha'$ and $\alpha$. The first term, $\cos(\Delta)$, represents the informational or entertainment value of engagement, while $C(\Delta)$ captures the cost of attention.\footnote{$d(\alpha,\alpha') = \min\{\,|\alpha' - \alpha|,\, 2\pi - |\alpha' - \alpha|\,\} \in [0,\pi]$ denotes the shortest angular distance on the circle.}

The cost function is assumed to be
\begin{equation}
	C(\Delta)
	=
	\begin{cases}
		\sin(\Delta), & \text{if } \Delta \leq \frac{\pi}{2}, \\
		1,            & \text{if } \Delta > \frac{\pi}{2}.
	\end{cases}
\end{equation}

Each influencer is matched with an advertiser promoting a product related to that topic. Hence, when influencer $\alpha$ engages with the content of $\alpha'$, the followers of $\alpha$ are also exposed to the advertisement associated with $\alpha'$. The probability that a follower purchases the product is increasing in the topic match $\cos(\Delta)$. Let $v \ge 0$ denote the surplus generated per successful purchase. We assume that the total value created by engagement is $R v \cos(\Delta)$. The payoff for an advertiser is
\begin{equation}
	U^A = e(\alpha'|\alpha)\big[R v \cos(\Delta) - p\big],
\end{equation}
where $p$ is the \emph{price of engagement}, i.e., the payment from the advertiser to influencer $\alpha'$.

Advertisers, however, observe only the \emph{quantity} of engagement rather than its source or \emph{quality}. We assume that the advertising market is competitive, $\expect[U^A]=0$, hence the expected payment from the advertiser to the influencer $\alpha'$ is\footnote{Later we relax this assumption. The results remain qualitatively the same when $p$ is determined by Nash bargaining between the influencer and the advertiser.}
\begin{equation}
	p = R v \, \expect\!\left[\cos(\Delta) \,\middle|\, e(\alpha'|\alpha)=1\right].
\end{equation}

We assume that influencers care about advertising revenue as well as the value they provide to their followers. The payoff of influencer $\alpha$ who is matched with $\alpha'$ consists of four parts. If $\alpha$ engages with $\alpha'$, she bears the attention cost $C(\Delta)$ of her followers and internalizes only a fraction $\gamma \in (0,1)$ of the benefit created. The remaining share, $1-\gamma$, is an external benefit to influencer $\alpha'$, who receives this engagement. In addition, the recipient of engagement receives the advertising income $p$. Formally,
\begin{align}
	U^I
	 & = e(\alpha'|\alpha)\big[ \gamma \cos(\Delta) - C(\Delta) \big]
	+ e(\alpha|\alpha')\big[ (1-\gamma) \cos(\Delta) + p \big].
\end{align}

The influencer's payoff captures the idea that the current value provided to followers has a persistent impact, with asymmetric costs and benefits. The more different the content is, the greater the attention cost for followers. We assume that the influencer internalizes followers' attention costs, which in reduced form captures the idea that the more different the content is, the more likely followers are to stop following. Conversely, if followers find the content interesting, both the content creator and the engaging influencer benefit, which in reduced form captures that they gain followers. Importantly, when the influencer engages with dissimilar content, the content creator’s followers and therefore the content creator are not affected. Thus, the costs of engaging are borne only by the influencer who engages, while high-quality engagement benefits both parties.

The total welfare from all engagements is
\begin{equation}
	W = \expect\!\left[ U^I + R U^F + U^A \right],
\end{equation}
where the expectation is taken over all interacting pairs of influencers $(\alpha, \alpha')$.

We consider two types of engagement behavior. Under \emph{natural engagement}, influencers decide independently whether to engage each time they have an opportunity to do so. Under \emph{cartel engagement}, the decision is determined by algorithmic rules imposed by the cartel, conditional on membership. We first consider each type of engagement in isolation, and then their interaction.

\subsection{Only Natural Engagement}

We assume first that each influencer independently chooses whether or not to engage with the influencer they are matched with. Taking the difference in $U^I$ between engagement and no engagement, we get that the net value of engagement is $\gamma \cos(\Delta) - C(\Delta)$.
This expression is positive if and only if $\gamma \ge \tan(\Delta)$, which gives the natural engagement function
\begin{equation}
	e^N(\alpha'|\alpha) = \one[\Delta \le \Lambda^N],
\end{equation}
where $\Lambda^N = \arctan(\gamma) < \arctan(1) = \frac{\pi}{4}$ is the natural engagement threshold.\footnote{The angle $\frac{\pi}{4}$ corresponds to $45\degree$.}

Importantly, influencer $\alpha$ does not internalize the positive externality that her engagement creates for $\alpha'$. In contrast, socially optimal engagement would account for this externality. The socially optimal engagement function is $e^S(\alpha'|\alpha) = \one[\Delta \le \Lambda]$ for some $\Lambda \in [0,\pi]$. The corresponding social welfare function can be written as\footnote{We define social welfare as the sum of all agents' payoffs. We use this specification for tractability, and our qualitative results would be unchanged if welfare were evaluated separately for different groups of agents.}
\begin{align}
	W(\Lambda)
	 & = \frac{1}{\pi} \int_0^{\Lambda} \Big(
	(R+1+Rv)\cos(\Delta) - (R+1)C(\Delta)
	\Big)\, d\Delta.
\end{align}

The welfare-maximizing cutoff is $\Lambda^S(v)=\arctan\!\left(1+\frac{Rv}{R+1}\right)$, which is strictly increasing in $v$, satisfies $\Lambda^S(0)=\frac{\pi}{4}$, and converges to $\frac{\pi}{2}$ as $v\to\infty$.\footnote{Proofs and derivations are in Online Appendix \ref{OA:proof_P_natural_social}.}

The following proposition summarizes the equilibrium and the social optimum.

\begin{proposition} \label{P:natural_social}
	There exists a unique equilibrium.
	There is more engagement in the social optimum than in equilibrium (natural engagement), but the additional engagement is of lower quality.
	In particular:
	\begin{enumerate}
		\item In equilibrium,
		      $e^N(\alpha'|\alpha) = \one[\Delta \le \Lambda^N]$, where
		      $\Lambda^N = \arctan(\gamma) < \frac{\pi}{4}$.
		\item In social optimum,
		      $e^S(\alpha'|\alpha) = \one[\Delta \le \Lambda^S(v)]$,
		      where $\Lambda^S(v)$ is a strictly increasing function of $v$, satisfying $\Lambda^S(0) = \frac{\pi}{4}$ and $\lim_{v \to \infty} \Lambda^S(v) = \frac{\pi}{2}$.
	\end{enumerate}
\end{proposition}

Hence, socially optimal engagement involves more engagement than natural engagement, but the additional engagement has lower topic similarity (lower quality). Moreover, for all $v$, $\Lambda^N < \frac{\pi}{4} \le \Lambda^S(v) < \frac{\pi}{2}$.

\subsection{Only Cartel Engagement}

Suppose that, instead of natural engagement, engagement arises through a cartel. Specifically, the cartel rules define an engagement requirement $\Lambda \in [0,\pi]$, such that the engagement function is given by $e^C(\alpha'|\alpha) = \one[\Delta \le \Lambda]$. All cartel members are assumed to follow this rule.\footnote{In practice, the rules are enforced by an algorithm that automatically detects and penalizes deviations.} Influencers therefore choose only whether to join the cartel, making this decision before learning the identity of their match. Advertisers are aware of the cartel's existence and the price of engagement reflects the specific value of $\Lambda$. The objective of the cartel is to maximize the expected payoff of its members.

The expected payoff from joining a cartel with engagement threshold $\Lambda$ is
\begin{align}
	\expect[U^C]
	 & = \frac{1}{\pi} \int_0^{\Lambda}
	\Big(
	(1+Rv)\cos(\Delta) - C(\Delta)
	\Big)\, d\Delta. \label{E:singlecartelU}
\end{align}

Before studying the optimal cartel, it is useful to describe which values of $\Lambda$ are feasible, that is, for which values of $\Lambda$ influencers are willing to join the cartel. There exists a maximal threshold $\Lambda^{\max} \in [0,\pi]$ such that influencers are willing to join the cartel only if $\Lambda \le \Lambda^{\max}$. Moreover, $\frac{\pi}{2} \le \Lambda^{\max} < \pi$.\footnote{Proofs and derivations are in Online Appendix  \ref{OA:proof_P_cartel}.}

The optimal $\Lambda$ that maximizes \eqref{E:singlecartelU} is $\Lambda^C(v)=\arctan(1+Rv)$.
This leads to the following proposition.

\begin{proposition} \label{P:cartel}
	Feasible cartel engagement levels belong to $[0,\Lambda^{\max}]$, where $\frac{\pi}{2} \le \Lambda^{\max} < \pi$.

	The optimal cartel engagement level $\Lambda^C(v)$ is a strictly increasing function of $v$, with $\Lambda^C(0)=\Lambda^S(0)=\arctan(1)=\frac{\pi}{4}$ and $\lim_{v\to\infty}\Lambda^C(v)=\frac{\pi}{2}$. Moreover, for all $v>0$, $\Lambda^C(v)>\Lambda^S(v)$.
\end{proposition}

To interpret this result, note that the optimal cartel internalizes the externality among influencers by making engagement reciprocal. Thus, it can achieve higher welfare than natural engagement. However, as long as there is some advertising revenue, i.e., $v>0$, the cartel engagement goes even further than socially optimal engagement, because the optimal cartel does not internalize the impact on followers beyond what is already reflected in influencers’ payoffs and therefore places relatively greater weight on advertising revenue.

\subsection{Both Natural and Cartel Engagement} \label{sec:both_natural_and_cartel}

We now combine natural and cartel engagement by assuming that a mass $1-\varepsilon$ of influencers choose their engagement freely, that is, they follow natural engagement. The remaining mass $\varepsilon>0$ of influencers is divided into a large number of cartels, where each cartel $i \in \{1,\dots,m\}$ independently chooses its engagement level $\Lambda_i^C$. The allocation of influencers into cartels and natural engagement is independent of their topic. We assume that $\varepsilon$ is small but positive. Importantly, advertisers know $\varepsilon$, but are unable to distinguish between different types of engagement (natural or cartel), and thus the advertising price $p$ is determined by the expected engagement across all types.
We focus on the limiting case as $m \to \infty$. In that limit, the choice of $\Lambda_i^C$ does not affect the price of engagement.\footnote{Proofs and derivations are in Online Appendix \ref{OA:both_natural_and_cartel}.}

First, consider influencers who do not belong to cartels, that is, those who choose engagement naturally. For them, the optimal engagement has not changed, because the only difference with the analysis above is the changed advertising price, and their advertising revenue does not depend on their own engagement decision. Therefore, their engagement function remains $e^N(\alpha'|\alpha)=\one[\Delta \le \Lambda^N]$, where $\Lambda^N=\arctan(\gamma)$.

For a member of cartel $i$, the analysis is different, as the price of advertising is
\begin{equation}
	p(\varepsilon)
	= Rv \expect[\cos(\Delta) \mid e(\alpha'|\alpha)=1]
	=
	Rv
	\frac{
		(1-\varepsilon) \sin(\Lambda^N)+ \varepsilon \sin(\Lambda^C)
	}{
		(1-\varepsilon) \Lambda^N + \varepsilon \Lambda^C
	}
	\geq v R \rho(\varepsilon,\gamma)
	,
\end{equation}
where $\Lambda^C$ is the engagement requirement of other cartels, which is equal for all cartels in equilibrium, and $\rho(\varepsilon,\gamma)$ is strictly positive for any $\varepsilon,\gamma \in (0,1)$.
The payoff to a member of cartel $i$ is
\begin{equation}
	\expect[U_i^C]
	=
	\frac{1}{\pi} \int_0^{\Lambda_i^C} \big(
	\cos(\Delta) - C(\Delta) + p(\varepsilon)
	\big)\, d\Delta.
\end{equation}
Because each individual cartel is small, $p(\varepsilon)$ is taken as given.

We refer to cartels with $\Lambda_i^C = \pi$ as \emph{general cartels}, since they place no restrictions on the topic similarity.\footnote{In our data, there are two types of cartels: general cartels, with no restrictions on topics, and topic cartels, which restrict entry to specific topics. In our theoretical model, the case where $\Lambda_i^C \ll \pi$ captures the essence of topic cartels, as only engagement with relatively closely related content is required.}
We show that general cartels are feasible and optimal for sufficiently large $v$. Due to the existence of natural engagement, the price of engagement is bounded below by a constant multiplied by $v$. Therefore, when $v$ is large, the advertising revenue dominates any other incentives.

\begin{proposition} \label{P:grandcartels}
	For any $\varepsilon < 1$ and $R > 0$, there exists $\widehat{v} > 0$ such that if $v \ge \widehat{v}$, then all cartels are general cartels, i.e., they set the required engagement level to $\Lambda_i^C = \pi$.
\end{proposition}

Intuitively, these optimal cartels generate pure noise in terms of topic match. More precisely, general cartels with $\Lambda_i^C = \pi$ have $\expect[\cos(\Delta)\mid \Delta \le \Lambda_i^C] = 0$, so there is no benefit to followers or surplus to advertisers, while the costs to followers (attention) and to influencers are strictly positive. Moreover, the cartels lower the price of engagement and therefore reduce the payoff of non-cartel members as well.

\begin{corollary} \label{C:grandcartels}
	If cartels are general cartels with $\Lambda_i^C=\pi$, then
	\begin{enumerate}
		\item Cartels strictly reduce social welfare.
		\item Non-cartel influencers would strictly prefer a lower share of cartels.
		\item Cartel members would strictly prefer a lower share of cartels.
	\end{enumerate}
\end{corollary}

\subsection{Interpretation of the Model} \label{SS:interpretation}

The model presented above can be extended in many directions without changing the main conclusions. The key elements are:
(1) the free-rider problem under natural engagement---as influencers' engagement choices create positive externalities for others, there is insufficient engagement in equilibrium relative to the socially optimal level;
(2) the cartel can internalize this externality through reciprocal engagement; and
(3) the distortion from the advertising market---because advertisers cannot distinguish high-quality engagement from low-quality engagement, and realistic cartels are too small to have a substantial impact on market prices, it may be optimal to organize a cartel that maximizes engagement quantity by generating what is essentially low-value engagement.
Before discussing possible extensions, we briefly outline the motivation behind our modeling choices.

Our model captures the manipulation of followers' attention. Engagement by influencer $\alpha$ with content $\alpha'$ signals to the platform algorithm that this content may be relevant for users with similar interests to $\alpha$, which includes $\alpha$'s followers. Consequently, this content is shown to these followers with higher probability. The follower's payoff function, $\cos(\Delta) - C(\Delta)$, is strictly decreasing in $\Delta$: it is high and positive for small $\Delta$ (followers enjoy consuming such content), but negative for large $\Delta$, where the cost of attention exceeds the value of being shown irrelevant content.

\subsection{Extensions} \label{SS:extensions}

\subsubsection{Sequential Engagement Choices} \label{SSS:sequential}

For simplicity, we have so far presented engagement choices as simultaneous. In practice, such interactions are inherently dynamic: influencers post content over time and can only engage with content that has already been created.

The presented model can be reinterpreted as a model of a dynamic process. Consider an infinite sequence of influencers, where influencer $t$ has type $\alpha_t$. Influencer $t$ can choose to engage with the content of influencer $t-1$, at distance $\Delta_t = |\alpha_t - \alpha_{t-1}|$. The only difference from our baseline model is that we previously assumed engagement choices were matched, i.e., $t$ could engage with $t-1$ and vice versa, but since types are independently drawn, this assumption does not affect the analysis.

Again, it is both important and natural to assume that when influencers join the cartel, they do not yet know the characteristics of the content they may be required to engage with.

\subsubsection{Advertisers with Bargaining Power} \label{SSS:bargaining}

We have so far assumed that advertisers are competitive, so that the price of engagement $p$ equals the expected surplus generated for advertisers. This assumption can be relaxed by allowing $p$ to be determined through Nash bargaining, with bargaining power $\beta>0$ for the influencer and $1 - \beta>0$ for the advertiser. In this case, the price of engagement is multiplied by $\beta$ throughout, and advertisers obtain a strictly positive expected payoff proportional to $1-\beta$. All conclusions remain unchanged, because from the influencer's perspective, this is equivalent to replacing the parameter $v$ with $\beta v$. The only difference is that advertisers, too, would now be strictly worse off in the presence of general cartels.

\subsubsection{Advertisers Observing Engagement Quality} \label{SSS:observable}

The key distortion in the model is that advertisers observe only whether engagement occurs, but not its source or quality. This assumption reflects real-world influencer marketing, where payments typically depend only on observable metrics such as the number of views, comments, or likes. However, it is possible (although more costly) for advertisers to employ more sophisticated tracking technologies and either (1) evaluate engagement quality as we do in the empirical analysis, or (2) track sales generated by specific influencers and engagements. In this case, the payment from the advertiser to the influencer would depend on the true match quality rather than its expectation, that is, $p = R v \max\{0, \cos(\Delta)\}$, or a fraction $\beta$ of it if the advertiser has bargaining power.

If we substituted this expression into the analysis above, cartels would never require engagement above $\frac{\pi}{2}$, because at this range the marginal advertising revenue is zero, and even a fully internalized marginal benefit of engagement is strictly less than the cost. In other words, general cartels can only exist because of imperfect observability.

\subsubsection{Alternative Objectives of Cartels} \label{SSS:objectives}

In the main model, we assumed that the goal of the cartel is to maximize its members' expected payoffs. This is a natural objective in a typical setting where a group of influencers agree to cooperate: they want the cartel organizer to choose $\Lambda$ that maximizes their payoffs. However, the model also allows us to consider alternative objectives for the cartel organizer:
\begin{enumerate}
	\item \emph{Socially optimal cartel:} As we showed, the socially optimal cartel is always feasible ($\Lambda^S(v) \le \Lambda^{\max}$), so a cartel organizer seeking to maximize social welfare could implement it.
	\item \emph{Engagement-maximizing cartel:} The maximal feasible engagement level $\Lambda^{\max}$ clearly maximizes the amount of engagement among all feasible cartels.
	\item \emph{Revenue-maximizing cartel:} If a cartel could charge an entry fee $\phi > 0$, the effective value of joining the cartel would be $\expect[U^C] - \phi$, so all influencers would join as long as $\phi \le \expect[U^C]$. The revenue-maximizing cartel would therefore set $\phi = \expect[U^C]$ and choose $\Lambda = \Lambda^C$ as above, extracting the entire surplus created for influencers.
	\item \emph{Advertising-revenue-maximizing cartel:} Advertising revenue conditional on engagement is $p$, which decreases in market-wide $\Lambda$, while the likelihood of engagement increases in cartel's own $\Lambda_i^C$. If the cartel is small enough not to have a significant impact on the market price of engagement, it would choose $\Lambda_i^C$ to maximize engagement.
\end{enumerate}
If $v$ is very large and each cartel is small, as discussed above, then all cases except the socially optimal cartel lead to the same conclusion: all cartels are general cartels with $\Lambda_i^C = \pi$. This is because $\pi$ is the largest feasible reach (maximizing engagement), and with sufficiently high $v$, maximizing engagement dominates all other considerations.

\subsubsection{Heterogeneous Reach}  \label{SSS:heterogeneousreach}

So far, we have assumed that all influencers have the same number of followers, $R$. In Online Appendix \ref{A:reach}, we relax this assumption by allowing each influencer $t$ to be characterized by a two-dimensional type $(\alpha_t, R_t)$, where $\alpha_t$ is the topic, still distributed uniformly, and $R_t$ denotes their reach (attention). We assume influencers' payoffs depend on reach.

We show that all qualitative results remain unchanged: natural engagement remains below the socially optimal level, while cartel engagement can exceed it. In particular, when there are many small cartels and the advertising revenue is very important for influencers, all cartels are general cartels.

The heterogeneous-reach extension provides additional insights. Cartels can now generate two distinct distortions. Not only can they require excessively broad engagement (i.e., low-quality topic matching), but high-reach influencers may also choose not to join, leading to low-quality engagement in terms of reach.

High-reach influencers may abstain from joining because of the asymmetry between what they contribute to the cartel (the attention of their own followers) and what they receive in return (the attention of the followers of an average member). To address this, real-world cartels often impose a minimum reach requirement. We show that cartels would indeed sometimes find it optimal to impose a high entry requirement in terms of reach.

\section{Data and Measures of Engagement Quality}\label{S:data}

\subsection{Data Sources}

We combine data from two sources: first, the detailed cartel communications from Telegram, and second,  Instagram posts and engagement data. A detailed description of our data collection is in Online Appendix \ref{A:DataCollection}.

\paragraph{Telegram cartel history.}
From Telegram, we collected the communication history of nine cartels: six general cartels and three topic cartels: fitness \& health, fashion \& beauty, travel \& food \citep{telegram2020desktopapi,authors2020telegramcartel}. This history provided us with three relevant pieces of information for each submission: the Telegram username, Instagram post shortcode, and the time of submission. According to the rules of these cartels, a user must comment on and like at least five posts preceding their own submission before submitting a post to the cartel for engagement.\footnote{While most of the cartels in our dataset require engagement with the last five posts, one topic cartel (fashion \& beauty) required engagement with last seven posts and a general cartel started out requiring seven but changed to five. In our analysis, to make it comparable, we focus on the first five comments in all cartels.}   This rule allows us to clearly identify which cartel members were bound to engage with which Instagram posts. In other words, we directly observe, instead of having to infer, which posts are included in the cartel. Similarly, we observe, instead of having to infer, which engagement originates from the cartel according to the cartel rules. The Telegram cartels include 220,893 unique Instagram posts that we were able to map to 21,068 Instagram users.

\paragraph{Instagram data.}
Our goal is to compare natural engagement to that acquired via cartels. In engagement, we focus on comments instead of likes or views because information on who views the post is not available, and data on who likes the post is more difficult to collect than comments. We already know which cartel members have to comment according to the cartel rules. For comparison, we needed to collect information on natural engagement.

We define \textit{natural engagement} as comments from users who don't belong to any of the cartels in our data.\footnote{In practice, the engagement that we call natural could come from followers, users seeing the influencer's post for the first time, or users in other cartels unobservable to us. If a substantial part of the natural engagement comes from other cartels or viewers induced by the cartel engagement, then our measure of natural engagement match quality is a lower bound, which would make our results even stronger. To avoid potential misclassification, we exclude all cartel commenters from natural engagement, including the ones who, according to the cartel rules, were not required to comment. Our results remain qualitatively the same when we include cartel commenters who were not required to comment in natural engagement.} To obtain information on natural engagement, we focus on each cartel member's first post in any of the nine cartels. For each cartel member's first post in cartels, we collected information on who commented on the post \citep{apify2024instagramcomments}. Then, we used a random number generator and picked a random non-cartel user who had commented on the post. The randomly chosen commenting Instagram users who don't belong to any of our cartels form our control group (natural engagement). Since these are from the earliest post in the cartel, they are less likely to be indirectly affected by the cartel activity.

We collected the text of all public Instagram posts and a photo for cartel members and for the randomly picked non-cartel users \citep{crowdtangle2021cartelposthistory,crowdtangle2024noncartelposthistory,authors2024instagramphotos}. We were unable to collect the content if the initial post had been deleted or made private. We were also unable to collect information on the non-cartel commenters if the initial post had no non-cartel commenters, or if the commenting user's account was private. We also didn't collect information on non-cartel commenting users if they had fewer than ten Instagram posts. Additionally, we excluded about 5\% of the non-cartel commenting users who had associated posts with cartel members.\footnote{The association can happen as Instagram allows  posts to be associated with multiple users (this is different from tagging a user), or it can happen when the user changes usernames.}

\subsection{Measuring Engagement Quality}\label{S:MeasuringEngagementQuality}

Our goal is to compare engagement that originates from cartels to that of natural engagement.
As demonstrated above, the relevant quality measure in this context is the similarity between the interests of the post author and the commenting user.\footnote{\label{fn:commenttext}We focus on similarity based on users’ posts and do not use the text of comments, because comments in this setting are not reliable indicators of engagement quality. Negative comments are rare, and while cartel comments are on average slightly longer (\Cref{T:SumStatComments} in the Online Appendix), this is by construction, as cartel rules require members to leave longer meaningful comments, making them artificially substantial and not reflective of genuine user interest.} In our analysis, we therefore consider engagement to be of high quality if it comes from Instagram users whose own Instagram content has similar topics. Therefore, we measure the similarity of the posts of commenting users to those of the post author. To analyze similarity, we use text and/or photos in Instagram posts and three alternative methods.

\subsubsection{Text Embeddings and Cosine Similarity of Users}\label{S:LaBSE}
First, we use a large language model named Language-agnostic BERT Sentence Embedding (LaBSE) to construct embeddings of text in Instagram posts \citep{feng_language-agnostic_2022}. An embedding represents text as a numerical vector in a multidimensional vector space. The vector representation of text is useful, allowing quantitative similarity comparison of texts via cosine similarity. Cosine similarity is a standard measure of text similarity. This measure is defined as the cosine of the angle between two vectors, providing a similarity score between -1 and 1, where close to 1 means that the texts (vectors) are highly similar. LaBSE builds upon one of the first large language models, Bidirectional Encoder Representations from Transformers (BERT), which was developed by Google researchers \citep{devlin_bert:_2019}. While BERT was originally implemented for the English language, LaBSE extends it to more than 100 languages. The multilingual effectiveness is necessary for us because our sample is multilingual. The LaBSE model transforms each post into a vector of length 768. It does so using a large neural network with approximately 470 million parameters. This enables the model to capture a large range of semantic features in multiple languages.

To create the input for the embedding, we restrict the sample in the following way. First, we restrict the sample to users who have at least ten posts. In the main analysis, we focus on 100 posts per user closest in the symmetric time window to the first post for the cartel member and to the post they commented on for the non-cartel users. Results are qualitatively similar when using a random sample or all posts from 2017 to 2020 (presented in Online Appendix \ref{A:EmpiricalRobustness}). In our main analysis, we create an embedding of each post using hashtags in the post. We focus on hashtags because they typically informatively capture the essence of Instagram posts. Online Appendix \ref{A:EmpiricalRobustness} presents results where the embeddings are created using the whole text of the posts.

To create the input for the embedding, we pre-process the text: (i) transform to lower case; (ii) replace all characters that are not letters, numbers, underscores, or hashtags with a space (these are the only characters allowed in Instagram hashtags); (iii) add a space before each hashtag; (iv) keep only words that start with a hashtag; (v) keep only the first 30 hashtags in each post because Instagram allows only up to 30 hashtags per post; (vi) drop all hashtags that have only a single character because these tend to be uninformative; (vii) drop all hashtags that don't include any letters because these tend to be uninformative. Before creating embeddings, we replace the hashtag and underscore symbols with a space.

Using the embeddings, we calculate the cosine similarity of users. To do that, first, we create an embedding for each Instagram post separately. After obtaining embeddings of each separate post, we generate a single measure for each user by taking the average of the post embeddings for each user. Online Appendix \ref{A:EmpiricalRobustness} presents results where instead of post embeddings, we first combine all users' posts for each month, and obtain one embedding per user and month pair, and then take the average over the months for each user. Using the average embeddings, we calculate the cosine similarity of user pairs.

\subsubsection{Photo and Text Embeddings and Cosine Similarity of Posts}
We also construct embeddings of photos and text. As the above LaBSE model can encode only text, we have to use a different model for processing photos alongside text. We use the Contrastive Language Image Pre-training (CLIP) model, developed by OpenAI \citep{radford_learning_2021}. CLIP maps the contents of photos and text into a shared embedding space. Because CLIP generates embeddings for photos and text that are directly comparable, it allows us to calculate similarity by combining both forms of information. The CLIP model transforms the text and photos into a vector of length 512. It does so using a neural network with approximately 86 million parameters. The advantage of the CLIP model is that it allows combining photos and text. On the other hand, the LaBSE model more precisely captures text.

We use a photo and the text from a single Instagram post for each user. For cartel members, we use the first post each cartel member posted to the cartels. For non-cartel members, we select their closest post within a symmetric time window to the cartel member's post they commented on. To create the text input for the CLIP embedding, we first pre-process the text, keeping the entire text, not just the hashtags: (i) transform it to lower case; (ii) replace question marks, exclamation marks, and new line breaks with a full stop; (iii) replace all characters that are not letters, numbers, full stops, underscores, hashtags, at symbols, or apostrophes with a space; (iv) drop separate groups of characters that don’t include any letters or numbers; (v) add a space before each hashtag; (vi) discard posts that are shorter than three characters. Since the CLIP model has a binding limit of 77 for the number of tokens (text units such as words or subwords), we have to split the text. Specifically, we split the text into sentences. In a small share of cases where the sentence is longer than 77 tokens, we further split it at the 77-token mark. Then, we generate the embedding for each sentence and take the average over all sentences in that post. We also generate an embedding for each photo. Then we take the average of text and photo embeddings of each post. Finally, using the average embeddings, we calculate the cosine similarity for each author and commenter pair.

\subsubsection{Determining Users' Topics Using Latent Dirichlet Allocation}
The models that generate embeddings allow us to measure similarity, but they are somewhat black boxes. To shed some light on the comparison of users' topics, we use a Latent Dirichlet Allocation (LDA) model. The LDA algorithm estimates a probability distribution of topics for each user based on the words used in their posts, and a probability distribution over the words for each topic.

We train the LDA model using hashtags from the same sample of posts with the same text pre-processing as in the LaBSE model used above. To improve learning from the underlying content, we reduce the set of hashtags using standard thresholds following a common approach in text analysis. Specifically, we exclude hashtags that fewer than 50 users use or more than 33\% of the sample uses. This reduces the number of unique hashtags from about 1.5 million to 19,032, giving us a typical medium-sized dictionary suitable for LDA models. To improve the model's performance and avoid giving more weight to users with longer content, we homogenize the length of content over the users. To do that, for each user, we cap the length of content at 1000 hashtags, which is about the 75th percentile and slightly less than three times the median length. Finally, we exclude users with fewer than 15 unique hashtags because there is not enough information to learn their topics. This reduces the number of users by 17 percent. We fix the number of topics to six based on interpretability and the coherence score (\Cref{F:LDA_coherence_scores}). We assign each topic a label based on the most representative hashtags in each topic, that is, the hashtags with the highest probability (\Cref{T:lda_topic_top_words}). The labels are fitness, beauty, fashion, food, entrepreneur, and travel.

\subsection{Sample}\label{S:Sample}

Our analysis is within subject. That is, for the same author, we compare his similarity to cartel commenters versus his similarity to non-cartel commenters. For this comparison, we need to be able to evaluate both similarities for each author. As described above, for some authors, we were not able to collect information on their non-cartel commenters. Furthermore, not all users had enough posts with text to calculate the outcome measures (the embeddings, cosine similarity, LDA topics). The main sample includes only the authors for whom: (A) we were able to find commenters both from cartels and not from cartels; and (B) we had sufficient data to calculate the similarity measures for both types of commenters. \Cref{T:SampleConstruction} in the Online Appendix describes how each restriction reduces the sample and Online Appendix \ref{A:DataCollection} provides further details.

The main regression analysis focuses on 8507 cartel members as authors of content. For those authors, we were able to calculate both the LaBSE and CLIP embeddings, as well as the same embeddings for at least one cartel and non-cartel commenter. The topic analysis focuses on an analogous sample of 6654 authors with LDA topic measures for both cartel and non-cartel commenters.\footnote{The sample is smaller because the LDA topic estimation has stricter data requirements as it uses only hashtags that are sufficiently common in the sample.}  The authors excluded from the main samples due to data limitations are similar to included authors in terms of their LDA topics (\Cref{F:LDA_pod_authors_excluded_sample_regr,F:LDA_pod_authors_excluded_sample_LDA}), but, as expected, have fewer posts (\Cref{T:SumStatMainSampleVsExcluded}).

\subsection{Summary Statistics}\label{S:SummaryStatistics}

Summary statistics of the main sample are presented in the Online Appendix. \Cref{T:SumStatGenVsTop} compares members of general cartels to members of topic cartels, and \Cref{T:SumStatCartelVsNot} compares cartel members to users not in the cartel.
Note that the non-cartel users in the sample are not representative of Instagram users, instead these are active users with enough public content. That is, first, since we are analyzing engagement, we have to focus on active Instagram users who comment on others' content. Second, when calculating users' interests based on their content, we have to focus on Instagram users with enough public content. This avoids potential measurement issues, such as  having more detailed information on cartel members than on users not in cartels. \Cref{T:SumStatCartelVsNot} shows that indeed, users not in cartels are rather similar to cartel members. One might worry that perhaps those users belong to other cartels not in our sample. In that case, our estimate of the natural engagement match quality is a lower bound, which would make our results even stronger.

\Cref{T:SumStatCartelBeforeVsAfterJoining} provides suggestive evidence that cartels are effective in generating engagement. It compares cartel members before and after joining the cartels. It shows that influencers' posts after joining the cartels have more likes and comments, have higher engagement performance, and after joining the cartels, influencers are more likely to have disclosed sponsored posts.

The distribution of topics in the cartels is as expected (\Cref{F:LDA_topic_distr_authors}). In the fashion \& beauty cartel, authors are posting more about fashion and beauty; in the fitness \& health cartel, about fitness; in the travel \& food cartel, about travel and food. In the general cartels, the topic distribution is rather uniform, with slightly more concentration on the topics of fashion and travel.

\begin{figure}[h!]
	\begin{center}
		\begin{subfigure}[t]{0.45\textwidth}
			\includegraphics[width=1\textwidth]{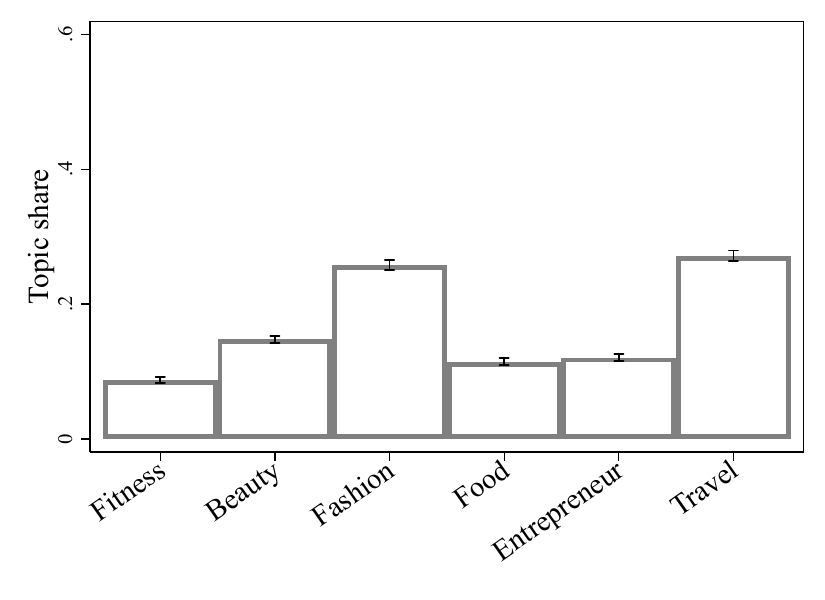}
			\caption{General cartels}
		\end{subfigure}
		\hfill
		\begin{subfigure}[t]{0.45\textwidth}
			\includegraphics[width=1\textwidth]{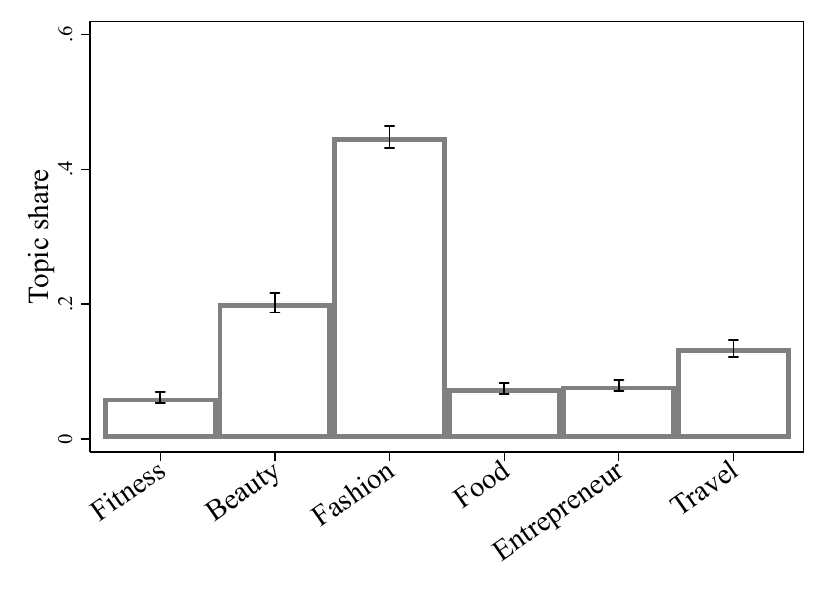}
			\caption{Fashion \& beauty cartel}
		\end{subfigure}
		\begin{subfigure}[t]{0.45\textwidth}
			\includegraphics[width=1\textwidth]{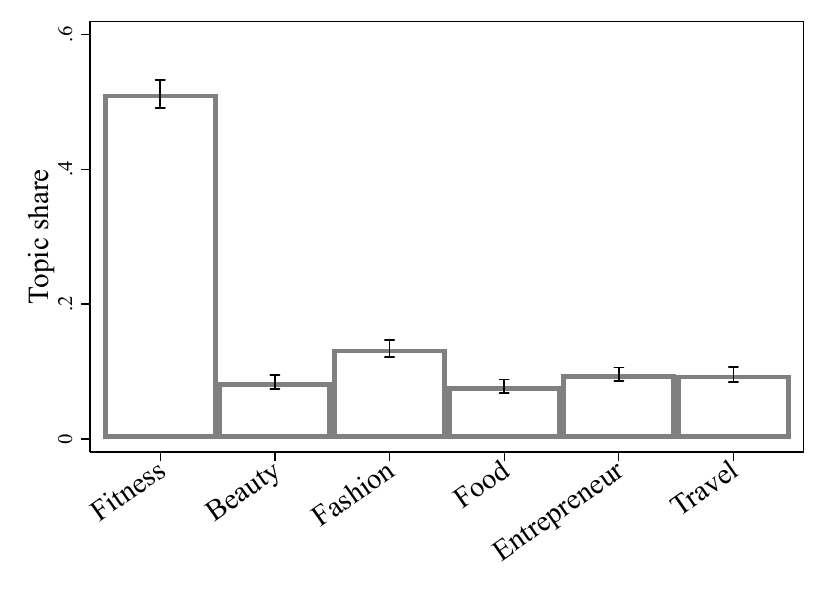}
			\caption{Fitness \& health cartel}
		\end{subfigure}
		\hfill
		\begin{subfigure}[t]{0.45\textwidth}
			\includegraphics[width=1\textwidth]{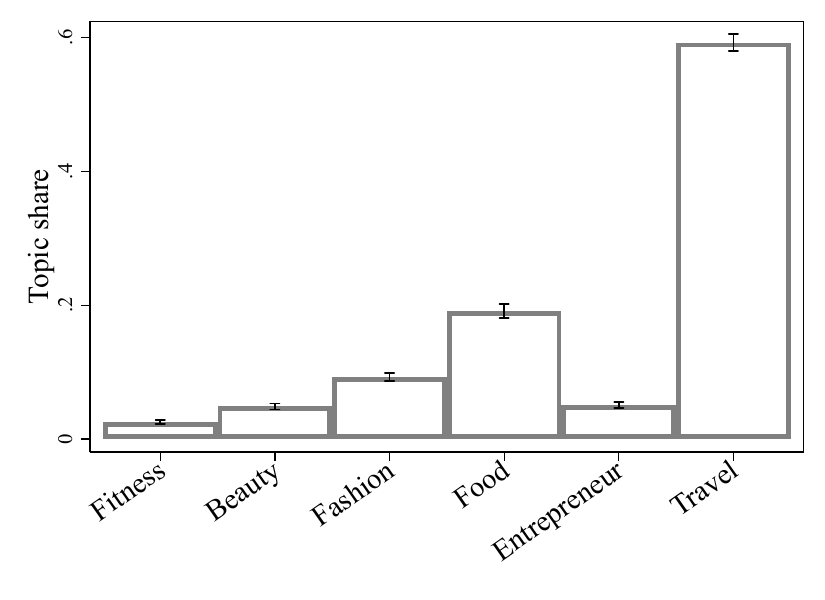}
			\caption{Travel \& food cartel}
		\end{subfigure}
		\caption{Authors' LDA topic distributions}
		\label{F:LDA_topic_distr_authors}
	\end{center}
	\vspace{-0.3cm}
	\footnotesize{Notes:
		The bars correspond to the average topic shares across authors and capped spikes describe the 95\% confidence intervals.
	}
\end{figure}

\section{Empirical Results} \label{S:empirical}

\subsection{Empirical Strategy}
Our empirical question is whether engagement from cartels is of lower quality than natural engagement. To answer the question, we estimate a panel data fixed effects regression where the outcome variable is the cosine similarity between an author and their commenter.
An observation is an author and their commenter pair. For each author, we focus on the first post in the cartel. Thus, we have only one post for each author. For each post, we have three types of commenters. Type one are the cartel members required to comment under cartel rules.\footnote{In the main analysis, we focus on cartel members required to comment instead of those actually commenting. We do this because we don't observe all comments. The robustness analysis shows that the results are similar when looking at who actually commented (Online Appendix \ref{A:EmpiricalRobustness}).} Here we separate general and topic cartel commenters.
Type two is what we call natural engagement: the non-cartel user who actually commented on the post. This serves as a benchmark to test whether cartel commenters have lower similarity (to authors) than non-cartel commenters.
Type three are what we call the counterfactual random users: these are randomly chosen non-cartel users.\footnote{The counterfactual random users are sampled from all the users in our dataset that are not members of any of the cartels. For each post, we first exclude the non-cartel commenter used to estimate the natural engagement, and then draw five random users from the remaining set.} This third group gives us another benchmark. It allows us to measure whether cartel commenters have higher similarity (to authors) than random users.
Our sample is a balanced panel, in the sense that we have all three types of users for each author: cartel commenters, non-cartel commenters, and random users.\footnote{For each post, we have one non-cartel commenter, five random users, and one to five cartel commenters. As described in \Cref{S:Sample}, for some posts, there are fewer than five cartel commenters in the sample because some cartel commenters didn't have enough publicly available content.}

For the first post in cartels of author $i$, the similarity to their commenter $j$ is:
\begin{align}
	Similarity_{ij}
	 & =\beta_{Gen} GeneralCartelCommenter_{ij}
	+\beta_{Top} TopicCartelCommenter_{ij} \nonumber         \\
	 & +\beta_{Ran} RandomUser_{ij} \label{E:MainRegression}
	+ InstagramAuthorFE_i
	+ \varepsilon_{ij},
\end{align}
where $Similarity_{ij}$ refers to the cosine similarity between author $i$ and commenter $j$;
$General \allowbreak CartelCommenter_{ij}$ is an indicator for a general cartel member $j$ who is required to comment;
$TopicCartelCommenter_{ij}$ is an indicator for a topic cartel member $j$ who is required to comment;
$RandomUser_{ij}$ is an indicator for a counterfactual random Instagram user not in the cartel;
$InstagramAuthorFE_i$ is the fixed effect for each author. Since we only have one post per author, this is equivalent to the post fixed effect.
The base category is natural engagement, that is, a commenter who is not in the cartel.

In the main analysis, we use two alternative similarity measures as outcomes: cosine similarity of users from LaBSE text embeddings and cosine similarity of their posts from CLIP photo and text embeddings.\footnote{The details of the outcome variables are described in \Cref{S:MeasuringEngagementQuality}.} We look at three mutually exclusive samples, each defined by the type of cartel containing the author's first cartel post: (1) authors whose first post is only in general cartels; (2) those whose first post is only in topic cartels; and (3) those whose first post is in both.\footnote{With some abuse of terminology, we say that a post is only in general (or topic) cartel even if it was posted in both but we have information on the commenters from only one of these. This affects only a small number of posts.}

To preview our main results, let us look at raw distributions of outcome variables (\Cref{F:Similarity}). Non-cartel commenters have the highest similarity with the author, and random users have the lowest. Commenters from general cartels have almost as low similarity as random users (\Cref{F:Similarity_cs_labse_first_gen}), while commenters from topic cartels have higher similarity (\Cref{F:Similarity_cs_labse_first_top}).

\begin{figure}[h!]
	\begin{center}
		\begin{subfigure}[t]{0.47\textwidth}
			\includegraphics[width=1\textwidth]{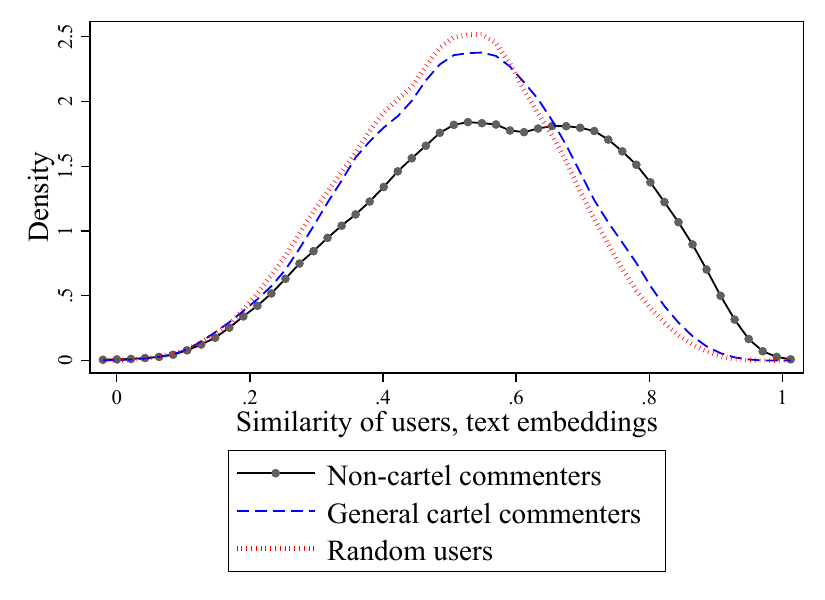}
			\caption{General cartels, users' similarity}\label{F:Similarity_cs_labse_first_gen}
		\end{subfigure}
		\hfill
		\begin{subfigure}[t]{0.47\textwidth}
			\includegraphics[width=1\textwidth]{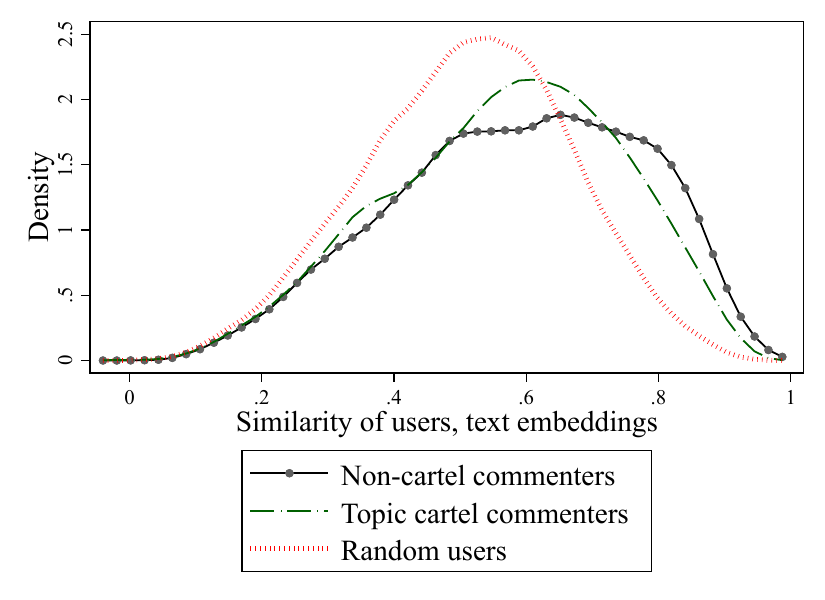}
			\caption{Topic cartels, users' similarity}\label{F:Similarity_cs_labse_first_top}
		\end{subfigure}
		\caption{Probability density of authors' similarity to commenters and random users}
		\label{F:Similarity}
	\end{center}
	\vspace{-0.3cm}
	\footnotesize{Notes:
		The figures present kernel density estimates using the Epanechnikov kernel function of authors' cosine similarity to non-cartel commenters (grey line with solid circle markers), to random users (red dotted line), to general cartel commenters (blue dashed line on \Cref{F:Similarity_cs_labse_first_gen}), and to topic cartel commenters (green dashed and dotted line on \Cref{F:Similarity_cs_labse_first_top}). The cosine similarity is calculated as the similarity of users using the text embeddings from the LaBSE model. \Cref{F:SimilarityClip} presents the probability density estimates for the similarity of posts using the photo and text embeddings from the CLIP model.
	}
\end{figure}

\subsection{Quality of Engagement Measured by Cosine Similarity}\label{SS:EmpiricalMainResults}

We find that in general cartels (columns 1 and 4 in \Cref{T:RegrMain}), authors' similarity to cartel commenters is significantly lower than their similarity to non-cartel commenters, who form the base category.\footnote{We define non-cartel commenters as those who don't belong to any of the cartels in our sample, but they might belong to cartels outside the sample. If some non-cartel commenters are members of cartels outside our sample, then we underestimate the difference between authors' similarity to cartel versus non-cartel commenters.}  Furthermore, similarity to general cartel members is almost as low as to random users. In contrast, in topic cartels (columns 2 and 5), authors' similarity to cartel commenters is only slightly lower than to non-cartel commenters. Similar results hold for posts that are in both general and topic cartels (columns 3  and 6). The regression results are robust to alternative ways to construct outcome variables and alternative samples, described in detail in Online Appendix \ref{A:EmpiricalRobustness}.

\begin{table}[h!]
	\begin{center}
		\begin{footnotesize}
			\caption{Panel data fixed effects estimates of authors' similarity to cartel commenters and random users versus non-cartel commenters}
			\label{T:RegrMain}
			\hspace*{-0.4cm}
			\begin{tabular}{lcccccc}
				\hline
 & (1) & (2) & (3) & (4) & (5) & (6) \\  & \multicolumn{6}{c}{Dependent variable: Cosine similarity} \\ \hline
 & \multicolumn{6}{c}{Posts in general or topic cartels} \\ & General & Topic & Both & General & Topic & Both \\ \hline &  \multicolumn{3}{c}{Similarity of users} & \multicolumn{3}{c}{Similarity of posts} \\ &  \multicolumn{3}{c}{Text embeddings} & \multicolumn{3}{c}{Photo+text embeddings} \\ \cmidrule(lr){2-4}\cmidrule(lr){5-7}
General cartel commenter&      -0.058***&               &      -0.060***&      -0.033***&               &      -0.034***\\
                    &     (0.003)   &               &     (0.008)   &     (0.001)   &               &     (0.003)   \\
Topic cartel commenter&               &      -0.023***&      -0.009   &               &      -0.016***&      -0.012***\\
                    &               &     (0.003)   &     (0.008)   &               &     (0.001)   &     (0.003)   \\
Random user         &      -0.071***&      -0.076***&      -0.062***&      -0.040***&      -0.040***&      -0.037***\\
                    &     (0.003)   &     (0.003)   &     (0.008)   &     (0.001)   &     (0.001)   &     (0.003)   \\
Wald test, $\beta_{Gen} = \beta_{Top}$, p-value&               &               &       0.000   &               &               &       0.000   \\
Base (non-cartel) mean&       0.574   &       0.585   &       0.577   &       0.657   &       0.657   &       0.659   \\
Author fixed effects&         Yes   &         Yes   &         Yes   &         Yes   &         Yes   &         Yes   \\
Authors             &        4756   &        3263   &         488   &        4756   &        3263   &         488   \\
Observations        &       44900   &       30569   &        6665   &       44900   &       30569   &        6665   \\
\hline

			\end{tabular}
		\end{footnotesize}
	\end{center}
	\footnotesize{Notes:
		Each column presents estimates from a separate panel data fixed effects regression. Unit of observation is an author and another user pair. Outcome variable is the cosine similarity of the author to his commenter or to a random user. In columns 1--3, the cosine similarity of users is calculated using the text embeddings from the LaBSE model; in columns 4--6, the cosine similarity of the corresponding users' posts is calculated using the photo and text embeddings from the CLIP model. Each regression includes author fixed effects (equivalent to the post fixed effects because only one post per author). In all the regressions, the base category is the author's similarity to a non-cartel commenter; and \textit{Base (non-cartel) mean} presents their average cosine similarity. \textit{General cartel commenter} is an indicator variable whether the commenter to whom the author's cosine similarity is calculated, is in the general cartel, and \textit{Topic cartel commenter} whether he is in the topic cartel. \textit{Random user} indicates that the author's similarity is calculated to a counterfactual random non-cartel user. The sample consists of authors whose first cartel post is either: only in general cartels (columns 1 and 4); only in topic cartels (columns 2 and 5; or in both general and topic cartels (columns 3 and 6).
		Standard errors in parentheses are clustered at the author level.
	}
\end{table}

To interpret these estimates, we use natural engagement and random match quality as benchmarks.
When we rescale cartel cosine similarity to the natural-random difference, an advertiser paying for general cartel engagement, while expecting natural engagement, captures only 3--18\% of the value. In contrast, with topic cartels, the advertiser still overpays but gets 60--85\% of the value.

Overall, our goal is to measure whether the additional attention that the cartel engagement brings is from users who are likely to be interested in the content. Above, we proxied the interests of the cartel-induced attention by the interests of the engaging cartel influencer.
This is a good proxy because cartel members' natural commenters are likely to be interested in the same topics as cartel members, irrespective of whether it is a topic or general cartel (as shown in \Cref{SS:LDA}). However, in Online Appendix \ref{A:EmpiricalRobustness} (\Cref{T:RegrRobust_cartel_commenters_noncartelcommenter,T:RegrRobust_cartel_commenters_noncartelcommenter_largeboth,T:RegrRobust_cartel_commenters_noncartelcommenter_heter}), we use an alternative approach based on topic match between the author and the commenting cartel members’ commenters.  The results remain qualitatively similar, but adding this additional layer of distance increases noise, and therefore, similarity measures are smaller.

\subsection{Quality of Engagement Measured by LDA Topic Match} \label{SS:LDA}

To further study engagement quality, we compare the LDA topic distributions of cartel versus non-cartel commenters, separately for general and topic cartels. For each author, we define the main topic as the one with the largest LDA probability. Then, within each group of authors sharing the same main topic, we compare the same topic share of the post's cartel versus non-cartel commenters.
First, \Cref{{F:LDA_CartelVsNatural}} illustrates that general and topic cartels are similar in terms of how often non-cartel commenters post on authors' main topic (grey bars).
Second, \Cref{F:LDA_GeneralVsNatural} shows that topic match from general cartels is worse than natural (non-cartel): non-cartel commenters post on the author's main topic about 42\% of the time, while cartel commenters do so only 22\% of the time. This comparison is different for topic cartels. \Cref{F:LDA_TopicVsNatural} focuses on the topics most prevalent in each topic cartel. It shows that the topic match from topic cartels is not worse than natural (non-cartel).

\begin{figure}[h!]
	\begin{center}
		\begin{subfigure}[t]{0.49\textwidth}
			\includegraphics[width=1\textwidth]{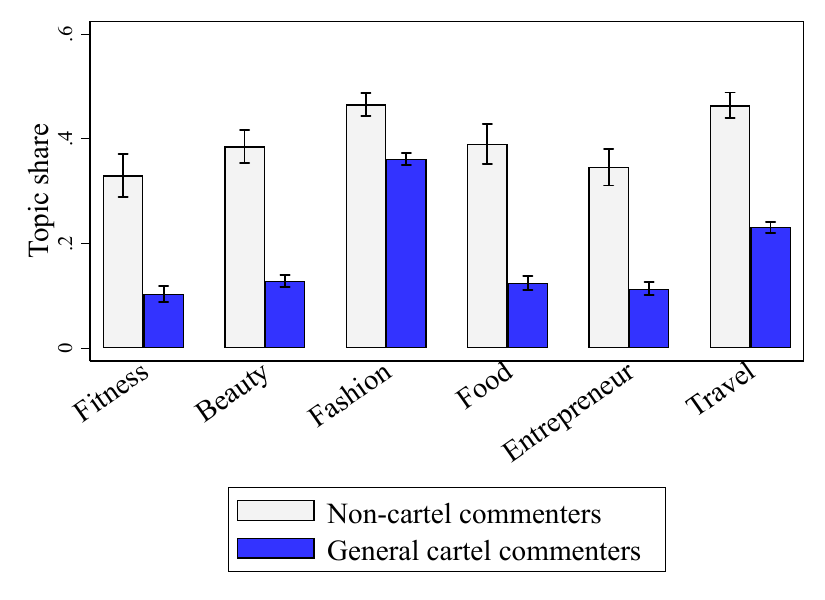}
			\caption{General cartel versus non-cartel commenters}\label{F:LDA_GeneralVsNatural}
		\end{subfigure}
		\begin{subfigure}[t]{0.49\textwidth}
			\includegraphics[width=1\textwidth]{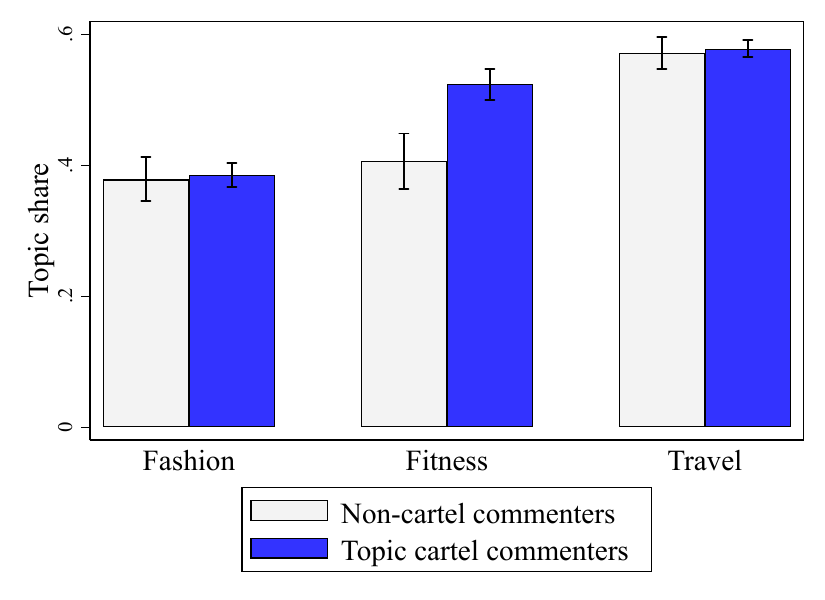}
			\caption{Topic cartel versus non-cartel commenters}\label{F:LDA_TopicVsNatural}
		\end{subfigure}
		\caption{LDA topic shares of  commenters from cartels versus non-cartels (natural)}
		\label{F:LDA_CartelVsNatural}
	\end{center}
	\vspace{-0.3cm}
	\footnotesize{Notes:
		The bars correspond to the average topic shares across commenters and capped spikes describe the 95\% confidence intervals. The grey bars capture the non-cartel commenters and blue bars the cartel commenters. The sample is restricted to general cartel commenters on \Cref{F:LDA_GeneralVsNatural} and topic cartel commenters on \Cref{F:LDA_TopicVsNatural}. Each set of bars uses a sample of commenters that is further restricted by the main topic of the author of the post they comment on, where the main topic is defined as the topic with the highest LDA probability. For example, on  \Cref{F:LDA_GeneralVsNatural}, for the first grey and blue bar, the sample is restricted to commenters on the posts in general cartels of the authors whose main topic is fitness.
	}
\end{figure}

\section{Discussion} \label{S:discussion}

Collusion can take many forms, especially in new and evolving industries. In this paper, we have documented and studied influencer cartels, a form of collusion in the rapidly growing influencer marketing industry, which has stayed under regulators' radar. Our empirical results indicate that engagement from general cartels is significantly lower in quality compared to natural engagement, while engagement from topic cartels is closer to natural engagement. Our theoretical model sheds light on the trade-offs involved and explores the associated welfare implications. The key distortion is the free-rider problem, which cartels could help mitigate through enforced commitment. However, cartels also introduce new distortions, such as over-engagement. These issues become particularly severe when the advertising market heavily rewards the quantity of engagement, encouraging the creation of fake engagement.

Our empirical and theoretical results have three policy and managerial implications. First, as general cartels are likely to be welfare-reducing, shutting these cartels down would be socially beneficial. Second, regulatory rules that prohibit buying and selling fake social media indicators should also prohibit obtaining these fake indicators via in-kind transfers, i.e., paying for engagement with engagement. Third, the current practice of advertisers to reward past engagement encourages harmful collusion. A better approach would be to compensate influencers based on the actual value they add. Alternatively, platforms could improve the outcomes by reporting match-quality-weighted engagement.

While our focus in this paper is to study specifically the distortions that influencer cartels create, the market structure and available data are rich enough to study other related questions. Future research could analyze whether joining a cartel really brings the desired growth in real followers and better advertising deals. While the cartel organizers claim this is the case and correlational evidence (provided in this paper and by \cite{weerasinghe_pod_2020}) supports it, the causal impact is difficult to measure because influencers might join a cartel at the same time they engage in other activities to induce growth.

Other potential extensions require additional data collection. To compare general and topic cartels that are similar in other dimensions, we focused only on Instagram influencer cartels that were organized by the same cartel organizer, had similar engagement rules, and were all relatively large. Future research could study heterogeneity in other dimensions: platforms, engagement requirements, and cartel sizes.

\bibliographystyle{ecta}
\bibliography{toomash-influencers,datasets}

\clearpage
\begingroup
\renewcommand{\thefootnote}{\fnsymbol{footnote}}
\setcounter{footnote}{0}
\begin{center}
{\LARGE Influencer Cartels\par}
\vspace{0.4cm}
{\LARGE Online Appendix\par}
\vspace{1.0cm}
Marit Hinnosaar\footnote{University of Nottingham and CEPR, \url{marit.hinnosaar@gmail.com}}
\hspace{1.0cm}
Toomas Hinnosaar\footnote{University of Nottingham and CEPR, \url{toomas@hinnosaar.net}}
\vspace{0.5cm}

May 27, 2026
\end{center}
\endgroup
\clearpage
\hypersetup{pageanchor=false}
\setcounter{section}{0}
\renewcommand{\thesection}{A\arabic{section}}
\renewcommand{\theHsection}{A\arabic{section}}
\renewcommand{\theHsubsection}{A\arabic{section}.\arabic{subsection}}
\renewcommand{\theHsubsubsection}{A\arabic{section}.\arabic{subsection}.\arabic{subsubsection}}
\renewcommand{\theHfigure}{A.\arabic{figure}}
\renewcommand{\theHtable}{A.\arabic{table}}
\renewcommand{\theHequation}{A.\arabic{equation}}
\counterwithin{table}{section}
\counterwithin{figure}{section}
\setcounter{page}{1}
\renewcommand{\thepage}{A\arabic{page}}
\setcounter{table}{0}
\setcounter{figure}{0}
\listofappendixsection
\listofappendixtable
\listofappendixfigure
\clearpage
\section{Online Appendix: Screenshots of Instagram and Engagement Pods}\label{A:Screenshots}
\addcontentsline{appsec}{appendixsection}{\protect\numberline{\thesection}Screenshots of Instagram and Engagement Pods}
\begin{figure}[ht!]
\begin{center}
\begin{subfigure}[t]{0.65\linewidth}
\includegraphics[width=1\linewidth]{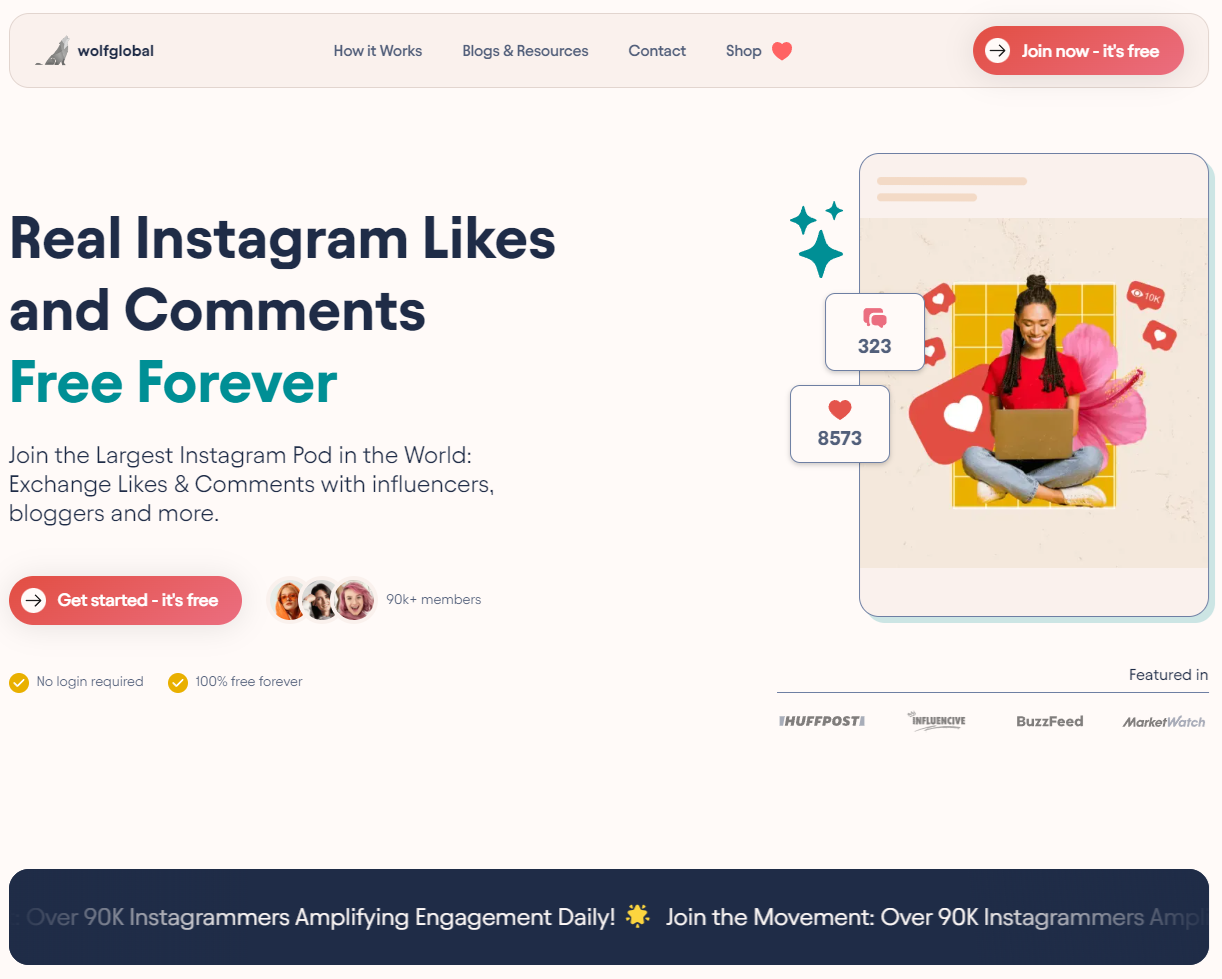}
\caption{Main page}
\end{subfigure}
\begin{subfigure}[t]{0.65\linewidth}
\includegraphics[width=1\linewidth]{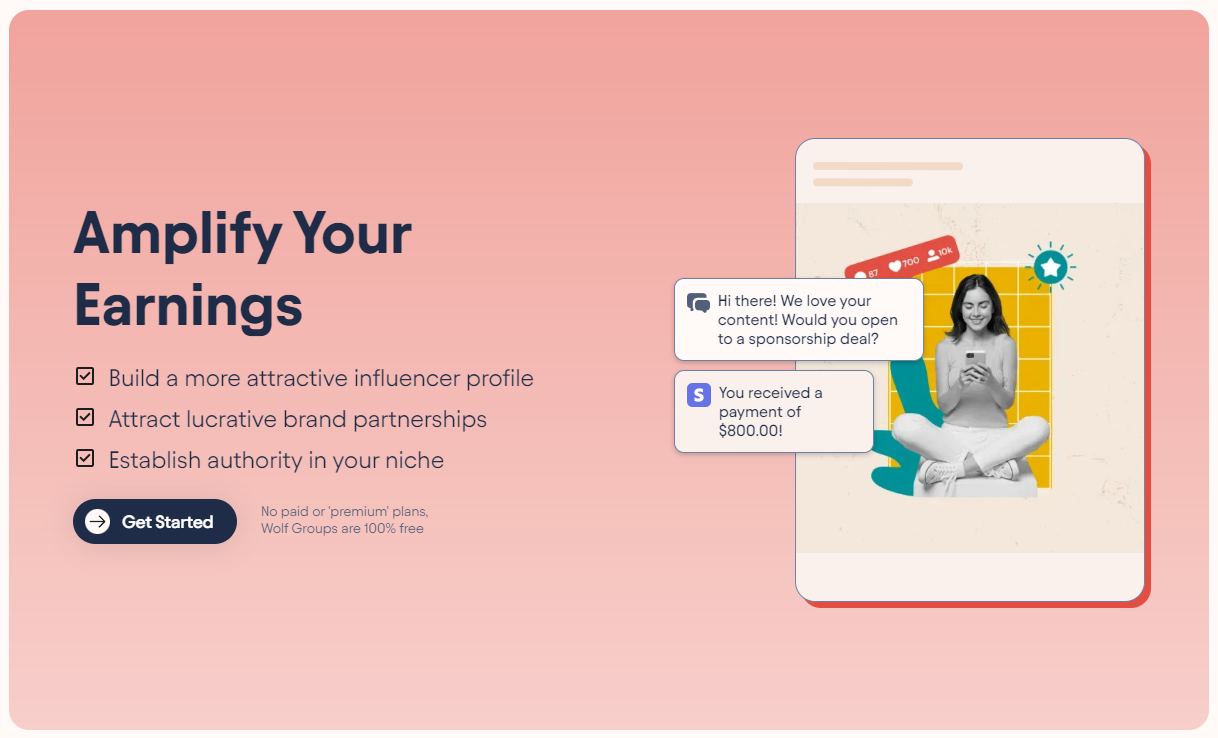}
\caption{Description of how influencers can amplify their earnings}
\label{F:wolf_amplify_earnings}
\end{subfigure}
\caption{Screenshots of Wolf Global Instagram Engagement Pods}
\appfigure{Screenshots of Wolf Global Instagram Engagement Pods}
\label{F:wolf_main_earnings}
\end{center}
\footnotesize{Notes: Screenshots of \url{https://www.wolfglobal.org/}, taken on March 4, 2024.
}
\end{figure}
\clearpage
\begin{figure}[h!]
\begin{center}
\includegraphics[width=0.7\linewidth]{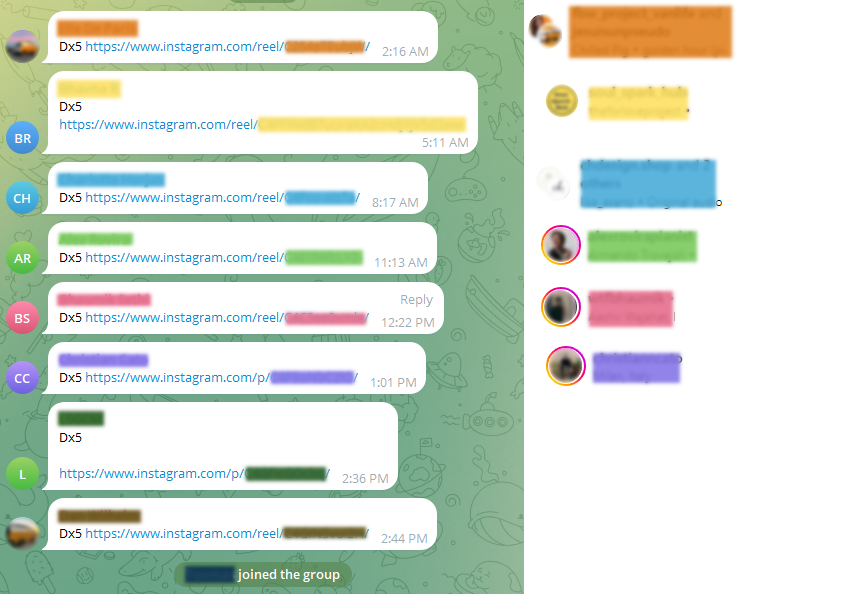}
\caption{Wolf Onyx Comments on Telegram app mapped to Instagram users}
\appfigure{Wolf Onyx Comments on Telegram app mapped to Instagram users}
\label{F:wolf_onyx_comments_mapped_anon}
\end{center}
\footnotesize{Notes: Screenshot of Telegram Wolf Onyx Comments, taken on March 4, 2024.
}
\end{figure}
\begin{figure}[h!]
\begin{center}
\includegraphics[width=0.55\linewidth]{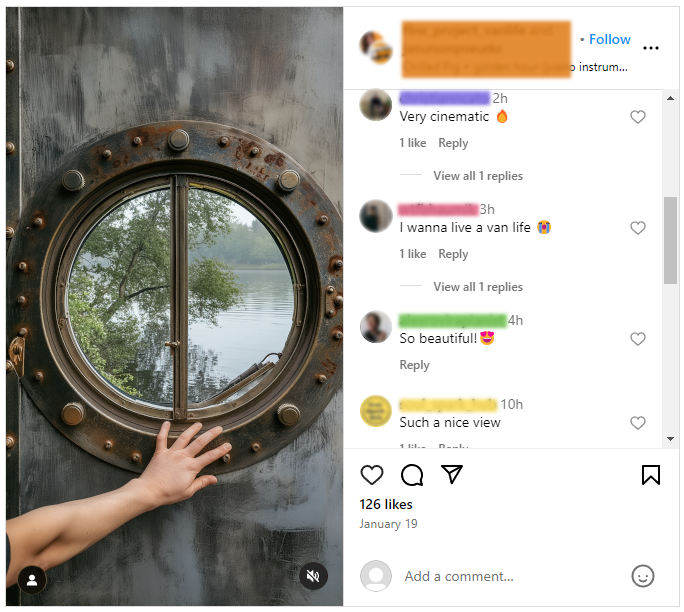}
\caption{Instagram comments coming from Wolf Onyx Comments}
\appfigure{Instagram comments coming from Wolf Onyx Comments}
\label{F:wolf_onyx_comments_practice_anon}
\end{center}
\footnotesize{Notes: Screenshot of Instagram, taken on March 4, 2024.
To preserve anonymity, the Instagram account names are blurred and the photo is replaced with a analogous photo generated by AI image generator.
}
\end{figure}
\begin{figure}[h!]
    \begin{center}
    \includegraphics[width=0.35\linewidth]{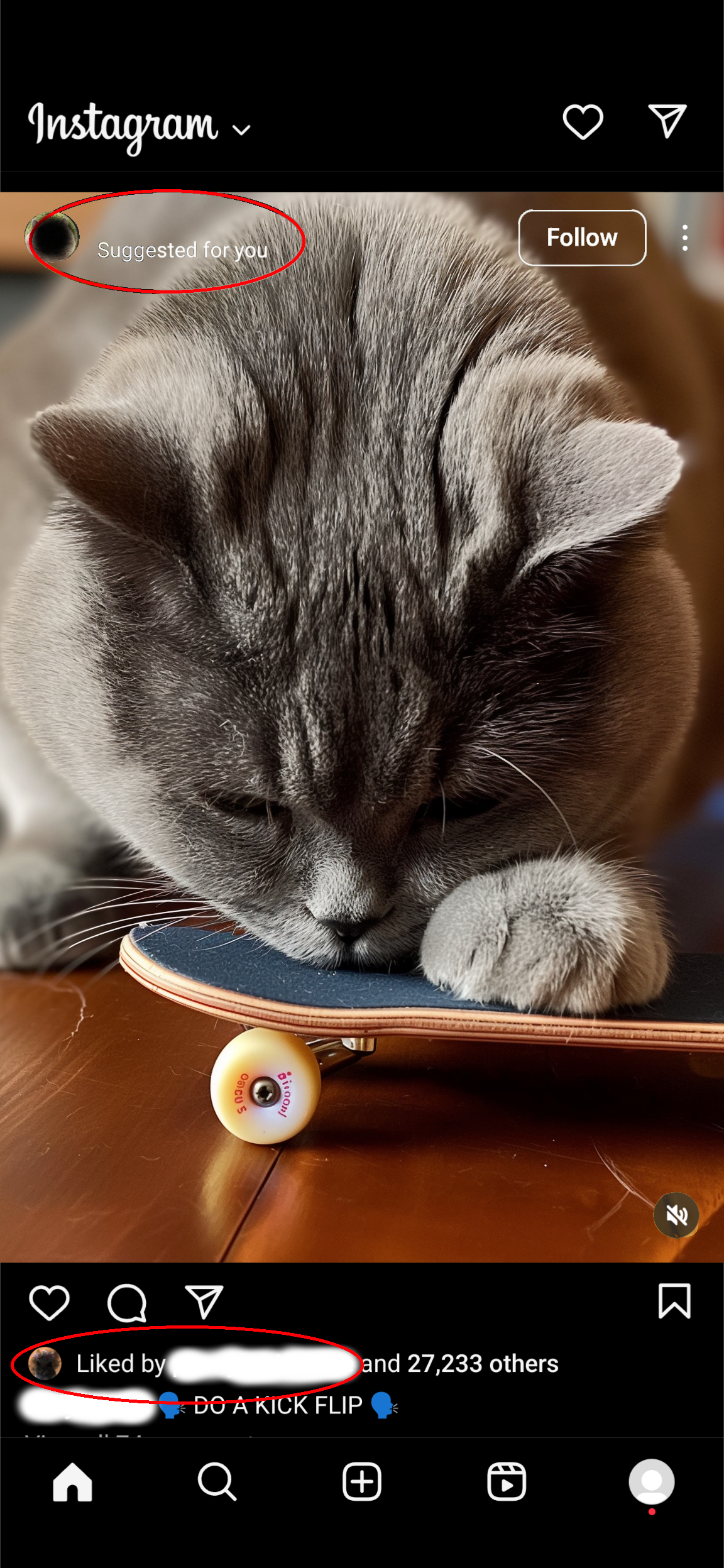}
    \caption{Suggested Instagram post and a user who liked it}
    \appfigure{Suggested Instagram post and a user who liked it}
    \label{F:suggested_for_you_Instagram}
    \end{center}
    \footnotesize{Notes:
    Screenshot of Instagram, taken on May 9, 2024.
    Upper red oval shows that the post was suggested to the viewer (who was not follower for this Instagram account) and lower red oval shows a specific user that the viewer does follow liked the post, indicating that one reason the post was suggested to the viewer was because of this engagement.
    To preserve anonymity, the Instagram account names are blurred and the photo is replaced with a analogous photo generated by AI image generator.
    }
\end{figure}

\clearpage
\section{Online Appendix: Google Trends}\label{A:GoogleTrends}
\addcontentsline{appsec}{appendixsection}{\protect\numberline{\thesection}Google Trends}
One way to estimate the prevalence of these cartels is through search activity. \Cref{F:GoogleTrends} presents Google Trends data for four relevant search terms: ``Instagram algorithm'', ``Instagram pod'', ``Instagram bot'', and ``patent pool''.\footnote{The influencer cartels are commonly called ``Instagram pods'', sometimes also ``influencer pods'' or ``engagement pods''.}
The graph for ``Instagram algorithm'' provides a useful baseline, as influencers are likely to be interested in understanding the platform's mechanics. The upward trend in searches aligns with the platform's growth, with notable spikes following changes to the algorithm.
Searches for ``Instagram pods'' (i.e., cartels) follow a similar pattern but with lower magnitude, averaging 53\% of algorithm-related searches. This suggests a significant portion of Instagram users are aware of and interested in cartels.
Searches for ``Instagram bot",  representing the most basic form of engagement fraud (paid computer-generated clicks, likes, and comments), are more common, averaging 260\% of algorithm-related searches.
In contrast, ``patent pools'' serve as a control benchmark and show relatively flat search activity, with lower volume compared to all Instagram-related terms.
\begin{figure}[h!]
	\begin{center}
		\includegraphics[width=0.95\textwidth]{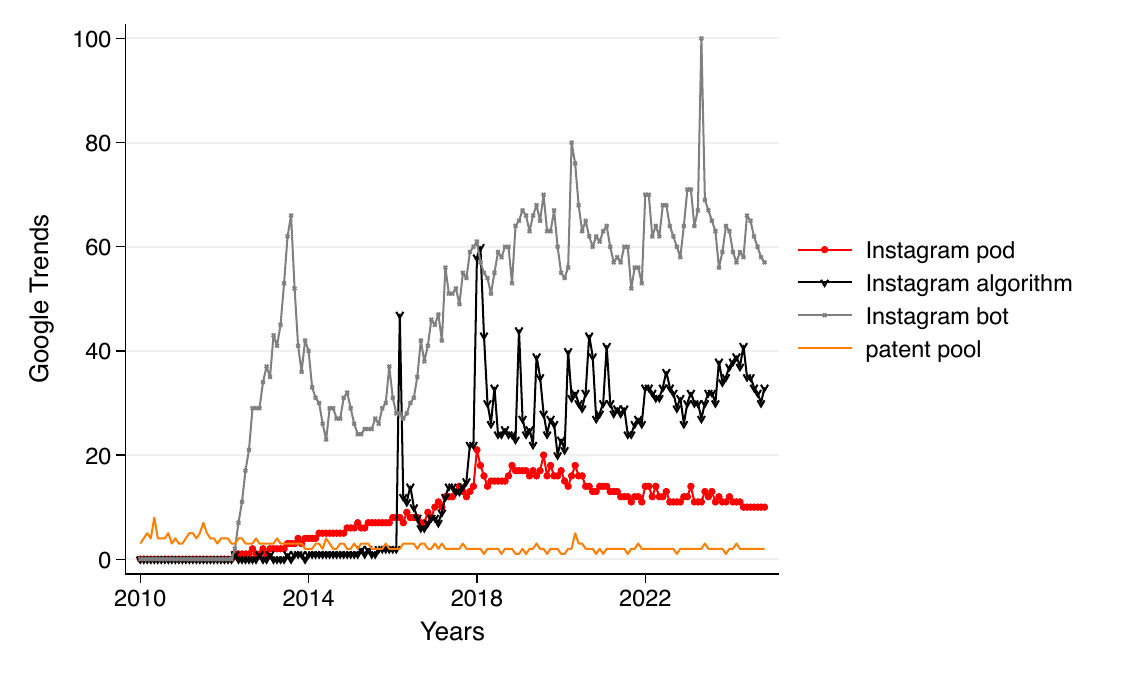}
		\caption{Google Trends}
		\appfigure{Google Trends}
		\label{F:GoogleTrends}
	\end{center}
	\footnotesize{Notes: The figure uses data from Google Trends. The lines measure worldwide Google search volume for search terms "Instagram pod", "Instagram algorithm", "Instagram bot", and "patent pool" in 2010-2024. Influencer cartels are commonly called ``Instagram pods''.
	}
\end{figure}

\clearpage
\section{Online Appendix: Proofs} \label{A:proofs}
\addcontentsline{appsec}{appendixsection}{\protect\numberline{\thesection}Proofs}
\subsection{Proof of \Cref{P:natural_social}} \label{OA:proof_P_natural_social}
Taking the difference in $U^I$ between engagement and no engagement yields net private gain $\gamma \cos(\Delta) - C(\Delta)$.
This expression is nonnegative if and only if
$\gamma \ge \tan(\Delta)$,
which holds if and only if $\Delta \le \Lambda^N$ with $\Lambda^N=\arctan(\gamma)$. Hence the equilibrium engagement rule is
$e^N(\alpha'|\alpha)=\one[\Delta\le \Lambda^N]$.
Because the engagement decision is independent and determined pointwise by the sign of $\gamma \cos(\Delta)-C(\Delta)$, the equilibrium is unique. Finally, since $\gamma<1$, we have $\Lambda^N=\arctan(\gamma)<\arctan(1)=\pi/4$.
Consider cutoff rules $e^S(\alpha'|\alpha)=\one[\Delta\le \Lambda]$. Social welfare is
\[
	W(\Lambda)
	=
	\frac{1}{\pi}\int_0^{\Lambda}
	\Big((R+1+Rv)\cos(\Delta)-(R+1)C(\Delta)\Big)\,d\Delta.
\]
The welfare-maximizing cutoff $\Lambda^S(v)$ satisfies the first-order condition
\[
	(R+1+Rv)\cos(\Lambda)-(R+1)C(\Lambda)=0.
\]
This condition is equivalent to
\[
	\tan(\Lambda)=1+\frac{Rv}{R+1},
\]
and therefore $\Lambda^S(v)=\arctan\!\left(1+\frac{Rv}{R+1}\right)$.
This expression is strictly increasing in $v$, and satisfies $\Lambda^S(0)=\arctan(1)=\pi/4$ and $\lim_{v\to\infty}\Lambda^S(v)=\pi/2$.
As $\Lambda^S(v)\ge \pi/4>\Lambda^N$, the social optimum prescribes engagement for all matches that engage in equilibrium and for additional matches with $\Delta\in(\Lambda^N,\Lambda^S(v)]$. These additional engagements occur at larger topic distances and therefore have lower similarity.
\qed
\subsection{Proof of \Cref{P:cartel}} \label{OA:proof_P_cartel}
The expected payoff from joining a cartel with engagement threshold $\Lambda$ is
\begin{align*}
	\expect[U^C]
	 & = \frac{1}{\pi} \int_0^{\Lambda}
	\biggl(
	\gamma \cos(\Delta) - C(\Delta)
	+ (1-\gamma) \cos(\Delta)
	+
	\underbrace{
		\frac{Rv}{\Lambda} \int_0^{\Lambda} \cos(\Delta')\,d\Delta'
	}_{=p}
	\biggr)
	d\Delta                             \\
	 & = \frac{1}{\pi} \int_0^{\Lambda}
	\Big(
	(1+Rv)\cos(\Delta) - C(\Delta)
	\Big)\, d\Delta.
\end{align*}
Define $\Lambda^{\max}\in[0,\pi]$ as the maximal feasible threshold, that is, the largest $\Lambda$ such that $\expect[U^C]\ge 0$.
We first show that $\Lambda^{\max} \ge \pi/2$. For all $\Lambda\le \pi/2$,
\[
	\expect[U^C]
	=
	\frac{1}{\pi}\int_0^{\Lambda}\big((1+Rv)\cos(\Delta)-\sin(\Delta)\big)\,d\Delta
	=
	\frac{1}{\pi}\Big((1+Rv)\sin(\Lambda)-(1-\cos(\Lambda))\Big).
\]
This expression is non-negative because\footnote{
	We are using the half-angle identities $\sin(\Lambda)=2\sin(\tfrac{\Lambda}{2})\cos(\tfrac{\Lambda}{2})$ and $1-\cos(\Lambda)=2\sin^2(\tfrac{\Lambda}{2})$.
}
\[
	(1+Rv)\,2\sin\!\frac{\Lambda}{2}\cos\!\frac{\Lambda}{2}
	\ge 2\sin^2\!\frac{\Lambda}{2}
	\;\;\Longleftrightarrow\;\;
	1+Rv \ge \tan\!\frac{\Lambda}{2},
\]
which holds for all $\Lambda\le \pi/2$ because then $\Lambda/2\le \pi/4$ and $\tan(\Lambda/2)\le 1 \le 1+Rv$.
Next, we show that $\Lambda^{\max}<\pi$. With $\Lambda=\pi$ we get
\[
	\expect[U^C]
	=
	\frac{1}{\pi} \int_0^{\pi} \big( (1+Rv)\cos(\Delta) - C(\Delta) \big)\, d\Delta
	= -\frac{1}{\pi} - \frac{1}{2} < 0.
\]
Hence there exists a maximal feasible threshold $\Lambda^{\max}\in(\pi/2,\pi)$ such that influencers are willing to join the cartel if and only if $\Lambda\le \Lambda^{\max}$.
We now characterize the optimal cartel engagement. Differentiating \eqref{E:singlecartelU} with respect to $\Lambda$ yields
\[
	\frac{d \expect[U^C]}{d \Lambda}
	=
	\frac{1}{\pi}\big((1+Rv)\cos(\Lambda)-C(\Lambda)\big).
\]
The optimal $\Lambda$ therefore satisfies
\[
	(1+Rv)\cos(\Lambda)-C(\Lambda)=0
	\;\;\Longleftrightarrow\;\;
	\tan(\Lambda)=1+Rv,
\]
and hence $\Lambda^C(v)=\arctan(1+Rv)$. This expression is strictly increasing in $v$, satisfies $\Lambda^C(0)=\arctan(1)=\pi/4$, and converges to $\pi/2$ as $v\to\infty$.
Finally, for all $v>0$,
\[
	\Lambda^C(v)=\arctan(1+Rv)
	>
	\arctan\!\left(1+\frac{Rv}{R+1}\right)
	=\Lambda^S(v),
\]
because $R/(R+1)<1$.
\qed
\subsection{Proofs for \Cref{sec:both_natural_and_cartel}} \label{OA:both_natural_and_cartel}
\subsubsection{Proof of \Cref{P:grandcartels}}
We focus on the limit as $m\to\infty$, so a single cartel takes the price $p(\varepsilon)$ as given. In equilibrium all cartels choose the same engagement requirement, denoted $\Lambda^C$.
By definition,
\[
	p(\varepsilon)
	=
	Rv
	\frac{
		(1-\varepsilon)\sin(\Lambda^N)+\varepsilon\sin(\Lambda^C)
	}{
		(1-\varepsilon)\Lambda^N+\varepsilon\Lambda^C
	}.
\]
Since $\sin(\Lambda^C)\ge 0$ and $\Lambda^C\le \pi$, we have
\[
	p(\varepsilon)
	\ge
	Rv
	\frac{(1-\varepsilon)\sin(\Lambda^N)}{(1-\varepsilon)\Lambda^N+\varepsilon\pi}.
\]
With $\Lambda^N=\arctan(\gamma)$, we have $\sin(\Lambda^N)=\gamma/\sqrt{1+\gamma^2}$, hence
\[
	p(\varepsilon)
	\ge
	Rv
	\frac{(1-\varepsilon)\dfrac{\gamma}{\sqrt{1+\gamma^2}}}{(1-\varepsilon)\arctan(\gamma)+\varepsilon\pi}
	=:Rv\,\rho(\varepsilon,\gamma),
\]
where $\rho(\varepsilon,\gamma)>0$ for all $\varepsilon,\gamma\in(0,1)$.
A cartel with $\Lambda_i^C=\pi$ is feasible if and only if its members obtain a nonnegative expected payoff,
\[
	\expect[U_i^C](\pi)
	=
	\frac{1}{\pi}\int_0^\pi\big(\cos(\Delta)-C(\Delta)+p(\varepsilon)\big)\,d\Delta
	=
	p(\varepsilon)-\frac{1}{\pi}-\frac{1}{2}
	\ge 0,
\]
that is, if $p(\varepsilon)\ge \frac{1}{\pi}+\frac{1}{2}$. Since $p(\varepsilon)\ge Rv\,\rho(\varepsilon,\gamma)$, there exists $\overline v>0$ such that for all $v\ge \overline v$ this feasibility condition holds.
For cartel $i$,
\[
	\expect[U_i^C](\Lambda_i^C)
	=
	\frac{1}{\pi}\int_0^{\Lambda_i^C}\big(\cos(\Delta)-C(\Delta)+p(\varepsilon)\big)\,d\Delta.
\]
Differentiating with respect to $\Lambda_i^C$ gives
\[
	\frac{d\expect[U_i^C]}{d\Lambda_i^C}
	=
	\frac{\cos(\Lambda_i^C)-C(\Lambda_i^C)+p(\varepsilon)}{\pi}.
\]
Evaluating at $\Lambda_i^C=\pi$ yields
\[
	\left.\frac{d\expect[U_i^C]}{d\Lambda_i^C}\right|_{\Lambda_i^C=\pi}
	=
	\frac{\cos(\pi)-C(\pi)+p(\varepsilon)}{\pi}
	=
	\frac{p(\varepsilon)-2}{\pi}.
\]
Using the lower bound $p(\varepsilon)\ge Rv\,\rho(\varepsilon,\gamma)$, there exists $\widehat v\ge \overline v$ such that for all $v\ge \widehat v$,
\[
	\left.\frac{d\expect[U_i^C]}{d\Lambda_i^C}\right|_{\Lambda_i^C=\pi}
	\ge
	\frac{Rv\,\rho(\varepsilon,\gamma)-2}{\pi}
	>0.
\]
Hence, for $v\ge \widehat v$, cartel $i$ strictly prefers $\Lambda_i^C=\pi$ to any smaller engagement requirement. Since cartels are symmetric, in equilibrium all cartels choose $\Lambda_i^C=\pi$.
\qed
\subsubsection{Proof of \Cref{C:grandcartels}}
Assume cartels are general cartels, so $\Lambda^C=\pi$. Then
\[
	p(\varepsilon)
	=
	Rv
	\frac{
		(1-\varepsilon)\sin(\Lambda^N)
	}{
		(1-\varepsilon)\Lambda^N+\varepsilon\pi
	}.
\]
Differentiating gives
\[
	p'(\varepsilon)
	=
	-\,\frac{Rv\,\pi\sin(\Lambda^N)}{\big[(1-\varepsilon)\Lambda^N+\varepsilon\pi\big]^2}
	<0.
\]
A cartel member’s expected payoff satisfies
\[
	\frac{d\expect[U_i^C]}{d\varepsilon}
	=
	p'(\varepsilon)
	<0,
\]
because $\varepsilon$ affects $\expect[U_i^C]$ only through $p(\varepsilon)$. A non-cartel influencer engages with probability $\Lambda^N/\pi$, hence her expected payoff depends on $\varepsilon$ through advertising revenue with the same factor, which implies
\[
	\frac{d\expect[U^I]}{d\varepsilon}
	=
	\frac{\Lambda^N}{\pi}\,p'(\varepsilon)
	<0.
\]
This proves parts (ii) and (iii).
General cartels imply $\expect[\cos(\Delta)\mid \Delta\le \pi]=0$, while attention costs and engagement costs remain strictly positive. Moreover, $p'(\varepsilon)<0$ reduces the payoffs of both cartel members and non-members. Therefore, increasing $\varepsilon$ strictly reduces social welfare, proving part (i).
\qed

\clearpage
\section{Online Appendix: Extension with Heterogeneous Reach} \label{A:reach}
\addcontentsline{appsec}{appendixsection}{\protect\numberline{\thesection}Extension with Heterogeneous Reach}
Here we extend the analysis to the case of heterogeneous reach. Specifically, we assume that every influencer $t$ is characterized by a two-dimensional type $(\alpha_t, R_t)$, where $\alpha_t \in [0,2\pi]$ denotes the topic and $R_t$ the reach. We assume that $\alpha_t$ is uniformly distributed, while $R_t$ follows a power-law distribution with mean $2$.\footnote{That is, the probability density function is $f(R) = 2R^{-3}$ for $R \ge 1$.} As in the main text, each influencer is randomly matched with another influencer $t'$ (equivalently, this can be interpreted as a dynamic model), where the distance is $d(\alpha_t, \alpha_{t'}) = \Delta$. Each influencer can engage either naturally or according to cartel rules.
To keep the analysis tractable, we make two additional simplifying assumptions relative to the main text. First, instead of modeling followers explicitly, we assume that influencers fully internalize the costs and benefits experienced by their followers. Specifically, the payoff of influencer $t$ depends only on the engagement choices, $a_t = e(\alpha_{t'} \mid \alpha_t)$ and $a_{t'} = e(\alpha_t \mid \alpha_{t'})$, and is given by
\begin{equation} \label{E:basic_payoff}
	U^I
	= a_t R_t [\gamma \cos(\Delta) - C(\Delta)]
	+ a_{t'} \big[(1-\gamma) R_{t'} \cos(\Delta) + p_t \big],
\end{equation}
where the price of engagement is given by $p_t = R_{t'} v \, \expect[\cos(\Delta) \mid a_{t'} = 1]$.
Because influencers jointly internalize the surplus generated for both consumers and advertisers, social welfare is simply the sum of all influencers' payoffs:
\begin{equation}
W = \expect[U^I].
\end{equation}
The second simplifying assumption is that we focus on two extremes with respect to the advertising market: either $v = 0$, i.e., there is no advertising, or $v$ is very large, i.e., advertising revenue dominates all other considerations.\footnote{
	To avoid dealing with infinities, we study limits of normalized payoffs, i.e., $\lim_{v \to \infty} (U / v)$, where $U$ denotes the relevant payoff or welfare function.
}
\subsection{Only Natural Engagement}
As in our main model, the positive externality of engagement creates a free-riding problem. Therefore, in equilibrium there is less engagement than is socially optimal, i.e., the conclusions from \Cref{P:natural_social} still hold. The level of natural engagement is $\Lambda^N = \arctan(\gamma) < \frac{\pi}{4}$, whereas the socially optimal engagement level $\Lambda^S$ maximizes
\begin{align}
W
&=
\frac{2}{\pi}
\int_0^{\Lambda}
\big(
(1+v)\cos(\Delta) - \sin(\Delta)
\big)
\, d\Delta,
\end{align}
where we have used the facts that the optimal $\Lambda \le \frac{\pi}{2}$ and that $\expect[R_t] = \expect[R_{t'}] = 2$.
Without an advertising market, the socially optimal engagement level satisfies $\cos(\Lambda) = \sin(\Lambda)$, so that $\Lambda^S(0) = \arctan(1) = \frac{\pi}{4}$. On the other hand, if the advertising market is extremely important, then, in the limit, the advertising surplus dominates all other components of social welfare, and therefore
\begin{equation}
\lim_{v \to \infty} \frac{W}{v} =
\frac{2}{\pi}
\int_0^{\Lambda}
\cos(\Delta)\, d\Delta,
\end{equation}
and the socially optimal engagement level converges to $\lim_{v \to \infty} \Lambda^S(v) = \arccos(0) = \frac{\pi}{2}$. The following proposition formalizes these observations.
\begin{proposition}
There is a unique equilibrium.
There is more engagement in the social optimum than in equilibrium, although the additional engagement is of lower quality. In particular:
	\begin{enumerate}
		\item The natural engagement level is $\Lambda^N = \arctan(\gamma) < \frac{\pi}{4}$.
		\item The socially optimal engagement level is $\Lambda^S = \frac{\pi}{4}$ if $v = 0$, and $\Lambda^S \to \frac{\pi}{2}$ as $v \to \infty$.
	\end{enumerate}
\end{proposition}
\subsection{Only Cartel Engagement}
Suppose that the cartel requires an engagement level $\Lambda$. An influencer $t$ who considers joining the cartel knows his type $(\alpha_t, R_t)$ but does not yet know the characteristics of the matched influencer $t'$. The expected payoff from joining the cartel is
\begin{align}
	\expect U^C(R_t)
	&= \frac{1}{\pi} \int_0^{\Lambda}
	\Big[
	\gamma R_t \cos(\Delta) - R_t C(\Delta)
	+ (1-\gamma)\,\expect[R_{t'} \mid a_{t'} = 1] \cos(\Delta)
	+ p_t
	\Big]\, d\Delta.
\end{align}
\subsubsection{No Advertising}
Let us again start with the case when $v=0$, so that $p_t=0$. It is straightforward to see that there exists a possibly infinite threshold $\oR \geq 1$, such that an influencer with reach $R_t$ joins the cartel if and only if $R_t \leq \oR$. Suppose first that $\Lambda \leq \frac{\pi}{2}$, so that $C(\Delta)=\sin(\Delta)$. Using a monotonic transformation $\lambda = \tan\left(\frac{\Lambda}{2}\right)$, the expression for $\expect U^C(R_t)$ simplifies to\footnote{
	Note that $\frac{1 + \cos(\Lambda)}{\sin(\Lambda)} = 1/\tan\left(\frac{\Lambda}{2}\right) = \frac{1}{\lambda}$, and $\sin(\Lambda) = \sin(2\tan^{-1}(\lambda)) = \frac{2 \lambda}{\lambda^2 +1}$.
}
\begin{align}
\expect U^C(R_t)
&=
\frac{2 \lambda (1-\gamma)}{\pi (1+\lambda^2)} \left(
\expect[R_{t'}|R_{t'} \leq \oR] - \frac{\lambda -\gamma}{1-\gamma} R_t
\right).
\end{align}
Now, if $\lambda \leq \gamma$, then the expression is positive for all $R_t$, and all influencers join the cartel. On the other hand, if $\lambda > \gamma$, then the marginal influencer type $\oR$ must satisfy\footnote{Note that $\expect[R_{t'} | R_{t'} \leq \oR] = \frac{2 \oR}{1 + \oR}$.}
\begin{equation}
\expect U^C(\oR)
=
\frac{2 \lambda (1-\gamma)}{\pi (1+\lambda^2)} \left(
\frac{2 \oR}{1 + \oR} - \frac{\lambda -\gamma}{1-\gamma} \oR
\right)
= 0
\;\;
\iff
\;\;
\oR = \frac{2-\lambda-\gamma}{\lambda-\gamma}.
\end{equation}
This analysis assumed that $\Lambda \leq \frac{\pi}{2}$ and therefore $\lambda = \tan\left(\frac{\Lambda}{2} \right) \leq \tan\left(\frac{\pi}{4} \right) = 1$. Specifically, when $\Lambda = \frac{\pi}{2}$, we have $\lambda = 1$ and therefore $\oR = 1$, i.e., in this extreme case only the lowest-reach influencers with $R_t = 1$ join the cartel. It is easy to see that when $\Lambda > \frac{\pi}{2}$, even the lowest type would not want to join the cartel.
Combining these observations, we get the following proposition.
\begin{proposition} \label{P:cartel_het}
	If $v=0$, then, depending on the cartel agreement, there can be three types of equilibria in the cartel entry game:
	\begin{enumerate}
		\item If $\lambda \leq \gamma$, all influencers join the cartel.
		\item If $\gamma < \lambda < 1$, all influencers with $R_t \leq \oR = \frac{2-\gamma-\lambda}{\lambda-\gamma}$ join the cartel.
		\item If $\lambda \geq 1$, no influencer joins the cartel.
	\end{enumerate}
\end{proposition}
The proposition implies that only cartels with engagement requirement $\Lambda \leq \frac{\pi}{2}$ are feasible. Compared to the case with homogeneous influencers, high engagement requirement $\Lambda$ introduces an additional distortion: influencers with high reach may choose not to join the cartel. This is because, for sufficiently high engagement requirements $\Lambda$, influencers face a trade-off when deciding whether to join: on the one hand, they are asked to engage more than they would prefer in isolation (i.e., $\gamma \cos(\Delta) - C(\Delta) < 0$), but on the other hand they receive additional engagement from other cartel members in return. The key insight is that the first component (cost) is proportional to the influencer's own reach $R_t$, whereas the second (benefit) is proportional to the average reach of a cartel member. Hence, influencers with the highest reach are the first to stay out of the cartel.
Next, let us consider an optimal cartel. Using \Cref{P:cartel_het}, we know that the optimal cartel must have $\lambda < 1$ (or equivalently, $\Lambda < \frac{\pi}{2}$). Therefore, depending on the level of the externality parameter $\gamma$, there are three cases to consider. First, suppose that the socially optimal engagement level is such that $\lambda^C < \gamma$, i.e., all influencers would join such cartel. It is straightforward to see that, in this case, the expected payoff of a cartel member, $\expect U^C(R_t)$, coincides with the social welfare function $W$. Hence, the optimal engagement coincides with the socially optimal one, which is such that $\Lambda^C = \Lambda^S = \frac{\pi}{4}$, or equivalently $\lambda^C = \lambda^S = \tan \left( \frac{\pi}{8} \right) = \sqrt{2} - 1 \approx 0.414$. For illustration, see \Cref{F:cartel_payoffs} (the green line, corresponding to $\gamma = \frac{1}{2}$).
\begin{figure}
	\begin{center}
		\includegraphics[trim=30pt 10pt 45pt 15pt,clip,width=0.8\linewidth]{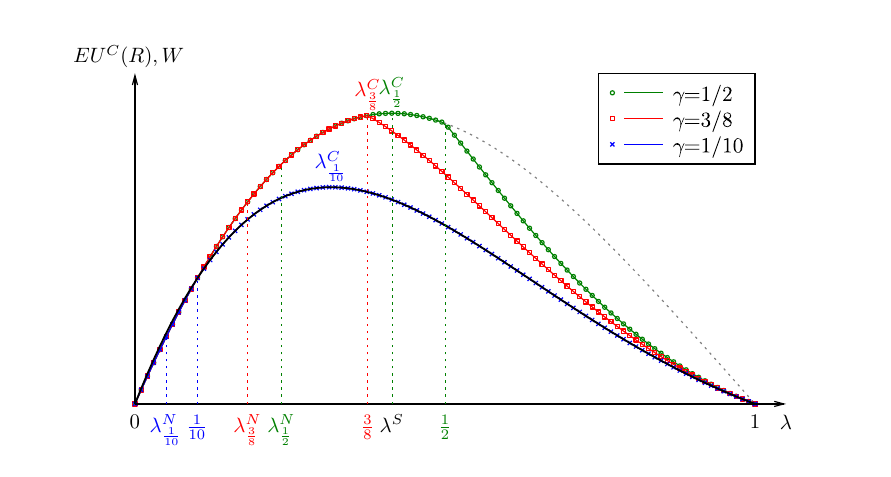}
		\caption{Welfare as a function of engagement requirement}
		\appfigure{Welfare as a function of engagement requirement}
		\label{F:cartel_payoffs}
	\end{center}
	\footnotesize{Notes: $\lambda^S$ denotes the socially optimal engagement level, $\lambda^N_{\gamma}$, the natural engagement, and $\lambda^C_{\gamma}$ the optimal cartel engagement corresponding to $\gamma$. The solid lines with markers represent expected payoffs for cartel members, and the dashed line is social welfare function if all influencers would join.}
\end{figure}
This argument fails when $\gamma$ is small, i.e., $\gamma < \lambda^S$. In this case, the socially optimal cartel is not feasible, because the highest-reach influencers would not join the cartel. We then need to augment the objective of the cartel, i.e., the mean payoff for a cartel member by taking expectation conditional on $R_t \leq \oR$ defined in \Cref{P:cartel_het}, and the expression becomes
\begin{align}
	\expect[ \expect U^C(R_t)\mid R_t \leq \oR]
	&=
	\frac{2 \lambda (1- \lambda)}{\pi (\lambda^2+1)} \frac{2-\gamma-\lambda}{1-\gamma},
	\label{E:Wcm_onlysome}
\end{align}
This expression has a unique maximizer $\lambda^*(\gamma)$ in $(0,1)$. If $\lambda^*(\gamma) > \gamma$, then the optimal cartel engagement is $\lambda^C = \lambda^*(\gamma)$. This is illustrated by the blue line in \Cref{F:cartel_payoffs} (the case $\gamma=\tfrac{1}{10}$).
The final case is when the externality parameter $\gamma$ is intermediate, i.e., $\lambda^*(\gamma) < \gamma < \lambda^S$. In this case, the optimal engagement lies at the boundary between the two regions, i.e., $\lambda^C = \gamma$. This case is illustrated by the red line in \Cref{F:cartel_payoffs} (the case $\gamma = \frac{3}{8}$).
The following corollary summarizes the observations above.
\begin{corollary} \label{C:cartel_welfare}
	Depending on $\gamma$, we have one of three cases:\footnote{
		The critical values are
		$\gamma^{\text{inc}} \approx 0.344$
		and
		$\lambda^S = \sqrt{2}-1 \approx 0.414$.
	}
	\begin{enumerate}
		\item If $\gamma \geq \lambda^S$, then the optimal cartel is the socially optimal cartel with $\lambda^C = \lambda^S$.
		\item If $\gamma^{\text{inc}} \leq \gamma < \lambda^S$, then the social optimum is not feasible as a cartel outcome. The optimal cartel is the one with the largest engagement level such that all influencers join the cartel, which is $\lambda^C = \gamma$.
		\item If $\gamma < \gamma^{\text{inc}}$, then the social optimum is not feasible as a cartel outcome. The optimal cartel, $\lambda^C = \lambda^*(\gamma)$, involves some influencers staying out of the cartel.
	\end{enumerate}
\end{corollary}
This version of our model can also shed light on why influencer cartels often impose entry requirements in practice. A typical requirement is to have at least some minimum number of followers, ranging from 1,000 to 100,000 in our sample. We saw that the cost of joining the cartel depends on the influencer’s own reach, while the benefit depends on the average reach of a cartel member. By imposing a minimum reach requirement, the cartel increases the average reach, making it more appealing to influencers with higher reach. The combination of these effects raises the average reach and benefits all cartel members. Therefore, we would expect the entry requirement to increase the average benefits for cartel members. On the other hand, excluding influencers with low reach means that fewer can join the cartel, which may reduce overall social welfare. The following proposition confirms this intuition.
\begin{proposition} \label{P:entry_requirement}
	Suppose that, in addition to an engagement requirement $\Lambda > 0$, the cartel imposes an entry requirement $\uR > 1$, so that only influencers with $R_t \geq \uR$ are eligible to join.
	The expected payoff of a cartel member, $\expect[U^C(R_t) \mid \uR \leq R_t]$, is proportional to $\uR$, and social welfare is proportional to $\uR^{-1}$.
\end{proposition}
The cartel may, therefore, choose to restrict eligibility, because such a restriction raises cartel members' welfare. However, there is a downside---the restriction reduces the number of cartel members, and this effect can be large enough to reduce overall welfare. In our model, the optimal minimum reach is infinitely large, but in practice, a very high reach requirement would be impractical.\footnote{Nevertheless, in our sample there is a cartel that requires at least 100,000 followers, and this cartel has a large number of members.}
\subsubsection{Advertising Dominates Other Incentives}
Let us now consider the other extreme, where advertising revenue dominates all other incentives. Specifically, in this case, the normalized value of joining a cartel with engagement requirement $\Lambda$ becomes
\begin{equation}
\lim_{v \to \infty} \frac{\expect U^C(R_t)}{v}
= \frac{\Lambda}{\pi}
\expect[ R_{t'} \cos(\Delta) \mid a_{t'} = 1]
=
\frac{\sin(\Lambda)}{\pi}
\expect[ R_{t'} \mid a_{t'} = 1].
\end{equation}
This expression is non-negative for all $\Lambda$, so any cartel is feasible and all influencers join, implying $\expect[ R_{t'} \mid a_{t'} = 1] = 2$. The expression is maximized when $\sin(\Lambda) = 1$, hence $\Lambda^C = \frac{\pi}{2}$, which also corresponds to the socially optimal cartel in this case.
Compared to the no-advertising case above, we do not have the secondary distortion of high-reach influencers staying out of the cartel. The reason is that, since advertising revenue dominates all other incentives, all influencers become primarily interested in the attention they receive from the cartel, i.e., they care much more about the average reach than about their own reach.
Note that a minimum entry requirement would still increase the expected payoff of a cartel member, as $\expect[R_{t'} \mid R_{t'} \geq \uR] = 2\uR$, so it would be optimal to introduce a high entry requirement. However, since the share of influencers satisfying this requirement is $Pr(R_{t'} \geq \uR) = \frac{1}{\uR^2}$, social welfare would in fact decrease with $\uR$.
\subsection{Both Natural and Cartel Engagement}
Again, we combine natural engagement and cartel engagement, assuming that mass $1-\varepsilon$ of influencers choose natural engagement, whereas mass $\varepsilon$ belong to an infinite number of small cartels, which choose their engagement requirements independently. We now focus only on the case $v \to \infty$. As implied by the discussion above, in this case there is no substantial difference from the homogeneous-reach model, as payoffs depend only on the average reach rather than the influencer’s own reach.
Therefore, the results from \Cref{P:grandcartels} and \Cref{C:grandcartels} apply here as well: general cartels with maximal engagement level $\Lambda_i^C = \pi$ are not only feasible but also optimal, maximizing cartel members’ payoffs. These cartels are harmful to everyone, as they reduce welfare, and all influencers would be better off if there were fewer such cartels.

\newpage
\section{Online Appendix: Data Collection} \label{A:DataCollection}
\addcontentsline{appsec}{appendixsection}{\protect\numberline{\thesection}Data Collection}
\subsection{Telegram Cartel History}
We collected Telegram cartel interaction history for 9 cartels: 6 general cartels (1K, 5K, 10K, 30K, 50K, 100K) and 3 topic cartels (fashion \& beauty, health \& fitness, travel \& food). The 9 cartels were formed the earliest in August 2017 (10K and 50K) and the latest in February 2018 (5K). We downloaded the data in July 2020. In June 2020 all cartels had new posts.
The Telegram cartel interaction history consists of three pieces of information: Telegram username, Instagram post shortcode, and time. The interaction history tells us which Telegram user, added when, and which Instagram post to the cartel. According to the cartel rules, this information allows to determine which cartel member has to comment and like which post. This is because one has to comment and like at least five posts by other users directly preceding one’s own.
\subsection{Mapping Telegram Posts to Instagram Users}
The Telegram cartels included 220,893 unique Instagram posts that we were able to map to 21,068 Instagram users. Specifically, the Telegram cartels included 316,462 unique Instagram posts altogether. Some posts are posted multiple times and/or multiple cartels, the 316,462 unique Instagram posts were posted in total 527,498 times. The cartel interaction files don’t include the Instagram username of the author of the Instagram post. We mapped the Instagram posts included in cartels to Instagram users using the following interactive procedure. For the first Instagram post in the cartel of each Telegram user, we searched for the post on Instagram to learn the Instagram username of the post’s author. Then we obtained from CrowdTangle the full list of all the Instagram posts of that Instagram username and matched those to the posts in the cartel. We checked the remaining unmatched posts in the cartel one by one until we either found a match for it on Instagram or determined that the post had been deleted on Instagram or made private. In this way, we were able to determine the Instagram usernames of 70\% of the posts in the cartels, altogether 21,068 Instagram usernames.
Of the 21,068 Instagram users, 22\% of users had posted in both topic and general cartels. Altogether, 11,158 users had posted in topic cartels and 14,566 users in general. This includes 4,656 users who had posted in both. Hence, the total number of users equals 11,158 + 14,566 - 4,656 = 21,068.
From CrowdTangle, for all the 21,068 Instagram users, we obtained the history of all their Instagram posts. This data included the time of the post, and the text of the post including the hashtags. The data was downloaded from August 19 to September 16, 2021.
\subsection{Instagram Comments}
For each Instagram user, for their first post in cartels (no matter which cartel), we collected information on who commented on the post. Our goal was to learn who engaged with the post and compare natural (non-cartel) engagement to that obtained via cartels. Therefore, we did not focus on the comment but instead only on the username that posted the comment. We restricted attention to the commenters on the post itself, and excluded commenters who commented on a comment.
We focused on each user’s first post in cartels to minimize the possibility that involvement in cartels had affected engagement. However, when the first post did not have enough information for the analysis, that is, when for the first post none of the cartel members who were required to comment existed anymore, then we focused on the second (if the second post existed) and so on. We did not require that the cartel members actually commented, only that the users still existed. For 18 users no post existed that satisfied the requirement reducing the sample to 21,050 users. Among the remaining 21,050 users, for 99.8\% (20,999), it was of their actual first posts. To simplify the exposition, in going forward, we call all the first posts satisfying the requirement, simply the first posts.
We used Apify to collect the comments. It allows access to the comments that are available without logging in to Instagram and provides only up to 50 comments for each post. The data included the username of the commenter and the text of the comment. While we don't use the text of the comment in our main analysis, we do analyze it in \Cref{A:TablesFigures}. The comments were downloaded in January 2024.
We were able to collect comments only for 16,630 posts, which is 79.0\% of the 21,050 first posts. We could not collect comments for all the first posts, because these posts either did not have any comments but mostly because we attempted to collect these comments more than two years after collecting posts itself, and in two years these users or their posts were either deleted or made private. Of the posts for which we were able to collect comments, some had no non-cartel comments, leaving us with 16,386 posts. Hence, we were able to find non-cartel comments for 77.8\% of the total 21,050 first posts. For both topic and general cartels the percentage of first posts for which we got non-cartel comments was similar, 78\%.
\subsection{Random Non-Cartel Commenter on Each Cartel Member’s First Cartel Post}
For each cartel member's first post in cartels, we used a random number generator and picked a random non-cartel user who had commented on the post. We picked these random non-cartel users for 16,386 posts. Some of the randomly picked non-cartel commenting users were the same across posts. Hence, we were left with 14,490 unique non-cartel commenting users.
For these 14,490 non-cartel users, we collected their information about their number of public posts. We collected this information using Apify in January 2024. Of these users, 24 didn’t exist anymore. So that we were left with 14,466 non-cartel users. Of those, 3,049 (21.0\%) had no public posts. We had to exclude those user because they had no information we could analyze. So that we were left with 11,417 public non-cartel users. We further limited the sample to non-cartel users who had at least 10 public posts and this restriction reduced the sample to 10,394 non-cartel users.
For these random non-cartel users, we obtained the history of all their Instagram posts from CrowdTangle. We did this to calculate authors' similarity to the non-cartel users who commented on the post. For these non-cartel users, the data was downloaded from CrowdTangle in January 2024. We were able to get the history only for 10,280 non-cartel users (99\%). For the remaining 114 usernames either they had changed the username, made the account private or deleted it.  Furthermore, we learned that 551 (5\%) of the 10,280 non-cartel users were associated with cartel members as they had posted at least one post associated with a cartel member. The association can happen as Instagram allows posts to be associated with multiple users (this is different from tagging a user) or it can happen when users change usernames. We excluded those 551 non-cartel users from our sample, while keeping the corresponding cartel members. That reduced the sample of non-cartel users to 9,729. These 9,729 non-cartel users mapped to 10,683 first posts because, as said above, some of these non-cartel users were commenting on multiple posts.
\subsection{Photos from First Posts}
We collected a photo from a single Instagram post for each user. For cartel members, we use the first post each cartel member posted to the cartels. For non-cartel members, we select their closest post within a symmetric time window to the cartel member’s post they commented on. The photos were downloaded in February and March in 2024. We were able to get the photos for only 16,693 (79\%) cartel members and 9,269 (95\%) non-cartel users. The reason why we did not get photos for 21\% of the cartel members is the same as for the comments, that in the two years the posts were either deleted or made private. Similarly, due the delay, we did not get photos for 5\% of the non-cartel users.
In the end, we were left with 5244 authors whose first post was in general cartels and 3751 authors whose first post was in topic cartels (including 488 authors that appear in both categories). Therefore, the percentage reduction due to data limitations was similar for general cartels and topic cartels. Thus we have no reason to expect that the data collection affected different cartel types differently.

\newpage
\section{Online Appendix: Additional Empirical Results} \label{A:TablesFigures}
\addcontentsline{appsec}{appendixsection}{\protect\numberline{\thesection}Additional Empirical Results}
\setlength{\tabcolsep}{2pt}
\begin{table}[h!]
    \begin{center}
        \begin{footnotesize}
            \caption{Summary statistics of cartel versus non-cartel comments}
            \apptable{Summary statistics of cartel versus non-cartel comments}
            \label{T:SumStatComments}
            \hspace*{-0.5cm}
            \begin{tabular}{lccc ccc ccc}
                \hline
 & (1) & (2) \\  & Commenter not in cartel & Commenter in cartel \\ \cmidrule(lr){2-3}
& \multicolumn{2}{c}{Panel A: General cartels} \\
Comment's length (in words)                                 &    4.574&    4.803\\
Share of negative comments                                  &    0.011&    0.008\\
Comment's similarity to the post's photo                    &    0.221&    0.227\\
Observations (comments)                                     &     4432&    10490\\
 \cmidrule(lr){2-3}
& \multicolumn{2}{c}{Panel B: Topic cartels} \\
Comment's length (in words)                                 &    4.942&    5.274\\
Share of negative comments                                  &    0.010&    0.006\\
Comment's similarity to the post's photo                    &    0.220&    0.228\\
Observations (comments)                                     &     3399&     8463\\
 \hline

            \end{tabular}
        \end{footnotesize}
    \end{center}
    \footnotesize{Notes:
        Negative comments are classified using VADER sentiment analysis model \citep{hutto_vader:_2014}. While only about 1\% of comments are classified as negative, most of these are misclassified as negative due to the use of slang (for example:  “hell of a view”, “killing this look”, “sick style keep it up”) or don't criticize the post (for example: “this is so scary” “I hate when it happens” “that looks so brutal”).
        Comments similarity to the photo in the post is calculated using the CLIP model (see \Cref{S:MeasuringEngagementQuality}).
        Panel A includes posts submitted only to general cartels, and panel B, only to topic cartels.
        According to the t-test, all differences between columns 1 and 2, except for the share of negative comments are statistically significant at 1 percent level.
    }
\end{table}
\setlength{\tabcolsep}{6pt}
\begin{table}[h!]
    \centering
    \begin{footnotesize}
        \caption{Most representative hashtags for each LDA topic}
        \apptable{Most representative hashtags for each LDA topic}
        \label{T:lda_topic_top_words}
        \begin{tabular}{cc p{11cm}}
            \hline
            \hline
Topic number & Topic label & Most representative hashtags \\ \hline \hline
Topic 1&Fitness&\#fitness \#gym \#workout \#motivation \#fitnessmotivation \#fitfam \#fit \#bodybuilding \#health \#training\\ \hline
Topic 2&Beauty&\#art \#music \#makeup \#artist \#photography \#beauty \#makeupartist \#design \#mua \#yoga\\ \hline
Topic 3&Fashion&\#fashion \#ootd \#instagood \#style \#photooftheday \#fashionblogger \#picoftheday \#beautiful \#photography \#beauty\\ \hline
Topic 4&Food&\#foodie \#foodporn \#food \#instafood \#liketkit \#foodphotography \#foodstagram \#foodblogger \#yummy \#delicious\\ \hline
Topic 5&Entrepreneur&\#motivation \#entrepreneur \#success \#business \#inspiration \#luxury \#quotes \#motivationalquotes \#mindset \#entrepreneurship\\ \hline
Topic 6&Travel&\#travel \#travelgram \#travelphotography \#wanderlust \#travelblogger \#nature \#photography \#instatravel \#photooftheday \#beautifuldestinations\\ \hline
\hline

        \end{tabular}
    \end{footnotesize}
\end{table}
\begin{table}[h!]
    \begin{center}
        \begin{footnotesize}
            \caption{Sample construction}
            \apptable{Sample construction}
            \label{T:SampleConstruction}
            \begin{tabular}{l c}
                \hline
 & Number of authors \\ \hline \multicolumn{2}{c}{Panel A: Number of authors in the regression sample} \\
Total number of authors in the cartel                       &    21068\\
.. with cartel commenters                                   &    21050\\
.. and with non-cartel commenters                           &    10683\\
.. and author has embeddings                                &    10171\\
.. and cartel commenter has embeddings                      &    10022\\
.. and non-cartel commenter has embeddings                  &     8507\\
 \hline
 \multicolumn{2}{c}{Panel B: Number of authors in the LDA sample} \\
Total number of authors in the cartel                       &    21068\\
.. with cartel commenters                                   &    21050\\
.. and with non-cartel commenters                           &    10683\\
.. and author has LDA                                       &     8936\\
.. and cartel commenter has LDA                             &     8835\\
.. and non-cartel commenter has LDA                         &     6654\\
 \hline

            \end{tabular}
        \end{footnotesize}
    \end{center}
\end{table}
\setlength{\tabcolsep}{12pt}
\begin{table}[h!]
    \begin{center}
        \begin{footnotesize}
            \caption{Summary statistics of authors in cartel included/excluded from the sample}
            \apptable{Summary statistics of authors in cartel included/excluded from the sample}
            \label{T:SumStatMainSampleVsExcluded}
            \begin{tabular}{l cc}
                \hline
 & (1) & (2) \\ & \multicolumn{2}{c}{Authors included vs excluded from the sample} \\ & Excluded & Included \\ \hline
 & \multicolumn{2}{c}{Panel A: Regression sample} \\ \cmidrule(lr){2-3}
Number of posts per user                                    &    583.0&    735.7\\
Number of posts in cartel per user                          &      9.5&     12.0\\
Number of likes per post                                    &    777.9&    699.8\\
Number of comments per post                                 &     28.8&     29.5\\
Overperforming score                                        &     -5.7&     -4.6\\
\% of disclosed sponsored posts                             &      0.8&      1.0\\
\% of users with disclosed sponsored posts                  &     25.9&     34.4\\
Observations (authors)                                      &    12561&     8507\\
 \hline
 & \multicolumn{2}{c}{Panel B: LDA sample} \\ \cmidrule(lr){2-3}
Number of posts per user                                    &    615.0&    709.0\\
Number of posts in cartel per user                          &     10.0&     11.6\\
Number of likes per post                                    &    790.3&    651.1\\
Number of comments per post                                 &     29.0&     29.2\\
Overperforming score                                        &     -5.6&     -4.5\\
\% of disclosed sponsored posts                             &      0.7&      1.1\\
\% of users with disclosed sponsored posts                  &     26.8&     34.7\\
Observations (authors)                                      &    14414&     6654\\
 \hline

            \end{tabular}
        \end{footnotesize}
    \end{center}
    \footnotesize{Notes:
        Average statistics per post are calculated by first taking averages across the posts for each user, and then calculating averages across users. \textit{Overperforming score} is calculated by CrowdTangle and measures the number of comments and likes relative to the user's previous 100 posts conditional on posts' type and age. Disclosed sponsored posts are identified following \cite{ershov_how_2025} based on disclosure hashtags (\#ad, \#sponsored, \#paidpartnership, \#brandedcontent, \#gifted, \#paid, \#partnership, \#promotion, \#branded, \#sponsoredby, \#paidby).
        According to the t-test, all differences between columns 1 and 2, except for \textit{Number of likes per post} in Panel A, \textit{Number of comments per post} and \textit{Overperforming score} in Panels A and B, are statistically significant at 1 percent level.}
\end{table}
\setlength{\tabcolsep}{6pt}
\begin{table}[h!]
    \begin{center}
        \begin{footnotesize}
            \caption{Summary statistics of authors in different types of cartel}
            \apptable{Summary statistics of authors in different types of cartel}
            \label{T:SumStatGenVsTop}
            \begin{tabular}{l ccc}
                \hline
 & (1) & (2) & (3)\\ & \multicolumn{3}{c}{Authors in cartels} \\ & General & Topic & Both\\ \cmidrule(lr){2-4}
Number of posts per user                                    &    860.8&    570.1&    624.3\\
Number of posts in cartel per user                          &     13.3&     10.0&     12.7\\
Number of words per post                                    &     51.4&     59.3&     58.9\\
Number of hashtags per post                                 &     14.0&     15.0&     14.4\\
Number of likes per post                                    &    829.1&    501.2&    768.2\\
Number of comments per post                                 &     32.1&     25.4&     31.0\\
Overperforming score                                        &     -6.6&     -1.8&     -3.7\\
\% of disclosed sponsored posts                             &      1.2&      0.8&      0.6\\
\% of users with disclosed sponsored posts                  &     38.2&     29.3&     30.9\\
Observations (authors)                                      &     4756&     3263&      488\\
 \hline

            \end{tabular}
        \end{footnotesize}
    \end{center}
    \footnotesize{Notes:
        The table includes authors in the main regression sample.
        Average statistics per post are calculated by first taking averages across the posts for each user, and then calculating averages across users. \textit{Overperforming score} is calculated by CrowdTangle and measures the number of comments and likes relative to the user's previous 100 posts conditional on posts' type and age. Disclosed sponsored posts are identified following \cite{ershov_how_2025} based on disclosure hashtags (\#ad, \#sponsored, \#paidpartnership, \#brandedcontent, \#gifted, \#paid, \#partnership, \#promotion, \#branded, \#sponsoredby, \#paidby). The sample includes all posts for each user, except for word and hashtag counts, which are calculated using the data for the main regression analysis (as described in \Cref{S:MeasuringEngagementQuality}): word counts use a single post per user, and hashtag counts use 100 posts per user.
        According to the t-test, all differences between columns 1 and 2 are statistically significant at 1 percent level. }
\end{table}
\begin{table}[h!]
    \begin{center}
        \begin{footnotesize}
            \caption{Summary statistics of users in cartels versus not in cartels}
            \apptable{Summary statistics of users in cartels versus not in cartels}
            \label{T:SumStatCartelVsNot}
            \begin{tabular}{lcc}
                \hline
 & (1) & (2) \\  & User not in cartel & User in cartel \\  \cmidrule(lr){2-3}
Number of posts per user                                    &    769.1&    730.9\\
Number of words per post                                    &     48.9&     53.8\\
Number of hashtags per post                                 &     13.7&     14.4\\
Number of likes per post                                    &    537.3&    730.8\\
Number of comments per post                                 &     31.1&     29.2\\
Overperforming score                                        &    -10.3&     -6.1\\
\% of disclosed sponsored posts                             &      0.8&      1.0\\
\% of users with disclosed sponsored posts                  &     28.6&     33.7\\
Observations (users)                                        &     7847&    12088\\
 \hline

            \end{tabular}
        \end{footnotesize}
    \end{center}
    \footnotesize{Notes:
        The table includes users (authors and commenters) in the main regression sample.
        Average statistics per post are calculated by first taking averages across the posts for each user, and then calculating averages across users. \textit{Overperforming score} is calculated by CrowdTangle and measures the number of comments and likes relative to the user's previous 100 posts conditional on posts' type and age. Disclosed sponsored posts are identified following \cite{ershov_how_2025} based on disclosure hashtags (\#ad, \#sponsored, \#paidpartnership, \#brandedcontent, \#gifted, \#paid, \#partnership, \#promotion, \#branded, \#sponsoredby, \#paidby). The sample includes all posts for each user, except for word and hashtag counts, which are calculated using the data for the main regression analysis (as described in \Cref{S:MeasuringEngagementQuality}): word counts use a single post per user, and hashtag counts use 100 posts per user.
        According to the t-test, all differences, except for \textit{Number of comments per post}, are statistically significant at 5 percent level.}
\end{table}
\begin{table}[h!]
    \begin{center}
        \begin{footnotesize}
            \caption{Summary statistics of cartel members before and after joining the cartel}
            \apptable{Summary statistics of cartel members before and after joining the cartel}
            \label{T:SumStatCartelBeforeVsAfterJoining}
            \begin{tabular}{lcc}
                \hline
 & (1) & (2) \\  & Before & After \\ & \multicolumn{2}{c}{first joining a cartel} \\  \cmidrule(lr){2-3}
Number of likes per post                                    &    501.6&    960.2\\
Number of comments per post                                 &     21.0&     43.0\\
Overperforming score                                        &     -8.4&      0.5\\
\% of disclosed sponsored posts                             &      0.6&      1.7\\
\% of users with disclosed sponsored posts                  &     21.9&     28.0\\
Observations (number of cartel members)                     &     8447&     8447\\
 \hline

            \end{tabular}
        \end{footnotesize}
    \end{center}
    \footnotesize{Notes:
        The table includes authors in the main regression sample who have posts both before and after joining the cartel.
        Average statistics per post are calculated by first taking averages across the posts for each user before and after joining the cartel, and then calculating averages across users. \textit{Overperforming score} is calculated by CrowdTangle and measures the number of comments and likes relative to the user's previous 100 posts conditional on posts' type and age. Disclosed sponsored posts are identified following \cite{ershov_how_2025} based on disclosure hashtags (\#ad, \#sponsored, \#paidpartnership, \#brandedcontent, \#gifted, \#paid, \#partnership, \#promotion, \#branded, \#sponsoredby, \#paidby).
        According to the t-test, all differences are statistically significant at 1 percent level.}
\end{table}
\clearpage
\begin{figure}[!ht]
    \begin{center}
        \includegraphics[width=0.5\textwidth]{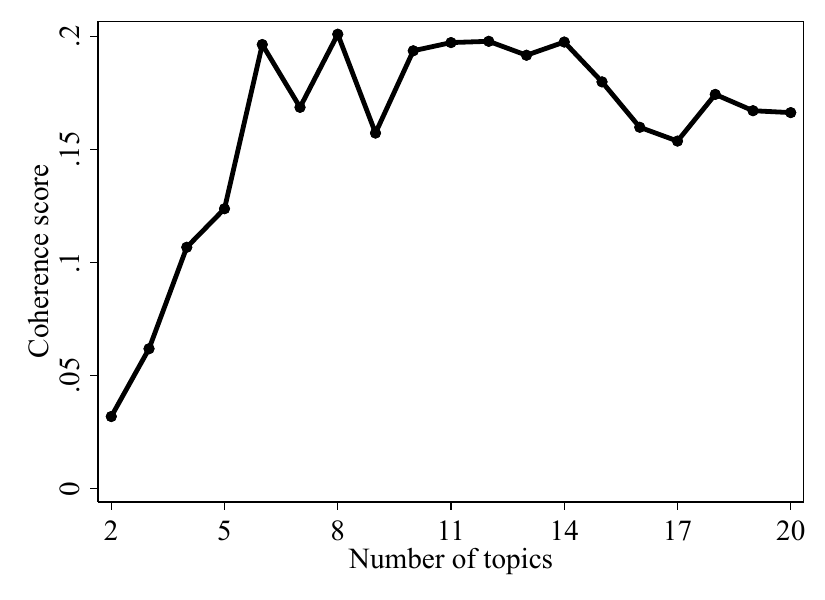}
        \caption{LDA model coherence scores by the number of topics}
        \appfigure{LDA model coherence scores by the number of topics}
        \label{F:LDA_coherence_scores}
    \end{center}
    \footnotesize{Notes: The figure presents the Normalized Pointwise Mutual Information (NPMI) coherence score (y-axes) for each LDA model, comparing models with the number of topics (x-axes) ranging from two to twenty. The score is calculated for the top ten words in each topic with a window size of five words.
    }
\end{figure}
\begin{figure}[h!]
    \begin{center}
        \begin{subfigure}[t]{0.4\textwidth}
            \includegraphics[width=1\textwidth]{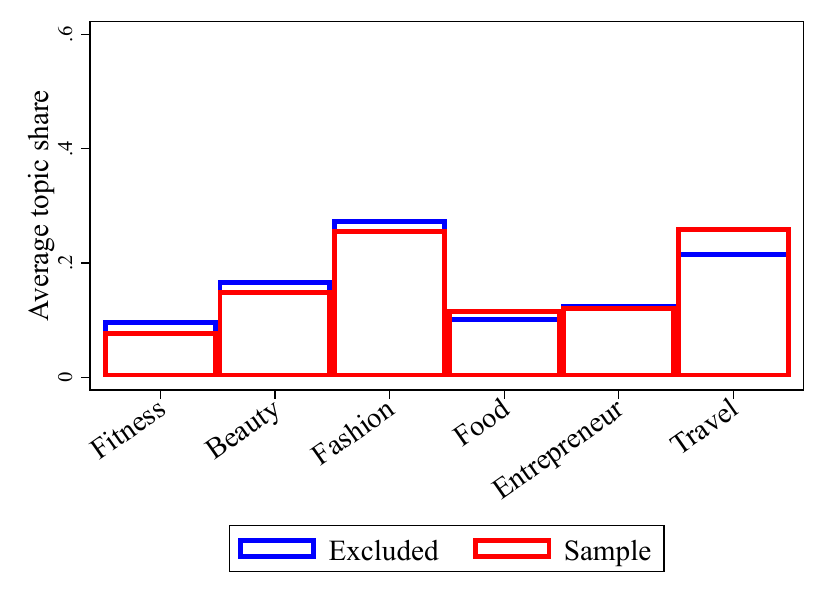}
            \caption{General cartels}
        \end{subfigure}
        \hspace{0.05\textwidth}
        \begin{subfigure}[t]{0.4\textwidth}
            \includegraphics[width=1\textwidth]{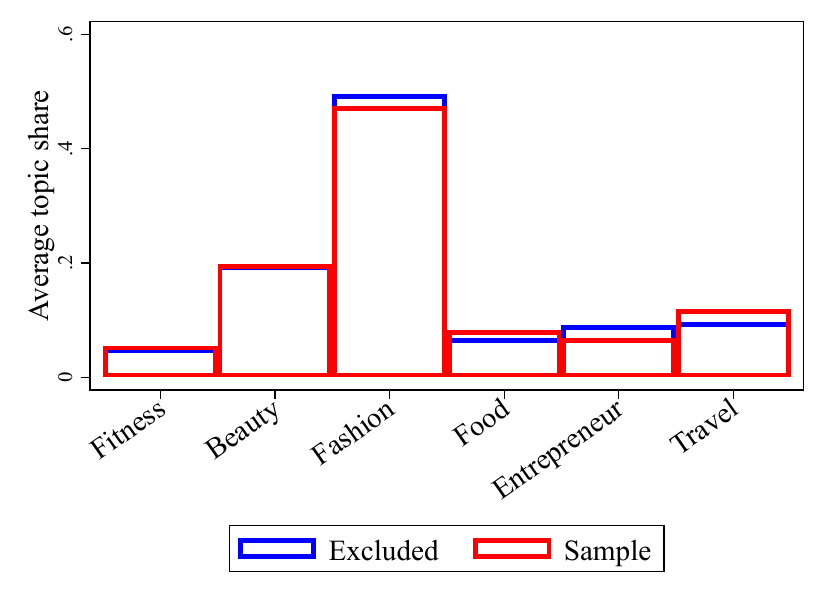}
            \caption{Fashion \& beauty cartel}
        \end{subfigure}
        \\
        \begin{subfigure}[t]{0.4\textwidth}
            \includegraphics[width=1\textwidth]{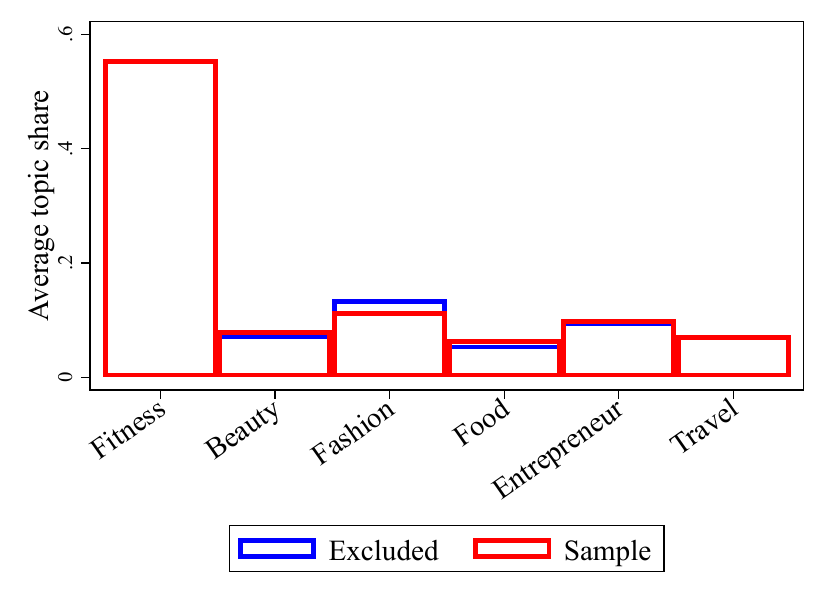}
            \caption{Fitness \& health cartel}
        \end{subfigure}
        \hspace{0.05\textwidth}
        \begin{subfigure}[t]{0.4\textwidth}
            \includegraphics[width=1\textwidth]{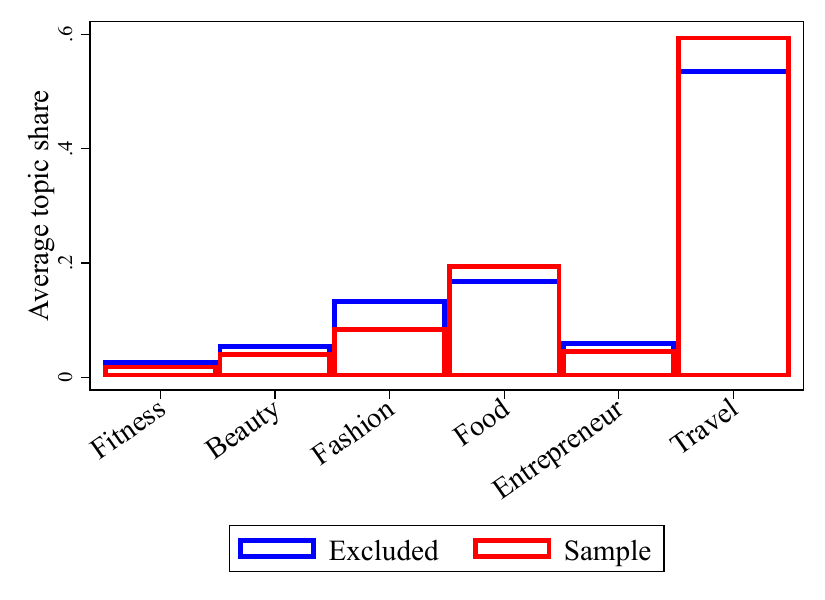}
            \caption{Travel \& food cartel}
        \end{subfigure}
        \caption{Authors' topics: regression sample versus excluded}
        \appfigure{Authors' topics: regression sample versus excluded}
        \label{F:LDA_pod_authors_excluded_sample_regr}
    \end{center}
\end{figure}
\begin{figure}[h!]
    \begin{center}
        \begin{subfigure}[t]{0.4\textwidth}
            \includegraphics[width=1\textwidth]{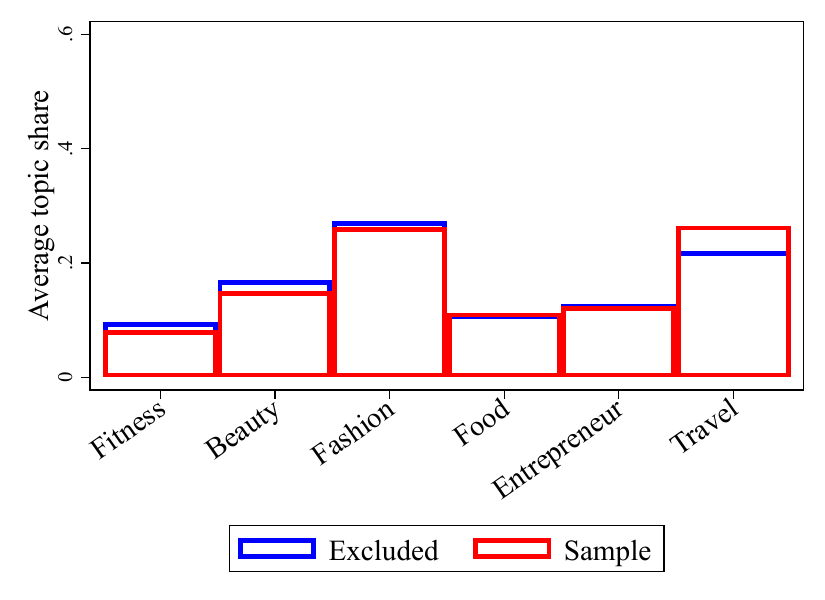}
            \caption{General cartels}
        \end{subfigure}
        \hspace{0.05\textwidth}
        \begin{subfigure}[t]{0.4\textwidth}
            \includegraphics[width=1\textwidth]{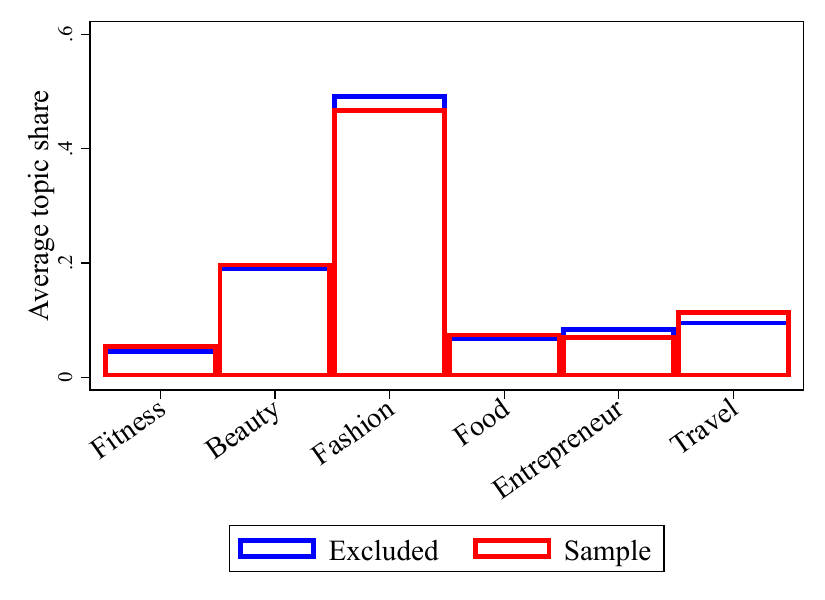}
            \caption{Fashion \& beauty cartel}
        \end{subfigure}
        \\
        \begin{subfigure}[t]{0.4\textwidth}
            \includegraphics[width=1\textwidth]{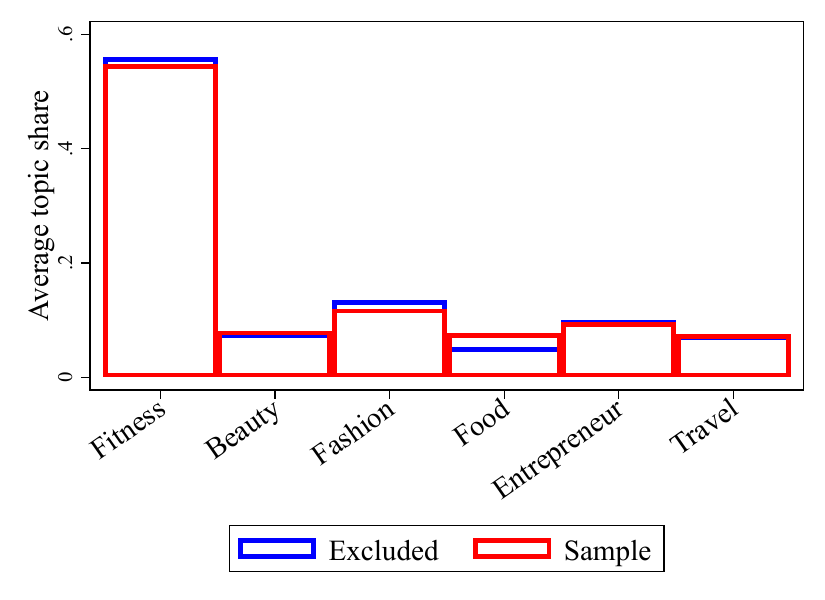}
            \caption{Fitness \& health cartel}
        \end{subfigure}
        \hspace{0.05\textwidth}
        \begin{subfigure}[t]{0.4\textwidth}
            \includegraphics[width=1\textwidth]{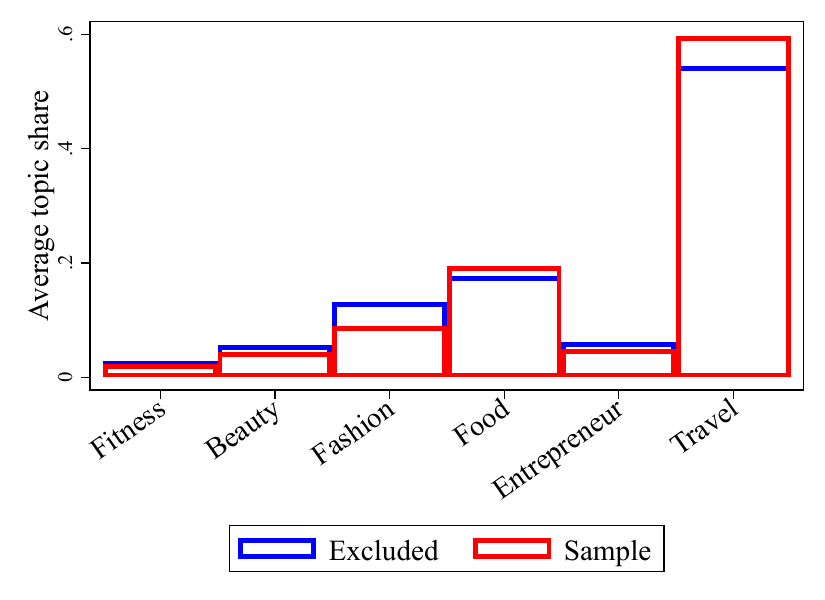}
            \caption{Travel \& food cartel}
        \end{subfigure}
        \caption{Authors' topics: LDA sample versus excluded}
        \appfigure{Authors' topics: LDA sample versus excluded}
        \label{F:LDA_pod_authors_excluded_sample_LDA}
    \end{center}
\end{figure}
\clearpage
\begin{figure}[h!]
    \begin{center}
        \begin{subfigure}[t]{0.47\textwidth}
            \includegraphics[width=1\textwidth]{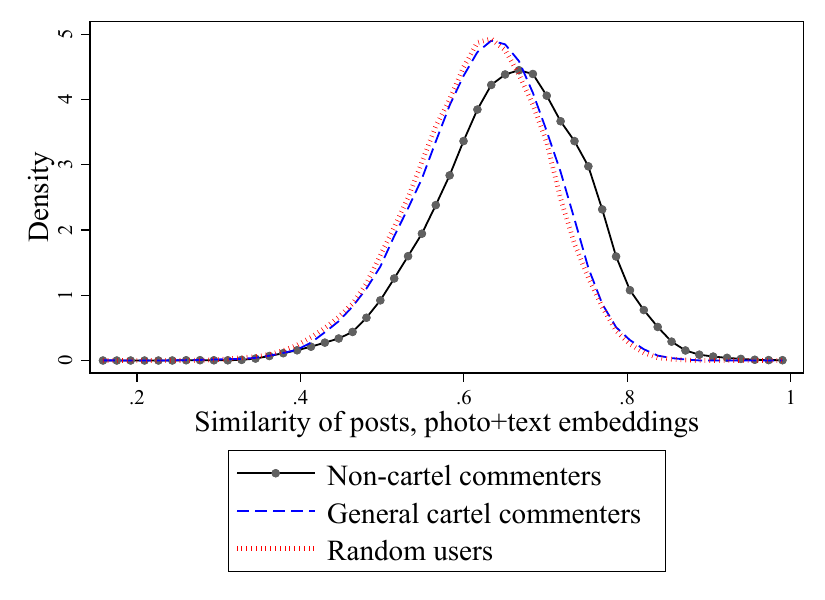}
            \caption{General cartels, posts' similarity}\label{F:Similarity_cs_clip_tp_gen}
        \end{subfigure}
        \hfill
        \begin{subfigure}[t]{0.47\textwidth}
            \includegraphics[width=1\textwidth]{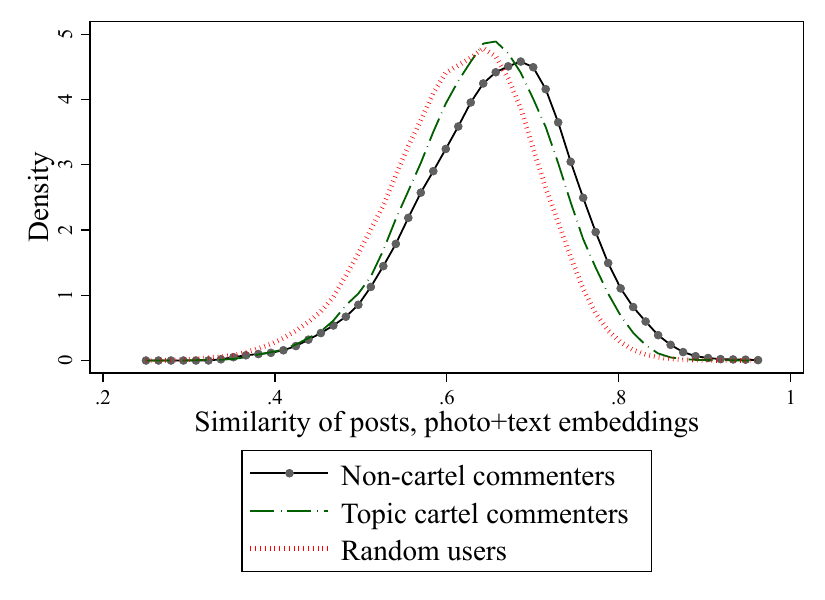}
            \caption{Topic cartels, posts' similarity}\label{F:Similarity_cs_clip_tp_top}
        \end{subfigure}
        \caption{Probability density of authors' similarity to commenters and random users}
        \appfigure{Probability density of authors' similarity to commenters and random users}
        \label{F:SimilarityClip}
    \end{center}
    \footnotesize{Notes:
        The figures present kernel density estimates using the Epanechnikov kernel function of authors' cosine similarity to non-cartel commenters (grey line with solid circle markers), to random users (red dotted line), to general cartel commenters (blue dashed line on \Cref{F:Similarity_cs_clip_tp_gen}), and to topic cartel commenters (green dashed and dotted line on \Cref{F:Similarity_cs_clip_tp_top}). The cosine similarity is calculated as the similarity of posts using the photo and text embeddings from the CLIP model.
    }
\end{figure}

\clearpage
\section{Online Appendix: Robustness Analysis} \label{A:EmpiricalRobustness}
\addcontentsline{appsec}{appendixsection}{\protect\numberline{\thesection}Robustness Analysis}
This Online Appendix describes that the regression results are robust to alternative ways to construct outcome variables and alternative samples.
\Cref{T:RegrRobust_labse_random_monthly,T:RegrRobust_labse_text,T:RegrRobust_clip_photos_text}  presents results where first, text embeddings (outcome variables in columns 1--3 in \Cref{T:RegrMain}) are constructed based on (A) random sample of posts; (B) all posts in 2017--2020; (C) the whole text instead of hashtags; second, embeddings from the CLIP model (outcome variables in columns 4--6 of \Cref{T:RegrMain}) are constructed based on (A) only the photos, (B) only the text. To alleviate the concern that, instead of the average match quality, advertisers care about sufficiently high matches, in \Cref{T:RegrRobust_similarity_above_threshold}, outcome variables are indicators for whether cosine similarity is above the 75th or 90th percentile. It also presents the results where the sample includes only the cartel commenters for whom we observe that they actually commented (\Cref{T:RegrRobust_actuallycommented}). In all these specifications, the estimates remain similar to the main results. When looking at each of the topic cartels separately, we find that cartel commenters from the fitness \& health cartel are the most similar to natural engagement (\Cref{T:RegrHeterogeneity}).
In our main analysis, the quality of the additional engagement that cartels bring is proxied by the similarity between the author and the cartel commenter. This is good proxy as for example, the topic analysis shows (\Cref{F:LDA_CartelVsNatural}) that cartel members natural commenters are interested in the same topic
as the cartel member irrespective of whether it is a topic or general cartel. However, as
a robustness check, in
\Cref{T:RegrRobust_cartel_commenters_noncartelcommenter,T:RegrRobust_cartel_commenters_noncartelcommenter_largeboth,T:RegrRobust_cartel_commenters_noncartelcommenter_heter},
the outcome variable is the similarity of the author and the commenting cartel member's non-cartel commenter. Adding this one additional layer of distance increases noise, and therefore, similarity measures are smaller, but results remain qualitatively similar.
\begin{table}[h!]
    \begin{center}
        \begin{footnotesize}
            \caption{Alternative outcome variables: random posts; all posts in 2017–2020}
            \apptable{Alternative outcome variables: random posts; all posts in 2017–2020}
            \label{T:RegrRobust_labse_random_monthly}
            \begin{tabular}{lcccccc}
                \hline
 & (1) & (2) & (3) & (4) & (5) & (6) \\  & \multicolumn{6}{c}{Dependent variable: Cosine similarity from text embeddings} \\ \hline
 & \multicolumn{6}{c}{Posts in general or topic cartels} \\ & General & Topic & Both & General & Topic & Both \\ \hline &  \multicolumn{3}{c}{Random 100 posts} & \multicolumn{3}{c}{All posts 2017-2020 } \\ \cmidrule(lr){2-4}\cmidrule(lr){5-7}
General cartel commenter&      -0.044***&               &      -0.044***&      -0.028***&               &      -0.038***\\
                    &     (0.002)   &               &     (0.007)   &     (0.002)   &               &     (0.007)   \\
Topic cartel commenter&               &      -0.012***&       0.006   &               &      -0.000   &       0.010   \\
                    &               &     (0.003)   &     (0.007)   &               &     (0.003)   &     (0.007)   \\
Random user         &      -0.073***&      -0.075***&      -0.062***&      -0.066***&      -0.068***&      -0.062***\\
                    &     (0.002)   &     (0.003)   &     (0.007)   &     (0.002)   &     (0.003)   &     (0.007)   \\
Wald test, $\beta_{Gen} = \beta_{Top}$, p-value&               &               &       0.000   &               &               &       0.000   \\
Base (non-cartel) mean&       0.626   &       0.630   &       0.623   &       0.651   &       0.656   &       0.647   \\
Author fixed effects&         Yes   &         Yes   &         Yes   &         Yes   &         Yes   &         Yes   \\
Authors             &        4749   &        3261   &         487   &        4750   &        3261   &         488   \\
Observations        &       44752   &       30459   &        6630   &       43650   &       29745   &        6546   \\
\hline

            \end{tabular}
        \end{footnotesize}
    \end{center}
    \footnotesize{Notes:
        Each column presents estimates from a separate panel data fixed effects regression. Unit of observation is an author and another user pair. Outcome variable is the cosine similarity of the author to his commenter or to a random user.
        In columns 1--3, the outcome variable is calculated based on the hashtags in user's 100 randomly chosen posts; in columns 4--6, it is calculated based on hashtags in all posts in 2017--2020 (for further details see \Cref{S:LaBSE}). Each regression includes author fixed effects (equivalent to the post fixed effects because only one post per author). In all the regressions, the base category is the author's similarity to a non-cartel commenter; and \textit{Base (non-cartel) mean} presents their average cosine similarity. \textit{General cartel commenter} is an indicator variable whether the commenter to whom the author's cosine similarity is calculated, is in the general cartel, and \textit{Topic cartel commenter} whether he is in the topic cartel. \textit{Random user} indicates that the author's similarity is calculated to a counterfactual random non-cartel user. Standard errors in parentheses are clustered at the author level.
    }
\end{table}
\setlength{\tabcolsep}{10pt}
\begin{table}[h!]
    \begin{center}
        \begin{footnotesize}
            \caption{Alternative outcome variable: whole text instead of hashtags}
            \apptable{Alternative outcome variable: whole text instead of hashtags}
            \label{T:RegrRobust_labse_text}
            \begin{tabular}{lccc}
                \hline
 & (1) & (2) & (3) \\  \multicolumn{4}{c}{Dependent variable: Cosine similarity from text embeddings of full text} \\ \hline
 & \multicolumn{3}{c}{Posts in general or topic cartels} \\ & General & Topic & Both  \\ \hline
General cartel commenter&      -0.043***&               &      -0.042***\\
                    &     (0.002)   &               &     (0.006)   \\
Topic cartel commenter&               &      -0.014***&      -0.003   \\
                    &               &     (0.003)   &     (0.006)   \\
Random user         &      -0.073***&      -0.079***&      -0.063***\\
                    &     (0.002)   &     (0.003)   &     (0.006)   \\
Wald test, $\beta_{Gen} = \beta_{Top}$, p-value&               &               &       0.000   \\
Base (non-cartel) mean&       0.659   &       0.665   &       0.662   \\
Author fixed effects&         Yes   &         Yes   &         Yes   \\
Authors             &        4756   &        3263   &         488   \\
Observations        &       44900   &       30569   &        6665   \\
\hline

            \end{tabular}
        \end{footnotesize}
    \end{center}
    \footnotesize{Notes:
        Each column presents estimates from a separate panel data fixed effects regression. Unit of observation is an author and another user pair. Outcome variable is the cosine similarity of the author to his commenter or to a random user. It is calculated using the whole text instead of only hashtags (for further details see \Cref{S:LaBSE}). Each regression includes author fixed effects (equivalent to the post fixed effects because only one post per author). In all the regressions, the base category is the author's similarity to a non-cartel commenter; and \textit{Base (non-cartel) mean} presents their average cosine similarity. \textit{General cartel commenter} is an indicator variable whether the commenter to whom the author's cosine similarity is calculated, is in the general cartel, and \textit{Topic cartel commenter} whether he is in the topic cartel. \textit{Random user} indicates that the author's similarity is calculated to a counterfactual random non-cartel user. Standard errors in parentheses are clustered at the author level.
    }
\end{table}
\setlength{\tabcolsep}{6pt}
\begin{table}[h!]
    \begin{center}
        \begin{footnotesize}
            \caption{Alternative outcome variables: photo embeddings, text embeddings}
            \apptable{Alternative outcome variables: photo embeddings, text embeddings}
            \label{T:RegrRobust_clip_photos_text}
            \begin{tabular}{lcccccc}
                \hline
 & (1) & (2) & (3) & (4) & (5) & (6) \\  & \multicolumn{6}{c}{Dependent variable: Cosine similarity of posts} \\ \hline
 & \multicolumn{6}{c}{Posts in general or topic cartels} \\ & General & Topic & Both & General & Topic & Both \\ \hline &  \multicolumn{3}{c}{Photos embeddings} & \multicolumn{3}{c}{Text embeddings} \\ \cmidrule(lr){2-4}\cmidrule(lr){5-7}
General cartel commenter&      -0.033***&               &      -0.037***&      -0.019***&               &      -0.021***\\
                    &     (0.002)   &               &     (0.005)   &     (0.001)   &               &     (0.004)   \\
Topic cartel commenter&               &      -0.008***&      -0.004   &               &      -0.009***&      -0.008** \\
                    &               &     (0.002)   &     (0.005)   &               &     (0.002)   &     (0.004)   \\
Random user         &      -0.051***&      -0.053***&      -0.044***&      -0.016***&      -0.018***&      -0.022***\\
                    &     (0.002)   &     (0.002)   &     (0.005)   &     (0.001)   &     (0.002)   &     (0.004)   \\
Wald test, $\beta_{Gen} = \beta_{Top}$, p-value&               &               &       0.000   &               &               &       0.000   \\
Base (non-cartel) mean&       0.487   &       0.493   &       0.491   &       0.777   &       0.780   &       0.784   \\
Author fixed effects&         Yes   &         Yes   &         Yes   &         Yes   &         Yes   &         Yes   \\
Authors             &        4756   &        3263   &         488   &        4756   &        3263   &         488   \\
Observations        &       44900   &       30569   &        6665   &       44900   &       30569   &        6665   \\
\hline

            \end{tabular}
        \end{footnotesize}
    \end{center}
    \footnotesize{Notes:
        Each column presents estimates from a separate panel data fixed effects regression. Unit of observation is an author and another user pair. Outcome variable is the cosine similarity of the author to his commenter or to a random user. It is calculated using embeddings from the CLIP model of either only photos (columns 1--3), or only text (columns 4--6). Each regression includes author fixed effects (equivalent to the post fixed effects because only one post per author). In all the regressions, the base category is the author's similarity to a non-cartel commenter; and \textit{Base (non-cartel) mean} presents their average cosine similarity. \textit{General cartel commenter} is an indicator variable whether the commenter to whom the author's cosine similarity is calculated, is in the general cartel, and \textit{Topic cartel commenter} whether he is in the topic cartel. \textit{Random user} indicates that the author's similarity is calculated to a counterfactual random non-cartel user. Standard errors in parentheses are clustered at the author level.
    }
\end{table}
\begin{table}[h!]
    \begin{center}
        \begin{footnotesize}
            \caption{Alternative outcome variables: indicator for cosine similarity above the 75th or 90th percentile}
            \apptable{Alternative outcome variables: similarity above 75th, 90th percentile}
            \label{T:RegrRobust_similarity_above_threshold}
            \begin{tabular}{lcccccc}
                \hline
 & (1) & (2) & (3) & (4) & (5) & (6) \\
 & \multicolumn{6}{c}{Posts in general or topic cartels} \\ & General & Topic & Both & General & Topic & Both \\ \hline &  \multicolumn{3}{c}{Similarity of users} & \multicolumn{3}{c}{Similarity of posts} \\ &  \multicolumn{3}{c}{Text embeddings} & \multicolumn{3}{c}{Photo+text embeddings} \\ \cmidrule(lr){2-4}\cmidrule(lr){5-7} & \multicolumn{6}{c}{Panel A: Dep. var: Cosine similarity above the 75th percentile} \\
General cartel commenter&      -0.168***&               &      -0.151***&      -0.157***&               &      -0.154***\\
                    &     (0.008)   &               &     (0.021)   &     (0.007)   &               &     (0.019)   \\
Topic cartel commenter&               &      -0.062***&       0.003   &               &      -0.087***&      -0.062***\\
                    &               &     (0.009)   &     (0.021)   &               &     (0.009)   &     (0.020)   \\
Random user         &      -0.201***&      -0.214***&      -0.156***&      -0.180***&      -0.185***&      -0.177***\\
                    &     (0.007)   &     (0.009)   &     (0.021)   &     (0.007)   &     (0.008)   &     (0.019)   \\
Wald test, $\beta_{Gen} = \beta_{Top}$, p-value&               &               &       0.000   &               &               &       0.000   \\
Base (non-cartel) mean&       0.393   &       0.419   &       0.383   &       0.391   &       0.396   &       0.389   \\
Author fixed effects&         Yes   &         Yes   &         Yes   &         Yes   &         Yes   &         Yes   \\
Authors             &        4756   &        3263   &         488   &        4756   &        3263   &         488   \\
Observations        &       44900   &       30569   &        6665   &       44900   &       30569   &        6665   \\
\cmidrule(lr){2-7}
 & \multicolumn{6}{c}{Panel B: Dep. var: Cosine similarity above the 90th percentile} \\
General cartel commenter&      -0.136***&               &      -0.146***&      -0.131***&               &      -0.124***\\
                    &     (0.006)   &               &     (0.018)   &     (0.006)   &               &     (0.017)   \\
Topic cartel commenter&               &      -0.071***&      -0.058***&               &      -0.061***&      -0.059***\\
                    &               &     (0.008)   &     (0.018)   &               &     (0.007)   &     (0.017)   \\
Random user         &      -0.158***&      -0.176***&      -0.149***&      -0.142***&      -0.122***&      -0.126***\\
                    &     (0.006)   &     (0.008)   &     (0.018)   &     (0.006)   &     (0.007)   &     (0.017)   \\
Wald test, $\beta_{Gen} = \beta_{Top}$, p-value&               &               &       0.000   &               &               &       0.000   \\
Base (non-cartel) mean&       0.216   &       0.243   &       0.236   &       0.214   &       0.198   &       0.223   \\
Author fixed effects&         Yes   &         Yes   &         Yes   &         Yes   &         Yes   &         Yes   \\
Authors             &        4756   &        3263   &         488   &        4756   &        3263   &         488   \\
Observations        &       44900   &       30569   &        6665   &       44900   &       30569   &        6665   \\
\hline

            \end{tabular}
        \end{footnotesize}
    \end{center}
    \footnotesize{Notes:
        Each column presents estimates from a separate panel data fixed effects regression. Unit of observation is an author and another user pair. Outcome variable is an indicator whether cosine similarity of the author to his commenter or to a random user is above the 75th or the 90th percentile. The percentiles are calculated using all author and user pairs (cartel, non-cartel, and random users) and all cartels. In columns 1--3, the cosine similarity of users is calculated using the text embeddings from the LaBSE model; in columns 4--6, the cosine similarity of the corresponding users' posts is calculated using the photo and text embeddings from the CLIP model. Each regression includes author fixed effects (equivalent to the post fixed effects because only one post per author). In all the regressions, the base category is the author's similarity to a non-cartel commenter; and \textit{Base (non-cartel) mean} presents their average cosine similarity. \textit{General cartel commenter} is an indicator variable whether the commenter to whom the author's cosine similarity is calculated, is in the general cartel, and \textit{Topic cartel commenter} whether he is in the topic cartel. \textit{Random user} indicates that the author's similarity is calculated to a counterfactual random non-cartel user. Standard errors in parentheses are clustered at the author level.
    }
\end{table}
\begin{table}[h!]
    \begin{center}
        \begin{footnotesize}
            \caption{Alternative sample: commenters who actually commented}
            \apptable{Alternative sample: commenters who actually commented}
            \label{T:RegrRobust_actuallycommented}
            \begin{tabular}{lcccccc}
                \hline
 & (1) & (2) & (3) & (4) & (5) & (6) \\  & \multicolumn{6}{c}{Dependent variable: Cosine similarity} \\ \hline
 & \multicolumn{6}{c}{Posts in general or topic cartels} \\ & General & Topic & Both & General & Topic & Both \\ \hline &  \multicolumn{3}{c}{Similarity of users} & \multicolumn{3}{c}{Similarity of posts} \\ &  \multicolumn{3}{c}{Text embeddings} & \multicolumn{3}{c}{Photo+text embeddings} \\ \cmidrule(lr){2-4}\cmidrule(lr){5-7}
General cartel commenter&      -0.055***&               &      -0.062***&      -0.033***&               &      -0.032***\\
                    &     (0.003)   &               &     (0.011)   &     (0.002)   &               &     (0.005)   \\
Topic cartel commenter&               &      -0.020***&       0.005   &               &      -0.014***&      -0.011** \\
                    &               &     (0.004)   &     (0.011)   &               &     (0.002)   &     (0.005)   \\
Random user         &      -0.067***&      -0.074***&      -0.056***&      -0.037***&      -0.040***&      -0.035***\\
                    &     (0.003)   &     (0.003)   &     (0.010)   &     (0.001)   &     (0.002)   &     (0.004)   \\
Wald test, $\beta_{Gen} = \beta_{Top}$, p-value&               &               &       0.000   &               &               &       0.000   \\
Base (non-cartel) mean&       0.566   &       0.581   &       0.573   &       0.652   &       0.655   &       0.656   \\
Author fixed effects&         Yes   &         Yes   &         Yes   &         Yes   &         Yes   &         Yes   \\
Authors             &        3361   &        2635   &         281   &        3361   &        2635   &         281   \\
Observations        &       27379   &       21797   &        2961   &       27379   &       21797   &        2961   \\
\hline

            \end{tabular}
        \end{footnotesize}
    \end{center}
    \footnotesize{Notes:
        The table presents estimates from the same regressions as in \Cref{T:RegrMain} except the sample includes only these cartel commenters who actually commented. Each column presents estimates from a separate panel data fixed effects regression. Unit of observation is an author and another user pair. Outcome variable is the cosine similarity of the author to his commenter or to a random user. Each regression includes author fixed effects (equivalent to the post fixed effects because only one post per author). In all the regressions, the base category is the author's similarity to a non-cartel commenter; and \textit{Base (non-cartel) mean} presents their average cosine similarity. \textit{General cartel commenter} is an indicator variable whether the commenter to whom the author's cosine similarity is calculated, is in the general cartel, and \textit{Topic cartel commenter} whether he is in the topic cartel. \textit{Random user} indicates that the author's similarity is calculated to a counterfactual random non-cartel user. Standard errors in parentheses are clustered at the author level.
    }
\end{table}
\begin{table}[h!]
    \begin{center}
        \begin{footnotesize}
            \caption{Heterogeneity by topic cartel}
            \apptable{Heterogeneity by topic cartel}
            \label{T:RegrHeterogeneity}
            \hspace*{-0.5cm}
            \begin{tabular}{lcccccc}
                \hline
 & \multicolumn{6}{c}{Posts in topic cartels, by the topic} \\ & Fashion  & Fitness & Travel & Fashion & Fitness & Travel \\ & \& beauty & \& health & \& food & \& beauty & \& health & \& food \\ \hline &  \multicolumn{3}{c}{Similarity of users} & \multicolumn{3}{c}{Similarity of posts} \\ &  \multicolumn{3}{c}{Text embeddings} & \multicolumn{3}{c}{Text+photo embeddings} \\ \cmidrule(lr){2-4}\cmidrule(lr){5-7}
Topic cartel commenter&      -0.014** &       0.005   &      -0.033***&      -0.011***&      -0.006   &      -0.018***\\
                    &     (0.007)   &     (0.008)   &     (0.004)   &     (0.003)   &     (0.004)   &     (0.002)   \\
Random user         &      -0.059***&      -0.081***&      -0.081***&      -0.038***&      -0.050***&      -0.038***\\
                    &     (0.006)   &     (0.008)   &     (0.004)   &     (0.003)   &     (0.004)   &     (0.002)   \\
Base (non-cartel) mean&       0.572   &       0.569   &       0.596   &       0.662   &       0.666   &       0.652   \\
Author fixed effects&         Yes   &         Yes   &         Yes   &         Yes   &         Yes   &         Yes   \\
Authors             &         728   &         573   &        1881   &         728   &         573   &        1881   \\
Observations        &        6386   &        4953   &       18184   &        6386   &        4953   &       18184   \\
\hline

            \end{tabular}
        \end{footnotesize}
    \end{center}
    \footnotesize{Notes:
        The table presents estimates from the same regressions as in \Cref{T:RegrMain} except the sample is a subset of the sample in columns 2 and 5 of \Cref{T:RegrMain}. Specifically, the sample consists of authors whose first cartel post is either: only in Fashion and beauty cartel (columns 1 and 4), only in Fitness and health cartel (columns 2 and 5), or only in Travel and food cartel (columns 3 and 6). Note that those authors whose post is in multiple cartels are excluded. Each column presents estimates from a separate panel data fixed effects regression. Unit of observation is an author and another user pair. Outcome variable is the cosine similarity of the author to his commenter or to a random user. In columns 1--3, the cosine similarity of users is calculated using the text embeddings from the LaBSE model; in columns 4--6, the cosine similarity of the corresponding users’ posts is calculated using the photo and text embeddings from the CLIP model. Each regression includes author fixed effects (equivalent to the post fixed effects because only one post per author). In all the regressions, the base category is the author's similarity to a non-cartel commenter; and \textit{Base (non-cartel) mean} presents their average cosine similarity. \textit{Topic cartel commenter} is an indicator variable whether the commenter to whom the author's cosine similarity is calculated, is in the topic cartel. \textit{Random user} indicates that the author's similarity is calculated to a counterfactual random non-cartel user. Standard errors in parentheses are clustered at the author level.
    }
\end{table}
\begin{table}[h!]
    \begin{center}
        \begin{footnotesize}
            \caption{Cartel commenters' non-cartel commenters}
            \apptable{Cartel commenters' non-cartel commenters}
            \label{T:RegrRobust_cartel_commenters_noncartelcommenter}
            \hspace*{-0.5cm}
            \begin{tabular}{lcccccc}
                \hline
 & (1) & (2) & (3) & (4) & (5) & (6) \\  & \multicolumn{6}{c}{Dependent variable: Cosine similarity} \\ \hline
 & \multicolumn{6}{c}{Posts in general or topic cartels} \\ & General & Topic & Both & General & Topic & Both \\ \hline &  \multicolumn{3}{c}{Similarity of users} & \multicolumn{3}{c}{Similarity of posts} \\ &  \multicolumn{3}{c}{Text embeddings} & \multicolumn{3}{c}{Photo+text embeddings} \\ \cmidrule(lr){2-4}\cmidrule(lr){5-7}
General cartel commenter (CC)&      -0.058***&               &      -0.051***&      -0.033***&               &      -0.029***\\
                    &     (0.003)   &               &     (0.008)   &     (0.001)   &               &     (0.003)   \\
Topic cartel commenter (CC)&               &      -0.023***&      -0.005   &               &      -0.016***&      -0.008** \\
                    &               &     (0.003)   &     (0.008)   &               &     (0.002)   &     (0.004)   \\
$\alpha_{Gen}$: General CC's non-cartel commenter&      -0.061***&               &      -0.054***&      -0.034***&               &      -0.024***\\
                    &     (0.003)   &               &     (0.009)   &     (0.001)   &               &     (0.004)   \\
$\alpha_{Top}$: Topic CC's non-cartel commenter&               &      -0.054***&      -0.024***&               &      -0.028***&      -0.022***\\
                    &               &     (0.004)   &     (0.009)   &               &     (0.002)   &     (0.004)   \\
$\beta_{Ran}$: Random user&      -0.070***&      -0.075***&      -0.055***&      -0.039***&      -0.041***&      -0.031***\\
                    &     (0.003)   &     (0.003)   &     (0.009)   &     (0.001)   &     (0.002)   &     (0.003)   \\
Wald test, $\alpha_{Gen} = \beta_{Ran}$&       0.000   &               &       0.772   &       0.000   &               &       0.006   \\
Wald test, $\alpha_{Top} = \beta_{Ran}$&               &       0.000   &       0.000   &               &       0.000   &       0.001   \\
Wald test, $\alpha_{Gen} = \alpha_{Top}$&               &               &       0.000   &               &               &       0.562   \\
Base (non-cartel) mean&       0.574   &       0.585   &       0.573   &       0.656   &       0.657   &       0.659   \\
Author fixed effects&         Yes   &         Yes   &         Yes   &         Yes   &         Yes   &         Yes   \\
Authors             &        4236   &        2882   &         392   &        4236   &        2882   &         392   \\
Observations        &       50660   &       34002   &        7573   &       50660   &       34002   &        7573   \\
\hline
\% value via general cartel&        13.8&            &         2.9&        15.0&            &        23.1\\
\% value via topic cartel&            &        27.9&        56.3&            &        30.9&        29.2\\
\hline

            \end{tabular}
        \end{footnotesize}
    \end{center}
    \footnotesize{Notes:
        The table presents estimates from the same regressions as in \Cref{T:RegrMain} (\Cref{E:MainRegression}) where in addition to indicators for the similarity between the author and the general cartel commenter (\textit{General cartel commenter (CC)}), topic cartel commenter (\textit{Topic cartel commenter (CC)}) and random user (\textit{Random user}), the regression also includes an indicator for the similarity between the author and the general cartel commenters' non-cartel commenter (\textit{General CC’s non-cartel commenter}) and the topic cartel commenters' non-cartel commenter (\textit{Topic CC’s non-cartel commenter}). The last two rows present the similarity between the author and cartel commenters' non-cartel commenter rescaled to the natural-random difference, calculated as: $100\times(1-\alpha_{t}/\beta_{Ran}), \;t \in\{Gen, Top\}$. This proxies the percentage of value of natural engagement that advertisers get via cartels. Standard errors in parentheses are clustered at the author level.
    }
\end{table}
\begin{table}[h!]
    \begin{center}
        \begin{footnotesize}
            \caption{Cartel commenters' non-cartel commenters: posts in both cartels}
            \apptable{Cartel commenters' non-cartel commenters: posts in both cartels}
            \label{T:RegrRobust_cartel_commenters_noncartelcommenter_largeboth}
            \hspace*{-0.5cm}
            \begin{tabular}{lcc}
                \hline
& (1) & (2) \\  & \multicolumn{2}{c}{Dependent variable: Cosine similarity} \\ & \multicolumn{2}{c}{Posts in both cartels} \\ \hline &  Similarity of users & Similarity of posts \\ &  Text embeddings & Text+photo embeddings \\ \cmidrule(lr){2-3}
Topic cartel commenter's non-cartel commenter&       0.025***&       0.005*  \\
                    &     (0.005)   &     (0.002)   \\
Base (general cartel commenter's non-cartel commenter) mean&       0.515   &       0.624   \\
Author fixed effects&         Yes   &         Yes   \\
Authors             &         722   &         722   \\
Observations        &        3730   &        3730   \\
\hline

            \end{tabular}
        \end{footnotesize}
    \end{center}
    \footnotesize{Notes:
        The table presents estimates from similar regressions as in columns 3 and 6 in \Cref{T:RegrRobust_cartel_commenters_noncartelcommenter}. The outcome variable is the similarity between the author and the cartel commenters' non-cartel commenter. The regression includes an indicator for the similarity between the author and the topic cartel commenters' non-cartel commenter. The base category is the similarity between the author and the general cartel commenters' non-cartel commenter. The sample is larger than in \Cref{T:RegrRobust_cartel_commenters_noncartelcommenter} because it includes also the authors for whom we were not able to calculate the similarity with non-cartel commenters which is the base category in \Cref{T:RegrRobust_cartel_commenters_noncartelcommenter}. Standard errors in parentheses are clustered at the author level.
    }
\end{table}
\begin{table}[h!]
    \begin{center}
        \begin{footnotesize}
            \caption{Cartel commenters' non-cartel commenters: heterogeneity by cartel topic}
            \apptable{Cartel commenters' non-cartel commenters: heterogeneity}
            \label{T:RegrRobust_cartel_commenters_noncartelcommenter_heter}
            \hspace*{-0.5cm}
            \begin{tabular}{lcccccc}
                \hline
 & \multicolumn{6}{c}{Posts in topic cartels, by the topic} \\ & Fashion  & Fitness & Travel & Fashion & Fitness & Travel \\ & \& beauty & \& health & \& food & \& beauty & \& health & \& food \\ \hline &  \multicolumn{3}{c}{Similarity of users} & \multicolumn{3}{c}{Similarity of posts} \\ &  \multicolumn{3}{c}{Text embeddings} & \multicolumn{3}{c}{Text+photo embeddings} \\ \cmidrule(lr){2-4}\cmidrule(lr){5-7}
Topic cartel commenter (CC)&      -0.015** &       0.009   &      -0.032***&      -0.012***&      -0.006   &      -0.018***\\
                    &     (0.007)   &     (0.009)   &     (0.004)   &     (0.003)   &     (0.004)   &     (0.002)   \\
$\alpha_{Top}$: Topic CC's non-cartel commenter&      -0.048***&      -0.039***&      -0.061***&      -0.024***&      -0.034***&      -0.028***\\
                    &     (0.007)   &     (0.010)   &     (0.005)   &     (0.004)   &     (0.005)   &     (0.002)   \\
$\beta_{Ran}$: Random user&      -0.061***&      -0.077***&      -0.080***&      -0.040***&      -0.051***&      -0.038***\\
                    &     (0.007)   &     (0.009)   &     (0.004)   &     (0.003)   &     (0.004)   &     (0.002)   \\
Wald test, $\alpha_{Top} = \beta_{Ran}$&       0.008   &       0.000   &       0.000   &       0.000   &       0.000   &       0.000   \\
Base (non-cartel) mean&       0.576   &       0.561   &       0.596   &       0.664   &       0.666   &       0.652   \\
Author fixed effects&         Yes   &         Yes   &         Yes   &         Yes   &         Yes   &         Yes   \\
Authors             &         596   &         443   &        1764   &         596   &         443   &        1764   \\
Observations        &        6443   &        4673   &       21547   &        6443   &        4673   &       21547   \\
\hline
\% value via topic cartel&        21.8&        49.4&        23.6&        41.1&        33.4&        27.5\\
\hline

            \end{tabular}
        \end{footnotesize}
    \end{center}
    \footnotesize{Notes:
        The table presents estimates from the same regressions as in \Cref{T:RegrHeterogeneity} where in addition to indicators for the similarity between the author and the topic cartel commenter (\textit{Topic cartel commenter (CC)}) and random user (\textit{Random user}), the regression also includes an indicator for the similarity between the author and the topic cartel commenters' non-cartel commenter (\textit{Topic CC’s non-cartel commenter}). The last row presents the similarity between the author and cartel commenters' non-cartel commenter rescaled to the natural-random difference ($100\times(1 -\alpha_{Top}/\beta_{Ran})$). This proxies the percentage of the value of natural engagement that advertisers get via cartels. Standard errors in parentheses are clustered at the author level.
    }
\end{table}

\end{document}